\def\draft{0}
\def\doubleblind{0}
\documentclass[11pt,letterpaper]{article}

\usepackage{mathpazo}
\usepackage{microtype}

\usepackage[margin=1in]{geometry}
\usepackage[T1]{fontenc}
\usepackage{bbm, bm}
\usepackage{graphicx}
\usepackage{amsmath, amsfonts, amssymb, amsthm, thmtools, mathtools, stmaryrd}
\usepackage{algorithm, algpseudocode}
\usepackage{pifont}
\usepackage{xspace}

\PassOptionsToPackage{hyphens}{url}

\usepackage{color}
\usepackage[dvipsnames]{xcolor} 
\usepackage[pdfstartview=FitH,pdfpagemode=UseNone,colorlinks,linkcolor=NavyBlue,filecolor=blue,citecolor=OliveGreen,urlcolor=NavyBlue]{hyperref}
\usepackage[capitalise,nameinlink]{cleveref}
\usepackage{framed}

\usepackage[style=alphabetic, backend=biber, minalphanames=3, maxalphanames=4, maxbibnames=99, maxcitenames=99]{biblatex}

\crefformat{section}{#2\S#1#3}
\crefformat{subsection}{#2\S#1#3}
\crefformat{subsubsection}{#2\S#1#3}
\Crefformat{section}{#2\S#1#3}
\Crefformat{subsection}{#2\S#1#3}
\Crefformat{subsubsection}{#2\S#1#3}

\declaretheoremstyle[bodyfont=\it,qed=\qedsymbol]{noproofstyle}

\numberwithin{equation}{section}

\newcommand{\excl}{(\ding{72})}

\declaretheorem[numberlike=equation]{observation}

\declaretheorem[name=Observation,numbered=no]{observation*}

\declaretheorem[numberlike=equation]{theorem}

\declaretheorem[name=Theorem,numbered=no]{theorem*}

\declaretheorem[numberlike=equation]{lemma}
\declaretheorem[name=Lemma,numbered=no]{lemma*}

\declaretheorem[numberlike=equation]{corollary}
\declaretheorem[name=Corollary,numbered=no]{corollary*}

\declaretheorem[numberlike=equation]{proposition}
\declaretheorem[name=Proposition,numbered=no]{proposition*}

\declaretheorem[numberlike=equation, refname={claim,claims}]{claim}
\declaretheorem[name=Claim,numbered=no]{claim*}

\declaretheorem[name=Conjecture,numbered=no]{conjecture*}

\declaretheorem[name=Question,numbered=no]{question*}

\declaretheoremstyle[headfont=\normalfont\bfseries,bodyfont=\normalfont,postheadspace=1em,qed=$\lozenge$]{rmkstyle}

\declaretheorem[style=rmkstyle,numberlike=equation]{remark}

\declaretheorem[numberlike=equation,style=rmkstyle]{definition}
\declaretheorem[unnumbered,name=Definition,style=rmkstyle]{definition*}

\declaretheorem[numberlike=equation,style=rmkstyle]{example}
\declaretheorem[unnumbered,name=Example,style=rmkstyle]{example*}

\declaretheorem[unnumbered,name=Notation,style=rmkstyle]{notation*}

\declaretheorem[unnumbered,name=Construction,style=rmkstyle]{construction*}

\newenvironment{fexample}
  {\vspace{-5 pt}\colorlet{shadecolor}{ProcessBlue!08}\begin{snugshade*}\begin{example}}
  {\end{example}\end{snugshade*}\vspace{-5 pt}}

\usepackage{singer-macros}

\newcommand{\Dist}[1]{\Delta(#1)}
\newcommand{\OWU}[1]{\Delta_1(#1)}

\newcommand{\Player}{\mathtt{Player}}

\newcommand{\val}[2]{\mathsf{val}_{#1}(#2)}
\newcommand{\Unif}{\mathsf{Unif}}

\newcommand{\tvd}{\mathsf{tvd}}

\newcommand{\cD}{\mathcal{D}}
\newcommand{\cE}{\mathcal{E}}

\newcommand{\cY}{\mathcal{Y}}
\renewcommand{\supp}{\mathsf{supp}}

\newcommand{\cF}{\mathcal{F}}

\renewcommand{\epsilon}{\varepsilon}
\newcommand{\eps}{\epsilon}

\newcommand{\PHM}[3]{\calM^{#1,#2}_{#3}}
\newcommand{\HG}[3]{\calH^{#1,#2}_{#3}}

\newcommand{\Den}[1]{\mu_{#1}}

\newcommand{\rwt}{\operatorname{rwt}}

\newcommand{\SF}{\mathrm{SF}}
\newcommand{\MaxCSP}[1]{{\rm{\textsc{Max-CSP}}}(#1)}
\newcommand{\ApxMaxCSP}[2]{{\rm{\textsc{Max-CSP}}}_{#1}(#2)}
\newcommand{\GapMaxCSP}[3]{{\rm{\textsc{Max-CSP}}}_{#1,#2}(#3)}
\newcommand{\optCSP}[1]{\mathsf{opt}^{\rm{CSP}}(#1)}
\newcommand{\optLP}[1]{\mathsf{opt}^{\rm{LP}}(#1)}
\newcommand{\alphaLP}[1]{\alpha^{\rm{LP}}_{#1}}

\ifnum\draft=1
\newcommand{\Authornote}[3]{{\sf\small\color{#3}{[#1: #2]}}}
\else
\newcommand{\Authornote}[3]{}
\fi

\title{Optimal Single-Pass Streaming Lower Bounds for Approximating CSPs}

\ifnum\doubleblind=0
\author{Noah G. Singer\thanks{Carnegie Mellon University, Pittsburgh, PA, USA. Email: \url{ngsinger@cs.cmu.edu}.}
\and Madhur Tulsiani\thanks{Toyota Technological Institute, Chicago, IL, USA. Email: \url{madhurt@ttic.edu}. Work supported in part by NSF award CCF-2326685.}
\and Santhoshini Velusamy\thanks{University of Waterloo, ON, Canada. Email: \url{santhoshini.velusamy@uwaterloo.ca}. Work supported in part by NSF award CCF 2348475 when the author was affiliated with Toyota Technological Institute at Chicago.}
}
\fi
\ifnum\doubleblind=1
\author{Anonymous authors}
\fi

\begin{document}

\maketitle

\thispagestyle{empty}

\begin{abstract}
    For an arbitrary family of predicates $\mathcal{F} \subseteq \{0,1\}^{[q]^k}$ and any $\epsilon > 0$,
    we prove a single-pass, linear-space streaming lower bound against the gap promise problem of distinguishing instances of $\textsc{Max-CSP}({\mathcal{F}})$ with at most $\beta+\epsilon$ fraction of satisfiable constraints from instances of with at least $\gamma-\epsilon$ fraction of satisfiable constraints, 
    whenever $\textsc{Max-CSP}({\mathcal{F}})$ admits a $(\gamma,\beta)$-integrality gap instance for the basic LP.
    This subsumes the linear-space lower bound of Chou, Golovnev, Sudan, Velingker, and Velusamy (STOC 2022),
    which applies only to a special subclass of CSPs with linear-algebraic structure.
    (Their result itself generalizes work of Kapralov and Krachun (STOC 2019) for $\textsc{Max-Cut}$.)
    Our approach identifies the right ``analytic'' analogues of previously-used linear-algebraic conditions; 
    this yields substantial simplifications while capturing a much larger class of problems.

    Our lower bound is essentially optimal for single-pass streaming, since:
    (1) All CSPs admit $(1-\epsilon)$-approximations in quasilinear space,
    and (2) sublinear-space streaming algorithms can simulate the LP (on bounded-degree instances),
    giving approximation algorithms when integrality gap instances do \emph{not} exist.

The starting point for our lower bound is a reduction from a "distributional implicit hidden
partition'' problem defined by Fei, Minzer, and Wang (STOC 2026) in the context of multi-pass
streaming. Our result is an analogue of theirs in the single-pass setting, where we obtain a much
stronger (and tight) space lower bound.

\paragraph*{Independent work.}Very recently, similar linear-space lower bounds for CSPs were also obtained by an independent work of
Fei, Minzer, and Wang (arXiv, April 2026). 
Their recent result also implies almost-linear space lower bounds for multi-pass streaming
algorithms with $o(\log n)$ passes.
\end{abstract}

\newpage

\pagenumbering{roman}
\tableofcontents
\clearpage

\newpage
\pagenumbering{arabic}
\setcounter{page}{1}

\section{Introduction}

\newcommand{\ie}{i.e.,\xspace}
\newcommand{\eg}{e.g.,\xspace}
\newcommand{\etal}{et al.\xspace}

Constraint Satisfaction Problems (CSPs) are some of the most fundamental and expressive problems in discrete optimization, which have been a tested for algorithmic techniques and have played a foundational role in the theory of approximation algorithms and inapproximability.
Investigation of the complexity of CSPs in various computational models have not only led to a new understanding of the limits of these computational models, but the finite nature of the templates for CSPs have led to surprising classifications and canonical algorithms.
In this work, we study the complexity of approximating CSPs in the model of single-pass streaming algorithms.

\vspace{-10 pt}
\paragraph{Constraint satisfaction problems.}
Fix $q,k \ge 2 \in \N$ and let $\calV$ be a finite set, elements of which we call \emph{variables} (typically, $\calV = [n] = \{1,\ldots,n\}$).
A \emph{constraint} on $\calV$ is a pair $C = (f,e)$, where $f : [q]^k \to \{0,1\}$ is a function called a \emph{predicate}
and $e = (i_1,\ldots,i_k) \in \calV^k$ is a $k$-tuple of \emph{distinct} variables.
($q$ is called the \emph{alphabet size} and $k$ the \emph{arity}.)
An \emph{assignment} to $\calV$ is a vector $x \in [q]^\calV$,
and an assignment $x$ \emph{satisfies} a constraint $C = (f,e=(i_1,\ldots,i_k))$ if $f(x_{i_1},\ldots,x_{i_k}) = 1$.

A \emph{constraint satisfaction problem (CSP)} is specified by a \emph{predicate family} $\calF \subseteq \{0,1\}^{[q]^k}$,
\ie a set of functions from $[q]^k \to \{0,1\}$.
An \emph{instance} $\Phi$ of $\MaxCSP{\calF}$ on a set of variables $\calV$ is list of constraints $(f,e)$ on $\calV$ with $f \in \calF$.
The \emph{value} of an assignment $x \in [q]^\calV$ is $\val{\Phi}{x} \coloneqq \Exp_{\rC \sim \Phi} [x\text{ satisfies }\rC]$, where the distribution over constraints in $\Phi$ is taken to be uniform.
The \emph{optimum value} of an instance $\Phi$ is $\optCSP{\Phi} \coloneqq \max_{x \in [q]^\calV} \val{\Phi}{x}$.
\begin{fexample}
	For $k=2$ and $q = 2$, and the set $[q]$ identified with $\{0,1\}$, consider the CSPs:
    \begin{itemize}
    \item \textsc{Max-Cut}:  defined using the single predicate $f_{\textsc{Cut}}(x_1,x_2) = x_1 \oplus x_2$
    \item \textsc{Max-DiCut}: defined using the predicate $f_{\textsc{DiCut}}(x_1,x_2) = x_1 \wedge \neg x_2$.
    \item \textsc{Max-2AND}: defined using the \emph{predicate family} \\[3 pt]
    $\calF = \{f_{\textsc{2AND}}^{b_1,b_2}(x_1,x_2) = (b_1 \oplus x_1) \wedge (b_2 \oplus x_2) ~|~ b_1,b_2 \in \{0,1\}\}$.
    \end{itemize}

We note that \textsc{Max-2AND} is an example of a ``CSP with literals" \ie a predicate family closed under negations of variables.
\end{fexample}

There are two natural approximation variants of the problem $\MaxCSP{\calF}$.
The first, denoted $\ApxMaxCSP{\alpha}{\calF}$, is to distinguish between the cases (for a given instance $\Phi$ and threshold $v$) between the cases $\optCSP{\Phi} \geq v$ and $\optCSP{\Phi} < \alpha \cdot v$.
The second problem, denoted $\GapMaxCSP{\beta}{\gamma}{\calF}$, a more fine-grained version of the above and requires distinguishing for a given instance $\Phi$ between the cases $\optCSP{\Phi} \ge \gamma$ and $\optCSP{\Phi} \le \beta$.
We will consider randomized algorithms in this work, for which it suffices to solve the above tasks with probability (say) $2/3$.

\vspace{-10 pt}
\paragraph{Streaming algorithms.}
We consider \emph{single-pass streaming algorithms}
for solving the problems $\ApxMaxCSP{\alpha}{\calF}$ and $\GapMaxCSP{\beta}{\gamma}{\calF}$.
A \emph{single-pass streaming algorithm} for $\ApxMaxCSP{\alpha}{\calF}$ or $\GapMaxCSP{\beta}{\gamma}{\calF}$
is given as input an ordered list (stream) of (not necessarily distinct) constraints $(f,e)$;
this list defines an associated instance of $\MaxCSP{\calF}$, with the associated constraint distribution being uniform on the list.
The stream is presented to the algorithm one constraint at a time,
and the algorithm can only use a bounded amount of memory space to store its current state between each successive pair of constraints.

Formally, \emph{space-$s$ single-pass streaming algorithm} on instances of $\MaxCSP{\calF}$ with $n$ variables and $m$ constraints may be specified by a function $\Gamma: \calC_n \times \{0,1\}^s \to \{0,1\}^s$, where $\calC_n$ denotes the set of all possible constraints on $n$ variables.
The initial state is taken to be $S_0 \gets 0^s \in \{0,1\}^s$, and $j$-th constraint leads to the update $S_j \gets \Gamma(C_j, S_{j-1})$.
We take (say, the first bit of) the final state $S_m$ as the output for the relevant distinguishing problem.
Note that the algorithm is allowed to depend on $n$ and (for the purpose of lower bounds) we also allow it to depend on $m$.
There is no constraint on the time-complexity or randomness used by the algorithm.

\subsection{Prior work on streaming approximability of CSPs}
Streaming approximability of CSPs has been extensively studied over the past decade; we recall some of the most relevant results below.
We refer the reader to the excellent surveys by \textcite{Sud22} and \textcite{Ass23},
and the column~\cite{Sin25} for more detailed accounts.

\vspace{-10 pt}
\paragraph{Approximability in $O(\sqrt{n})$ space.} The first lower bounds against streaming algorithms for CSPs were studied by \textcite{KK15}, and \textcite{KKS15}.
In particular, \cite{KKS15} showed that the $\textsc{Max-Cut}_{1,1/2+\epsilon}$ problem requires $\Omega_\epsilon(\sqrt{n})$ space for every $\epsilon > 0$ (while a random assignment trivially achieves an approximation ratio $1/2$).
Follow-up works~\cite{CGV20,GT19} proved optimal $\Omega(\sqrt{n})$-space lower bounds for $\textsc{Max-Unique-Games}$ and $\textsc{Max-DiCut}$.
We highlight that, unlike $\textsc{Max-Cut}$ and $\textsc{Max-Unique-Games}$,
the $\textsc{Max-DiCut}$ problem turns out to admit nontrivial streaming approximations in $\log n \ll \sqrt n$ space~\cite{GVV17}.
This line of work culminated in the work of~\textcite{CGSV24},
who proved a \emph{dichotomy theorem} stating that for every predicate family $\calF$ and $\gamma > \beta$,
either $\GapMaxCSP{\gamma+\epsilon}{\beta-\epsilon}{\calF}$ admits a $\polylog_\epsilon(n)$-space algorithm for every $\epsilon > 0$ (using the ``bias" of the variables in the CSP instance)
or $\GapMaxCSP{\gamma-\epsilon}{\beta+\epsilon}{\calF}$ requires $\Omega_\epsilon(\sqrt n)$ space for every $\epsilon > 0$.\footnote{
    The lower bound in~\cite{CGSV24} is currently known only to hold against ``sketching'' algorithms, which are a special subclass of streaming algorithms.
    Whether this distinction is only a technical artifact, or whether there are streaming algorithms outperforming sketching algorithms, is an interesting open problem.
}

\vspace{-10 pt}
\paragraph{Approximability in sublinear space.} Beyond $\sqrt n$ space, the picture becomes more complex.
On the algorithmic side, a series of works~\cite{SSSV23-random-ordering,SSSV23-dicut,SSSV25,ABFS26}
established that \emph{better} approximations are attainable for the $\textsc{Max-DiCut}$ problem in $o(n)$ space (in comparison to the $o(\sqrt n)$-space regime).
This opened up the exciting possibility that the right approximation thresholds for many other CSPs may be \emph{different} in the (arguably, more natural) setting of sublinear space, compared to the ones obtained for $O(\sqrt{n})$ space above.
On the hardness side, \textcite{KK19} showed (following up on~\cite{KKSV17}) that the $\textsc{Max-Cut}_{1,1/2+\epsilon}$ problem requires $\Omega_\epsilon(n)$ space for every $\epsilon > 0$,
improving the \cite{KKS15} result (which had an $\Omega_\epsilon(\sqrt n)$-space lower bound).
This result was generalized by~\textcite{CGS+22-linear-space}, who gave $\Omega(n)$-space lower bounds against CSPs satisfying a certain (linear) structural condition called ``width''.
While the work~\cite{CGS+22-linear-space} greatly clarified the technical aspects of proving $\Omega(n)$-space lower bounds (especially in the case $k > 2$), their techniques relied heavily on the presence of such linear structure, leaving open the approximability of many CSPs of interest.
Moreover the``width'' condition itself is quite brittle, and some CSPs which are not ``wide'' were still known to admit lower bounds via reduction from wide CSPs (but were not captured by the above result).

\vspace{-10 pt}
\paragraph{Approximability for multi-pass streaming.} While the above results concern single-pass streaming algorithms, another line of work has also studied the \emph{multi-pass} approximability of CSPs~\cite{BDV18,AKSY20, AN21,CKP+23, KPSY23,SSSV23-random-ordering,KPSY23,SSSV25,FMW25,Vel25}.
Recently, a beautiful result of Fei, Minzer and Wang~\cite{FMW26} obtained a dichotomy theorem for multi-pass streaming using a linear programming relaxation called the ``basic LP''.
They proved that for every predicate family $\calF$, if there exists a $(\gamma,\beta)$-LP integrality gap instance (an instance certifying inapproximability via the basic LP) for $\MaxCSP{\calF}$,
then for every $\epsilon > 0$, the $\GapMaxCSP{\gamma-\epsilon}{\beta+\epsilon}{\calF}$ problem requires $p \cdot s \ge \Omega_\epsilon(n^{1/3})$,
where $p$ is the number of passes and $s$ is the amount of space.
Conversely, if no such instance exists, there is an algorithm for $\GapMaxCSP{\gamma+\epsilon}{\beta-\epsilon}{\calF}$ with $p \cdot s \le \polylog_\epsilon(n)$.
While this result completely characterizes polynomial space approximability in the the multi-pass setting, it does not imply bounds more than $\Omega(n^{1/3})$ for the single-pass setting.

\subsection{Results}
We show that the basic LP also characterizes the sublinear-space approximability of (bounded-degree instances) of all CSPs in the single-pass setting.
In particular, we show that basic LP lower bounds for all CSPs can be converted to similar lower bound on approximation achieved by any algorithm in $o(n)$ space.
Before describing the formal statements of the results, we briefly describe the basic LP relaxation for $\MaxCSP{\calF}$.

\vspace{-10 pt}
\paragraph{The basic linear programming relaxation for a CSP.}
Let $\Phi$ be an instance of $\MaxCSP{\calF}$.
The \emph{basic linear programming (LP) relaxation} of $\Phi$ is formulated in its ``distributional'' version in \Cref{fig:basic-lp} below.
The ``variables'' of this LP are distributions $\calY_C$ over $[q]^k$ for every constraint $C \in \supp(\Phi)$.
(Each such distribution can be encoded as $q^k-1$ standard LP variables.)
An important fact about this LP is that it is a \emph{relaxation} of the underlying CSP instance $\Phi$: we always have $1 \ge \optLP{\Phi} \ge \optCSP{\Phi}$, where $\optLP{\Phi}$ is the maximum feasible value for the LP.
\begin{figure}[h]
\hrule
\vline
\begin{minipage}[t]{0.99\linewidth}
\vspace{-5 pt}
{\small
\begin{align*}
    &\mbox{maximize}\quad ~~ &&\Exp_{\rC = (\rf,\re) \sim \Phi} \bracks*{ \Exp_{\ra \sim \calY_\rC} \rf(\ra) }
    \\
&\mbox{such that}\quad \quad
     &&\Pr_{\ra \sim \calY_C}[\ra_\ell = b] = \Pr_{\ra \sim \calY_{C'}}[\ra_{\ell'} = b], &&&\forall C = (f, e), C' = (f', e') \in \Phi, i_\ell = i'_{\ell'}, b \in [q]. \\
    & &&\calY_C \in \Dist{[q]^k} &&& \forall C \in \Phi.
\end{align*}}
\vspace{-10 pt}
\end{minipage}
\hfill\vline
\hrule
\caption{Distributional view of basic LP.}\label{fig:basic-lp}
\end{figure}

A $(\gamma,\beta)$-LP Integrality Gap instance for $\MaxCSP{\calF}$ is an instance $\Phi$ such that $\optCSP{\Phi} \le \beta$ and $\optLP{\Phi} \ge \gamma$. Our main result is the following:

\begin{restatable}[Linear-space lower bounds against single-pass streaming algorithms for CSPs]{theorem}{mainthm}\label{thm:main}
Let $\calF \subseteq \{0,1\}^{[q]^k}$. For every $1 \ge \gamma > \beta \ge 0$, if there exists a $(\gamma,\beta)$-LP integrality gap instance of $\MaxCSP{\calF}$, then for every $\eps > 0$, any single-pass streaming algorithm for $\GapMaxCSP{\gamma-\epsilon}{\beta+\epsilon}{\calF}$ on $n$ variables requires $\Omega_\epsilon(n)$ space.
Moreover, $\GapMaxCSP{1}{\beta+\epsilon}{\calF}$ also requires space $\Omega_{\eps}(n)$ for the case $\gamma =1$.
\end{restatable}
We emphasize that \Cref{thm:main} subsumes all previous known $\Omega_\epsilon(n)$-space streaming lower bounds~\cite{KK19,CGS+22-linear-space},
and is an analogue of (though formally incomparable with) the multi-pass lower bound in \cite{FMW26}.
We also remark that \Cref{thm:main} is essentially an \emph{optimal lower bound} against single-pass streaming algorithms for all CSPs, in the following sense:
\begin{itemize}
\item With respect to space complexity,
    as observed in e.g.~\cite[\S1.1]{CGS+22-linear-space}, algorithms with $\tilde{O}_\epsilon(n)$ space can solve the $\ApxMaxCSP{1-\epsilon}{\calF}$ problem for every predicate family $\calF$ and $\epsilon > 0$.
Hence, one cannot hope to improve the space lower bound in the theorem (for any $\calF$).
\item With respect to approximation factors, when a $(\gamma,\beta)$-integrality gap instance does \emph{not} exist,
    \cite{STV25} observe that for every $\epsilon > 0$ and $D \in \N$, the $\GapMaxCSP{\gamma+\epsilon}{\beta-\epsilon}{\calF}$ problem
    admits an $o_{\epsilon,D}(n)$-space algorithm on instances of maximum degree $D$.
    Removing the bounded-degree assumption remains an interesting open question.\footnote{
        We feel that the maximum degree assumption can likely be removed, as was true for Max-Dicut~\cite{ABFS26}.}
\item    The hard instances we use in the proof of \Cref{thm:main} have \emph{bounded maximum degree},
    and so \Cref{thm:main} already proves a \emph{linear-space dichotomy theorem for bounded-degree streaming CSPs}.
\end{itemize}
Thus, \Cref{thm:main} yields optimal results for problems considered in several previous works on single-pass streaming lower bounds~\cite{KK15,KKS15,KKSV17,GT19,KK19,CGV20,CGS+22-linear-space,CGSV24}.

Translating \Cref{thm:main} to $\ApxMaxCSP{\alpha}{\calF}$-type problems yields the following.
For a predicate family $\calF \subseteq \{0,1\}^{[q]^k}$, let $\alphaLP{\calF}$ be the best possible approximation ratio over all instances $\Phi$ \ie
\[
    \alphaLP{\calF} \coloneqq \inf_{\Phi}
    \frac{\optCSP{\Phi}}{\optLP{\Phi}}. \]
Then \Cref{thm:main} immediately yields the following:
\begin{corollary}
For every $\epsilon > 0$ and every $\calF \subseteq \{0,1\}^{[q]^k}$, any single-pass streaming algorithm for $\ApxMaxCSP{\alphaLP{\calF}+\epsilon}{\calF}$ on $n$ variables requires $\Omega_\epsilon(n)$ space.
\end{corollary}

As a special case of our results, we also get new lower bounds against many specific CSPs.
In addition to previous examples for $k=2$, we consider some examples with higer arity below.
\begin{fexample}\label{ex:higher arity}
    Let $k > 2$ but fix $q=2$ (with $[q]$ identified with $\{0,1\}$).
    \begin{itemize}
    \item \textsc{Max-$k$-XOR}: defined by predicates $f_{\textsc{$k$-XOR}}^b(x_1,\ldots,x_k) = x_1 \oplus \cdots \oplus x_k \oplus b$ for $b \in \{0,1\}$.
    \item $\textsc{Max-LTF}^w$: For $w \in \R^k$, the linear threshold function (LTF) predicates are the $2^k$ functions \[ f_{\textsc{LTF}}^{w,b}(x_1,\ldots,x_k) = \begin{cases} 1 &\text{if } (-1)^{x_1+b_1} \cdot w_1 + \cdots + (-1)^{x_k+b_k} \cdot w_k > 0, \\
    0 &\text{if } (-1)^{x_1+b_1} \cdot w_1 + \cdots + (-1)^{x_k+b_k} \cdot w_k < 0. \end{cases}
\]
The LTF is said to be balanced if $\Exp_x[f_{\textsc{LTF}}^{w,b}] = 1/2$ (for all $b \in \{0,1\}^k$).
    \item \textsc{Max-Exactly-$\ell$-Of-$k$}: defined by the predicate $f_{\textsc{Exactly-$\ell$-Of-$k$}}(x) = 1[\wt(x)=\ell]$ ($\wt(\cdot)$ denotes the Hamming weight).
    \end{itemize}
\end{fexample}

\begin{corollary}\label{cor:examples}
    In all of the following, $n$ denotes the number of variables.
    \begin{itemize}
        \item For every $k \in \N$ \emph{odd} and $\epsilon > 0$, any single-pass streaming algorithm for $\textsc{Max-$k$XOR}_{1,1/2+\epsilon}$ requires $\Omega_\epsilon(n)$ space.
        \item If $w$ denotes the weights of the balanced (approximation resistant) LTF constructed by \textcite{Pot19},
        then any single-pass streaming algorithm for $\textsc{Max-LTF}^w_{1,1/2+\eps}$ requires $\Omega_\epsilon(n)$ space.
        \item For all $k \in \N$, $\ell \ne k/2$, and $\epsilon > 0$,
        any single-pass streaming algorithm for \\
        $\textsc{Max-Exactly-$\ell$-Of-$k$}_{1,\rho(\ell,k)+\epsilon}$
        requires $\Omega_\epsilon(n)$ space, where $\rho(\ell,k) \coloneqq \max_{p \in[0,1]} p^\ell(1-p)^{k-\ell}$.
        \item For all $\gamma \in (\tfrac12,1]$ and $\beta = (3\gamma-1)/2$, any single-pass streaming algorithm for $\textsc{Max-DiCut}_{\gamma-\epsilon,\beta+\epsilon}$ requires $\Omega_\epsilon(n)$ space.
          \item For all $\gamma \in (\tfrac12,1]$ and $\beta = (2\gamma+1)/4$,
any single-pass streaming algorithm for $\textsc{Max-2SAT}_{\gamma-\epsilon,\beta+\epsilon}$ requires $\Omega_\epsilon(n)$ space.
    \end{itemize}
\end{corollary}

We also present a more detailed comparison to previous work below.

\subsection{Comparison with the prior work}\label{sec:comparison}
The most general linear-space lower bounds for CSPs in prior work are by \textcite{CGS+22-linear-space}, based on a structural property that they called \emph{width}, which we now describe.
Put simply, a predicate $f : \Z_q^k \to \{0,1\}$ is \emph{wide} if its support contains a translate of the ``diagonal'' line $\{b \cdot \mathbb{1} : b \in \Z_q\}$, where $\mathbb{1} = (1,\ldots,1) \in \Z_q^k$ is the all-$1$'s vector.
More generally, the \emph{width} of $f$ is 
\[
    \omega(f) \coloneqq \max_{y \in \Z_q^k} \braces*{ \Pr_{\rb \in \Z_q} [f(y+b\cdot\mathbb{1}) = 1] }, 
\]
and the width of $\calF \subseteq \{0,1\}^{[q]^k}$ is $\omega(\calF) \coloneqq \min_{f \in \calF} \omega(f)$.
The result of~\cite{CGS+22-linear-space} is the following:

\begin{theorem}[{\cite{CGS+22-linear-space}}]\label{thm:cgsvv}
    For every $\calF \subseteq \{0,1\}^{[q]^k}$ and $\epsilon > 0$,
    every single-pass streaming algorithm for $\GapMaxCSP{\omega(\calF)-\epsilon}{\rho(\calF)+\epsilon}{\calF}$ requires $\Omega_\epsilon(n)$ space, where \[
        \rho(\calF) \coloneqq \lim_{n \to \infty} \quad \inf_{\Phi\text{ instance of }\MaxCSP{\calF}\text{ on }n\text{ variables}} \braces*{ \optCSP{\Phi} }    \]
    is the best possible uniform lower bound on the optimum value of instances of $\MaxCSP{\calF}$.
\end{theorem}

The authors of~\cite{CGS+22-linear-space} use this lower bound to prove linear-space streaming inapproximability for various CSPs,
including $\textsc{Max-$q$Coloring}$, $\textsc{Max-$q$UniqueGames}$, and $\textsc{Max-$q$LessThan}$ (see \cite[\S4.1]{CGS+22-linear-space}).
For instance, $\textsc{Max-$q$Coloring}$ is wide (i.e., $\omega=1$ for every member of the predicate family).

Note that the basic LP framework for approximability is stronger than the width criterion, in the following sense:
every instance $\Phi$ of $\MaxCSP{\calF}$ has $\optLP{\Phi} \ge \omega(\calF)$.
(The argument is just to set the local distribution $\calY_C$ for each constraint $C$ to be uniform over the corresponding translated line;
the marginals are consistent because they are in fact uniform.)
Thus, all linear-space lower bounds from the work of~\cite{CGS+22-linear-space} (i.e., via \Cref{thm:cgsvv})
can alternatively be derived from our \Cref{thm:main}.

\begin{remark}
    For CSPs including $\textsc{Max-$q$UniqueGames}$~\cite{CGS+22-linear-space} do not apply \Cref{thm:cgsvv} directly;
    rather, they identify a subfamily of predicates $\calF' \subseteq \calF$ and observe that the resulting lower bound for
    $\GapMaxCSP{\omega(\calF')-\epsilon}{\rho(\calF')+\epsilon}{\calF'}$ is a lower bound for
    $\GapMaxCSP{\omega(\calF')-\epsilon}{\rho(\calF')+\epsilon}{\calF}$ too.
    This trick is not needed to apply our \Cref{thm:main}, which (unlike \Cref{thm:cgsvv}) is ``monotone'' with respect to the predicate family.
\end{remark}

We also consider below concrete examples of CSPs where our bounds (i.e., \Cref{thm:main}) are substantially stronger than the bounds of~\cite{CGS+22-linear-space} (i.e., \Cref{thm:cgsvv}).
A predicate family $\calF$ is said to be approximation resistant for a family of algorithms, if it is hard for the algorithms to distinguish the two cases in $\GapMaxCSP{1}{\rho(\calF)+\epsilon}{\calF}$ (where $\rho(\calF)$ is a trivial lower bound). 
Note that in the case $q=2$, the width of a predicate $f : \{0,1\}^k \to \{0,1\}$ can only be either $1/2$ or $1$.
Correspondingly, the~\cite{CGS+22-linear-space} criterion (\Cref{thm:main}), applied to any family $\calF$, either shows that for every $\epsilon > 0$, $\GapMaxCSP{1}{\rho(\calF)+\epsilon}{\calF}$ requires $\Omega_\epsilon(n)$ space (approximation resistance),
or it only shows that for every $\epsilon > 0$, $\GapMaxCSP{1/2-\epsilon}{\rho(\calF)+\epsilon}{\calF}$ requires $\Omega_\epsilon(n)$ space.
We argue that this bound is sometimes too crude for many interesting approximation resistant CSPs.
\vspace{-10 pt}
\paragraph{Approximation-resistant CSPs which evade \cite{CGS+22-linear-space}.}
We observe that a predicate $f : \{0,1\}^k \to \{0,1\}$ is wide (i.e., has $\omega(f) = 1$) iff there is a point $x \in \{0,1\}^k$ such that $f(x) = f(\neg x) = 1$;
otherwise, $\omega(f) = \frac12$.
We claim that all the CSPs in \Cref{ex:higher arity} have $\omega(\calF) = \frac12$,
and therefore the \cite{CGS+22-linear-space} condition proves hardness only for the $\GapMaxCSP{1/2-\epsilon}{\rho(\calF)+\epsilon}{\calF}$ problem.
Indeed:
\begin{itemize}
    \item If $k$ is odd and $x$ is a string where the XOR of all bits is $1$,
    then the XOR of all bits in $\neg x$ will be $0$.
    Therefore, the $\textsc{Max-$k$XOR}$ problem has width $1/2$ (for odd $k$).
    \item If $\ell \ne k/2$ and $x$ is a $k$-bit string with Hamming weight $\ell$, then the Hamming weight of $\neg x$ cannot be $\ell$ (indeed, it must be $k-\ell$).
    Hence, the $\textsc{Max-Exactly-$\ell$-Of-$k$}$ problem has width $1/2$ (for $\ell \ne k/2$).
    \item If $w$ is a weight vector for any balanced LTF, then if $\sum_{\ell=1}^k (-1)^{x_\ell+b_\ell}w_\ell > 0$,
    we have
    \[\sum_{\ell=1}^k (-1)^{-x_\ell+b_\ell}w_\ell < 0\, .\]
    Therefore, the $\textsc{Max-LTF}^w$ problem has width $1/2$.
\end{itemize}
Therefore, \Cref{thm:cgsvv} only asserts a linear-space lower bound against $\GapMaxCSP{\frac12-\epsilon}{\rho(\calF)+\epsilon}{\cdot}$.
In the cases of $\textsc{Max-$k$XOR}$ and $\textsc{Max-LTF}^w$ ($w$ balanced), this is trivial \footnote{
     Actually, there is an \emph{ad hoc} reduction from $\textsc{Max-$(2k-1)$XOR}$ to $\textsc{Max-$2k$XOR}$ which does yield approximation resistance of $\textsc{Max-$(2k-1)$XOR}$, since \Cref{thm:cgsvv} does give approximation resistance for even $k$: Create a single 
new dummy variable $v'$ and add it into every constraint.
     to get a new instance $\Phi'$ of $\textsc{Max-$2k$XOR}$
    ($\Phi'$ has exactly the same optimal value as $\Phi$). 
     The other CSPs we have listed here do not appear to be handled by such \emph{ad hoc} arguments.
}
, since already $\rho(\calF) = \frac12$.
In the case of $\textsc{Max-Exactly-$\ell$-Of-$k$}$, one can show $\rho(\calF) = \binom{k}\ell \max_{p \in [0,1]} p^\ell (1-p)^{k-\ell}$.
For instance, $\rho(f_{\textsc{Exactly-$2$-Of-$3$}}) = 4/9$,
and so \cite{CGS+22-linear-space} does give a linear-space lower bound for $\textsc{Max-Exactly-$2$-Of-$3$}_{1/2-\epsilon,4/9+\epsilon}$.
However, all of these problems are known to admit $(1,\rho(\calF)+\eps)$-LP integrality gaps (they are approximation resistant for the basic LP)~\cite{Pot19, STV25},
and therefore we prove hardness for the $\GapMaxCSP{1}{\rho(\calF)+\epsilon}{\calF}$ problem in each case (establishing the first three items in~\Cref{cor:examples}).

\paragraph{Gap problems and approximation curves.}
For other interesting problems, the existing lower bounds of~\cite{KK19,CGS+22-linear-space} do yield the best possible linear space lower bounds in terms of the overall approximation factor,
but, as discussed in \cite[\S1.2.1]{CGS+22-linear-space}, LP-based lower bounds (as in our \Cref{thm:main}, or the lower bounds of~\cite{FMW26}) give ``finer-grained'' lower bounds.
For instance, the $\textsc{Max-DiCut}$ predicate has width $\omega(f_{\textsc{DiCut}}) = 1/2$ (since it only has one satisfying assignment),
and since its trivial approximability is $\rho(f_{\textsc{DiCut}}) = 1/4$,
the \cite{CGS+22-linear-space} result gives a linear-space lower bound for $\textsc{Max-DiCut}_{1/4+\epsilon,1/2-\epsilon}$.\footnote{
    In fact, \cite{CGS+22-linear-space} is not needed; the same lower bound can be achieved simply by taking the lower bound instances for $\textsc{Max-Cut}_{1/2+\epsilon,1}$ from~\cite{KK19} and randomly directing the edges. }
Based on the LP analysis in \cite[\S1.2.1]{FMW26}, our new lower bound (\Cref{thm:main}) gives a linear-space lower bound against $\textsc{Max-DiCut}_{\gamma-\epsilon,(3/2)\gamma-1/2+\epsilon}$ for every $\gamma \in (\frac12,1]$;
for instance, at $\gamma = \frac34$, we get a linear-space lower bound against $\textsc{Max-DiCut}_{3/4-\epsilon,5/8+\epsilon}$.

\subsection{Concurrent and independent work}
Very recently, two beautiful results, independent of the current work, have also made significant progress towards resolving the approximability of CSPs in the streaming model.

The first is a work of \textcite{FMW26-tight} which proves a linear-space lower bound
which also applies to the multi-pass model. 
For the single-pass model, their results are analogous to the ones shown here, yielding a
linear-space lower bound from integrality gaps for the basic LP. 
Moreover, both our results and \cite{FMW26-tight} rely on the Fourier-pseudorandomness framework of
\textcite{KK19}, using the analytic conditions obtained from the integrality gap
instances to obtain Fourier $\ell_1$ bounds on the distributions over assignments maintained by the
streaming algorithm (see proof overview below).
However, the results of \cite{FMW26-tight} also imply a $n \cdot 2^{-O(p)}$ space lower bound for
algorithms with $p = o(\log n)$ passes.
In addition to Fourier bounds above, their work also uses a decomposition of protocols into structured vs. pseudorandom components (following~\cite{FMW25,FMW26}).

The second result, by \textcite{ABF26}, shows that if the
approximation ratio achieved by the basic LP is $\alpha < 1$, then an $\alpha-\eps$ approximation
can be obtained $O(n^{1 - \Omega_{\eps}(1)})$ space in a single pass, \emph{even for instances with unbounded  degree}.
This can be seen as complimentary to our results (linear-space lower bounds for any better
approximations) and proves that the characterization obtained in \cref{thm:main} for single-pass
algorithms is indeed optimal.

\newcommand{\tv}[1]{\| #1 \|_{\mathsf{tvd}}}
\newcommand{\mper}{\,.}
\newcommand{\mcom}{\,,}

\subsection{Overview of proofs and techniques}
We present a brief overview of the techniques used in our proof, and also discuss their connection
to related works on lower bounds for streaming problems. 
Particularly relevant to our work are the recent dichotomy results for \emph{multi-pass} streaming
algorithms by Fei, Minzer and Wang~\cite{FMW26} which prove an $\Omega(n^{1/3}/p)$ space lower bound
for $p$-pass streaming algorithms using LP integrality gaps, and the linear space (single-pass)
lower bounds by Chou \etal~\cite{CGS+22-linear-space} for a \emph{subclass} of $k$-CSPs with certain algebraic structure.
While the setup for the problem we consider is similar to \cite{FMW26}, the structure of our proof significantly extends the framework of Chou \etal~\cite{CGS+22-linear-space}, which itself builds on an earlier result of Kapralov and Krachun~\cite{KK19} for Max-CUT.

\vspace{-5 pt}
\paragraph{Distribution labeled hypergraphs from LP instances.}
Recall that for a predicate family $\calF \subseteq \{f : [q_0]^k \to \{0,1\}\}$, an instance $\Phi$ of $\MaxCSP{\calF}$ on (say) $n_0$ variables is specified by a collection of constraints $C = (f,e)$, where $e = (i_1, \ldots, i_k) \in [n_0]^k$ denotes an ordered $k$-uniform hyperedge on which the predicate $f$ is applied.
The basic LP corresponding to $\Phi$ (see \cref{fig:basic-lp} ) can be seen to search for \emph{local distributions} $\calY_C$ on $[q_0]^k$, for each constraint $C \in \Phi$, such that the any two local distribution sharing a variable agree on the marginal distribution of the variable.

For $0 < \beta < \gamma \leq 1$, a $(\gamma,\beta)$-integrality gap instance $\Phi$ is one where the LP optimum is $\gamma$ but no true assignment to the variables achieves a value higher than $\beta$.
As is often the case for translating LP/SDP lower bounds to other computational models~\cite{Rag08, CLRS16, Lee15, KMR22, GT17, FMW26}, the integrality gap instance can be used as a (constant-sized) ``gadget" or ``template" to produce an infinite family of hard instances for the streaming problem. 
We view the integrality gap instances as ordered $k$-uniform hypergraphs on the vertex set $[n_0]$ with each edge corresponding to a constraint $C$ labeled by the corresponding distribution $\calY_C$.

\vspace{-5 pt}
\paragraph{Reducing to one-wise uniform distributions.} A very useful observation of Fei, Minzer and Wang~\cite{FMW26} is that by suitably increasing the alphabet size $q_0$ to (say) $q$, the marginal distribution of each variable can be thought of as \emph{uniform} on $[q]$ (a similar construction is also used by Yoshida~\cite{Yos11}).
This is because for rational LP solutions with all probabilities multiples of (say) $1/q$, the uniform distribution on $[q]$ can be used to generate any distribution with probabilities $p_1, \ldots, p_{q_0}$: we partition $[q]$ into sets $S_1, \ldots S_{q_0}$ of size $|S_j| = p_j \cdot q$ and use the induced distribution on the sets.
Each $\calY_C$ then gives a distribution $\calY'_C$ on $[q]^k$, where $a \sim \calY_C$ corresponds to a random sample from $S_{a_1} \times \cdots \times S_{a_k}$, with the marginals of all $k$ variables in $\calY'_C$ being uniform on $[q]$. 
We refer to such distributions as \emph{one-wise uniform}, and will only need to work with these for our proof.

Note that this reduction formally \emph{changes} the predicate family to $\calF' = \{ f \circ \phi ~|~  f \in \calF \}$, where $\phi: [q] \to [q_0]$ denotes the partition $\phi(z_{\ell}) = j \Leftrightarrow z_{\ell} \in S_j$.
However, instances of $\MaxCSP{\calF'}$ generated via this reduction encode instances of $\MaxCSP{\calF}$ and a good (streaming) approximation for instances of the former implies an approximation for $\MaxCSP{\calF}$. 
We will suppress this distinction below, and will simply refer to the (new) predicate family as $\calF$ and one-wise uniform distributions as $\calY_C$ for ease of notation, and will only consider these in the discussion below.
We also identify the new alphabet $[q]$ with the set $\Z_q$.
\vspace{-5 pt}
\paragraph{Signal vs. noise problems.}
As in all previous works \cite{KK19, CGSV24, CGS+22-linear-space, FMW26}, we consider an
average-case version of the problem $\GapMaxCSP{\gamma}{\beta}{\calF}$ where the goal is to
distinguish if the random instances are from a $\yes$ distribution corresponding to a hidden signal
(a planted assignment) or from a $\no$ distribution with random noise.
Following \cite{FMW26}, who also considered the same average-case problem, we refer to it
as the ``Distributional Implicit Hidden Partition'' (DIHP) problem. 

To define the relevant distributions, we consider ``literal shifts'' $f_z(x) = f(x-z)$ for predicates $f
\in \calF$ (note that $\calF$ \emph{may not} be closed under literal shifts).
For a fixed (hidden) assignment $X \in \Z_q^n$ and $y \in [q]^k$, note that  if we consider a
constraint $(f_z, e)$ for $z = X_e - y$ (where $X_e$ denotes the restriction of $X$ to $e$) then $f_z(X_e) = f(y)$. 
Consider a random collection $\Lambda$ of such constraints, with each constraint being generated
independently as follows: (i) choose a random $\rC = (\rf,\re_0)\in \Phi$ from the gadget LP integrality gap
instance and a random $\ry \sim \calY_C$,  (ii) for a randomly drawn $\re \in [n]^k$ include the
constraint $(e,f_z)$ with $z = X_e - y$.
Then expected fraction of constraints in $\Lambda$ satisfied by $X$ equals 
\[
\Exp_{(\re,\rf_{X_e - y}) \sim \Lambda}[f_{X_e - y}(X_e)] 
~=~ \Exp_{C \sim \Phi}\Exp_{y \sim
  \calY_C}[f_{y - X_e}(X_e)] 
~=~ \Exp_{C \sim \Phi}\Exp_{y \sim \calY_C}[f(y)] \mcom
\]
which equals the LP value $\gamma$ on the instance $\Phi$.
Moreover, this remains true in expectation over the choice of $\rX$, \emph{even if restrict to the
  subcollection $\Lambda_0$ of constraints where $y = X_e$} (so that the literal shift is $0$ and we remain
within the family $\calF$). This is because we will have $|\Lambda_0| \approx |\Lambda|/q^k$ with
high probability, and the expected fraction of constraints from $\Lambda$ included in $\Lambda_0$
and satisfied by $\rX$ is
\[
\Exp_{(\re,\rf_{X_e - y}) \sim \Lambda}\left[f_{X_e - y}(X_e) \cdot \mathbb{1}_{(\re,\rf_{\ry-X_e}) \in
    \rLambda_0}\right]
~=~ \Exp_{C \sim \Phi}\Exp_{y \sim \calY_C}\left[f(y) \cdot \Exp_{X \sim \Z_q^n} \mathbb{1}_{y = X_e}\right]
~=~ \frac{1}{q^k} \cdot \Exp_{C \sim \Phi}\Exp_{y \sim \calY_C}[f(y) ] 
\]
We consider an average-case version of the problem
$\GapMaxCSP{\gamma}{\beta}{\calF}$, where the goal is to distinguish between the
$\yes$ case of instances generated as above, with a ``planted'' random $X \sim \Unif(\Z_q^n)$
satisfying $\gamma$ fraction of the constraints), and the $\no$ case where each literal shift $z$ is
chosen uniformly from $\Z_q^k$. 
Also as in previous works~\cite{CGSV24, FMW26} and with an eye towards the relevant communication
problem described next, we will think of instances as ``padded'' where we specify the entire collection of
constraints in $\Lambda$, but the actual CSP instance only consists of the constraints in $\Lambda_0$
(which are easily identifiable).
The value of the instances in the $\no$ case can be seen to be equal to that of an assignment for
the LP instance $\Phi$, with with each variable rounded independently according to its marginal
distribution (which is at most $\beta$). 
We will restrict the discussion below to the special case with a \emph{single} one-wise uniform distribution $\calY$, which contains all our conceptual ideas.

Such ``planted vs. random'' problems are also considered in earlier works on lower bounds for
streaming problems (including \cite{KKS15, KK19, CGSV24, CGS+22-linear-space}) and many other
computational models for approximation and refutation of random CSPs (see \cite{KMOW17} for an
excellent summary). 
Of particular relevance to us are the parallels with works on lower bounds for Sum-of-Squares (SoS) relaxations of CSPs~\cite{BCK15, KMOW17}, which similarly extend previous results based on algebraic structure of XOR predicates~\cite{Gri01,Sch08,Tul09} using the analytic structure of \emph{pairwise} uniformity.

\vspace{-5 pt}
\paragraph{Communication games.}
As in the case of most streaming lower bounds \cite{KKS15, KK19, CGSV24, CGS+22-linear-space, FMW26}, we reduce to proving lower bounds for a related communication problem. 
We consider random instances on $n$ variables above as consisting of $T$ partial $k$-uniform hypermatchings (with each edge in $[n]^k$) of size $m = \alpha \cdot n$, for suitable constants $T \gg 1$ and $\alpha \ll 1$. 
This is used to define a communication game with $T$ players, where for each $t \in [T]$, $\Player_t$ receives a pair of random variables $(\rM_t, \rZ_t)$ where the variables $\rM_t$ denote the random hypermatchings, and the variables $\rZ_t \in \Z_q^{m \times k}$ are distributed differently in the $\yes$ and $\no$ distributions:
\begin{itemize}
\item[$\yes$:] For a common $\rX \sim \Unif(\Z_q^n)$ and independent $\rY_t \sim \calY^{\otimes m}$, we have $\rZ_t = \Pi_{\rM_t}\rX - \rY_t$. 
Here, $\Pi_M$ denotes the projection of $\rX$ to the vertices in $M$ and $\calY^{\otimes m}$ denotes the $m$-wise product of the distribution $\calY$ on $\Z_q^k$.
\item[$\no$:] For each $t \in [T]$, we independently sample $\rZ_t \sim \Unif(\Z_q^{m \times k})$.
\end{itemize}
Note that the above is specialized to the case of a single one-wise uniform distribution $\calY$. For the general case, each matching $\rM_t$ is labeled with a randomly chosen edge $e_0$ from the gadget hypergraph from the LP instance $\Phi$ with a corresponding distribution $\calY_{e_0}$, and we take $\rZ_t = \Pi_{\rM_t}\rX - \rY_t$ for $\rY_t \sim \calY_{e_0}^{\otimes m}$. We also take the set of $n$ variables to be $n_0$-partite, where $n_0$ denotes the number of vertices in the gadget hypergraph.

The communication model we consider is that of one-way broadcast communication, where the players speak in sequence $1, \ldots, T$, and their messages are then visible to all other players.
We show a lower bound of $\Omega(n)$ on the total communication for any protocol for distinguishing the above distributions. 
This is the technical contribution of our work, and also where argument diverges from \cite{FMW26}, who prove an $\Omega(n^{1/3})$ communication lower bound in the stronger number-in-hand model relevant for multi-pass streaming (and in fact there is an $O(\sqrt{n})$ \emph{upper bound} in this model). 
Our goal is thus to show a much stronger lower bound in a weaker communication model.

\vspace{-5 pt}
\paragraph{Hybrid argument and posterior distributions.}
Let $\rZ_t^Y$ and $\rZ_t^N$ denote the samples from $\Z_q^{m \times k}$ for the $t$-th player, drawn respectively from the $\yes$ and $\no$ distributions.
Also, let $\rS_1^Y, \ldots, \rS_t^Y$ denote the sequence of messages sent by the first $t$ players when given inputs $(\rM_1,\rZ_1^Y), \ldots, (\rM_1,\rZ_1^Y)$ from the $\yes$ distribution. We denote this sequence as $\rS^Y_{1:t}$ and define $\rS^N_{1:t}$ analogously. 
We also assume that the in addition to their messages, the players can publish their matchings $\rM_t$ (but not the samples $\rZ_t$) ``for free'', which only makes our lower bound stronger.

Our goal is to show that for any protocol (which can be assumed to be deterministic using Yao's principle) with $o(n)$ communication, we must have that $\tv{\rS_T^Y - \rS_T^N} \leq \delta$ for a suitably small $\delta$ (since the $T$-th player can be assumed to simply output $\yes$ or $\no$).
We will instead show the stronger bound that $\tv{(\rM_{1:T}, \rS_{1:T}^Y) - (\rM_{1:T}, \rS_{1:T}^N)} \leq \delta$ where $\rM_{1:T}$ denote the aggregated matchings.
Analogous to \cite{KK19, CGS+22-linear-space}, we proceed via a hybrid argument, showing that for
each $t \in [T]$, when the first $t-1$ messages are drawn from the $\yes$ distribution, we must have 
\[
\tv{\Player_t(\rM_{1:t}, \rS^Y_{1:t-1}, \rZ^Y_t) - \Player_t(\rM_{1:t}, \rS^Y_{1:t-1}, \rZ^N_t)} ~\leq~ {\delta}/{T}  \mper
\]
Consider the case $t=1$. In this case, the distributions of the random variables $\rZ^Y_1 = \Pi_{\rM_1}\rX - \rY_1$ and $\rZ^N_1 \sim \Unif(\Z_q^{m \times k})$ are clearly identical, since we have $\rX \sim \Unif(\Z_q^n)$. 
However, this is no longer true at $t=2$ since the distribution of $\rX$ is no longer uniform \emph{conditioned on the output of ~$\Player_1$}.
The proofs in \cite{KK19} and \cite{CGS+22-linear-space} track the \emph{posterior} distribution of $\rX$ conditioned on the output of players at various steps, and show that as long as the Fourier tails for the (density function of) the posterior distribution are bounded in $\ell_1$ norm, the distributions of $\rZ_t^Y$ and $\rZ_t^N$ are close, even conditioned on outputs of the previous players.
However, they crucially rely on the (linear) structure of the underlying predicates generating the distribution of $\rY_t$ for proving such Fourier bounds. 
While the results of \cite{KK19} are for Max-CUT, the ones in \cite{CGS+22-linear-space} only hold
for predicates $f: \Z_q^k \to \{0,1\}$ where $f^{-1}(1)$ contains a ``line'' $L_y = \{y + b\cdot
\mathbb{1} ~|~ b \in \Z_q\}$, where $\mathbb{1} = (1, \ldots, 1)$.

No such linear structure is available in our case since the goal is to prove lower bounds for \emph{all} CSPs. %
The technical core of our argument is obtaining such Fourier bounds only relying on the (weak) \emph{analytic} property that the distributions $\calY$ generating $\rY_t$ can be taken to be one-wise uniform.
As mentioned earlier, this phenomenon also has a parallel in the literature on SoS lower bounds~\cite{KMOW17}, where the technical challenge was working with the weaker analytic Fourier structure provided by pairwise uniformity (instead of the linear structure of XOR predicates).
\vspace{-5 pt}
\paragraph{Fourier bounds via one-wise uniformity.} 
We now consider the key technical challenge, which arises in understanding the Fourier structure for the posterior distribution. 
For the case $t = 2$, let the output of $\Player_1(\rM_1,\rZ_1)$ be in $\Z_q^r$ (it will be convenient to think of output also as $q$-ary). 
For a fixed matching $M$ (value of $\rM_1$) the outputs of $\Player_1$ partition $\Z_q^{m \times k}$ (the domain of $\rZ_1$). 
Fixing an output then defines a set $\calB \subseteq \Z_q^{m \times k}$, which is typically of size $q^{mk - r}$.
This induces a posterior distribution $\calD_{M, \calB}$ on $\Z_q^n$ with probabilities
\[
\calD_{M,\calB}(X_0) ~=~ \Pr_{\rX \sim \Z_q^n, \rY \sim \calY^{\otimes m}}\left[ \rX = X_0 ~\mid~ \Pi_{M}\rX - \rY \in \calB \right] \mper
\]
For a distribution $\calD$ on $\Z_q^N$, let $\Den{\calD}$ denote its density $q^N \cdot \calD$ (so that $\Exp[\Den{\calD}] = 1$) and let $\Den{\calB}$ be the density of $\Unif(\calB)$. 
The density $\Den{\calD_{M,\calB}}$ can be expressed as $\Den{\calD_{M,\calB}} = \Den{\calB} *
\Den{\calY^{\otimes m}}$. 
The Fourier coefficient for any ``frequency'' $U \in \Z_q^{n}$ is then given by
\[
\hat{\Den{\calD_{M,\calB}}}(U) ~=~ 
\begin{cases}
\hat{\Den{\calB}}(\Pi_{M}U) \cdot \hat{\Den{\calY^{\otimes m}}}(-\Pi_{M}U) &~\text{if}~\supp(U) \subseteq M \\
0 &~\text{otherwise}
\end{cases} \mper
\]
The Fourier conditions needed on the posterior distribution require understanding, for each $h$, the sum $\sum_{\wt(U) = h} |\hat{\Den{\calD_{M,\calB}}}(U)|$.
Writing $V = \Pi_M U \in \Z_q^{m \times k}$ as a matrix with rows $V_1, \ldots, V_m$, the argument of \cite{CGS+22-linear-space}  can be seen as decomposing this sum based on the number of nonzero rows (say) $\ell$ in $V$.
They then study the $\ell_2$ norm of Fourier coefficients for $V$ with exactly $h$ nonzero entries and $\ell$ nonzero rows (together with Cauchy-Schwarz and counting to obtain $\ell_1$ bounds).
Note that this ``$(h,\ell)$-inequality'' is a refinement of the usual level-$h$ inequality, which bounds the Fourier mass of coefficients with $h$ nonzero entries. 
Such inequalities are obtained in ~\cite{CGS+22-linear-space} by using the algebraic structure of predicates defining $\calY$ to eliminate one nonzero entry for each nonzero row.

While we do not have such algebraic structure, the one-wise uniformity of $\calY$ implies that 
\[
\hat{\Den{\calY^{\otimes m}}}(-V) = \prod_{i \in [m]}\hat{\Den{\calY}}(-V_i) \neq 0
\quad \implies \quad
\forall i \in [m], ~~ |\supp(V_i)| \neq 1. \]
Due to the product structure of Fourier coefficients in $\Den{\calD_{M,\calB}}$, it then suffices to bound the Fourier mass of the function $\Den{\calB}: \Z_q^{m \times k} \to \R_{\geq 0}$ (with $\norm{\Den{\calB}}_{\infty} \leq q^r$) on Fourier coefficients corresponding to matrices $V = \Pi_M U$ with ``singleton-free'' rows (since $\hat{\Den{\calY^{\otimes m}}}(-V) = 0$ otherwise).

To prove such a refinement of the level-$h$ inequality, we first recall that usual inequality is proved by applying the Bonami--Beckner operator $T_{\rho}$, which has the same eigenvalue $\rho^h$ for all Fourier coefficients of weight $h$, and then using hypercontractivity to bound $\norm{T_\rho f}_2 \leq \norm{f}_p$ for an appropriate $p$.
To prove our ``singleton-free" inequality, we now design a \emph{custom noise operator} with eigenvalues taking advantage of this singleton-free structure. In particular, consider the linear operator $S_{\rho}$ defined on $\Z_q^k$ with characters $\chi_u(x) = \omega_q^{\langle u,x\rangle}$ as eigenvectors, and eigenvalues
\[
S_{\rho}(\chi_u) ~=~ \lambda_{\rho}(u) \cdot \chi_u 
\qquad \text{where} \qquad 
\lambda_\rho(u) ~=~
\begin{cases}
1 & \text{if}~~u=0 \\
0 & \text{if}~|\supp(u)|=1 \\
\rho^{|\supp(u)|-1} & \text{otherwise}
\end{cases} \mper
\]
One can check that for characters $\chi_V$ corresponding to $V \in \Z_q^{m \times k}$, we have $S_{\rho}^{\otimes m}(\chi_V) = \rho^{h - \ell} \cdot \chi_V$, where $h = \supp(V)$ and $\ell$ is the number of (singleton-free) nonzero rows. 
We can also show that if the Fourier spectrum of $g: \Z_q^{m \times k} \to \R$ is singleton-free, then $\norm{S_\rho^{\otimes m} g}_2 \leq \norm{T_{\rho}^* g}_2$, where $T_{\rho}^*$ is a ``row-wise" Bonami-Beckner operator, resampling inputs corresponding to entire rows in $\Z_q^k$, with probability $1-\rho$.
This comparison also yields a hypercontractive inequality for our new operator, and we use this obtain the required $(h,\ell)$-inequality. 

This argument (described in \cref{sec:singleton-free}) is a novel technical contribution of our work, and may also be useful for other Fourier-analytic applications.
Beyond generalizing the bound in \cite{CGS+22-linear-space}, it also significantly simplifies their proof, avoiding the careful elimination of entries in nonzero rows.

\vspace{-5 pt}
\paragraph{Covering with structured centers.}
While the above argument suffices for analyzing the posterior distribution at $t=2$ (conditioning on the output of $\Player_1$), understanding the Fourier structure at later times presents a slight twist on the above problem. 
If $\calD_{t}$ is the posterior distribution after $\Player_{t}$ speaks, $M$ is the revealed matching for $\Player_{t+1}$ and $\calB \subseteq \Z^{m \times k}$ is a set corresponding to the output of $\Player_{t+1}$, the posterior distribution $\calD_{t+1,M,\calB}$ is given by
\[
\calD_{t+1,M,\calB}(X_0) ~=~ \Pr_{\rX \sim \calD_{t}, \rY \sim \calY^{\otimes m}}\left[ \rX = X_0 ~\mid~ \Pi_{M}\rX - \rY \in \calB \right] \mcom
\]
where we now condition the distribution $\calD_t$ instead of the uniform distribution on $\Z_q^n$.
The new density is now given by $\Den{\calD_{t+1,M,\calB}} = \nu_{M,\calB} \cdot \Den{\calD_t} \cdot (\Den{\calB} * \Den{\calY^{\otimes m}})$, where $\nu_{M,\calB}$ is a normalizing factor, and the Fourier coefficients are given by
\[
\hat{\Den{\calD_{t+1,M,\calB}}}(W) ~=~ 
C_{M,\calB} \cdot \sum_{\substack{U \in \Z_q^n \\ \supp(W-U) \subseteq M}} \hat{\Den{\calD_t}}(U) \cdot \hat{\Den{\calB}}(\Pi_{M}(W-U)) \cdot \hat{\Den{\calY^{\otimes m}}}(-\Pi_{M}(W-U))
\]
Assuming the Fourier $\ell_1$ norm of $\Den{\calD_t}$ can be controlled using an induction hypothesis, understanding the norm for $\Den{t+1,M,\calB}$ now requires bounding, for each $U$, expressions of the form
\[
\sum_{\substack{\wt(W)=h \\ \supp(W-U) \subseteq M} }\left| \hat{\Den{\calB}}(\Pi_{M}(W-U)) \cdot \hat{\Den{\calY^{\otimes m}}}(-\Pi_{M}(W-U)) \right|
\]
While we can again study this expression by using a fine-grained weight inequality in terms of the number of non-zero rows in $\Pi_M (W-U)$, the weight constraint is now on $W$, while the relevant row structure is only available for $\Pi_M(W-U)$.
To handle this mismatch, we need to ``cover'' such a sum for an arbitrary $U$ using a similar sum for ``structured centers'' $U'$,  such that for $W'$ defined by $W'-U' = W-U$, we have that $W'$ itself is singleton-free
 and we can also bound $\wt(W')$ (with smaller values of $h' < h$).
\footnote{A similar covering argument was actually also used in \cite{CGS+22-linear-space}, though it's purpose there was only to reduce the weight $h$. We need to modify this argument to provide the required  structure for $U'$ and $W'$.}
We rely on the properties of the random matching $M$ to bound the number of structured centers used in our covering, leading to the desired inductive Fourier bound.





\section{Preliminaries}

\subsection{General notation}
We use \textbf{boldface} to represent random variables. For an event $E$, we also use $\Ind{E}$ to
denote the indicator for the event.
For a function $g : \Z_q^N \to \C$ and $p \in [1,\infty]$, we use the convention that \[
    \|g\|_p \coloneqq \parens*{ \Exp_{\rX \in \Z_q^N} [|g(\rX)|^p]}^{1/p}. \]
In particular, $\|g\|_1 = \Exp_{\rX \in \Z_q^N}[|g(X)|]$ and $\|g\|_\infty = \max_{X \in \Z_q^N} |g(X)|$.

\subsection{Distributions}

For a finite set $S$, we use $\Dist{S}$ to denote the set of all distributions over $S$.
If $S$ is a finite set and $S^k$ its $k$-fold product set, a distribution $\calD \in \Dist{S^k}$ is \emph{one-wise uniform} if for $(\rs_j)_{j \in [k]} \sim \calD$,
the marginal distribution of $\rs_j$ is uniform on $S$ for every $j \in [k]$.
We use $\OWU{S^k}$ to denote the set of all one-wise uniform distributions.

For a distribution $\cD$, $\cD^{\otimes \ell}$ denotes the $\ell$-wise product distribution
$\cD\times\cdots\times\cD$.
In particular, if $\calD \in \Dist{\Z_q^k}$, we can view $\calD^{\otimes m} \in \Dist{\Z_q^{m \times k}}$ as a distribution on matrices.

Abusing notation, we can view a distribution $\calD \in \Dist{S}$ as a function $\calD : S \to \R$, where $\calD(x)$ represents the probability that a sample $\rx \sim \calD$ satisfies $\rx = x$.
We can then define the \emph{density} function $\Den{\calD} : S \to \R$ via \[
\Den{\calD}(x_0) \coloneqq \calD(x_0) \cdot |S|. \]
One useful property of density functions is:
\begin{proposition}\label{prop:prelim:density reweight}
    Let $\calD \in \Delta(S)$ be a distribution and $\phi : S \to \C$ a function.
    Then \[
    \Exp_{\rx \in S}[\Den{\calD}(\rx) \cdot \phi(\rx)] = \Exp_{\rx \sim \calD}[\phi(\rx)]. \]
\end{proposition}

\begin{proof}
    We can expand \[
    \Exp_{\rx \in S}[\Den{\calD}(\rx) \cdot \phi(\rx)]
    = |S| \Exp_{\rx \in S}[\calD(\rx) \cdot \phi(\rx)]
    = \sum_{x \in S} \calD(x) \cdot \phi(x)
    = \Exp_{\rx \sim \calD}[\phi(\rx)]. \qedhere \]
\end{proof}

\subsection{Total variation distance}

\begin{proposition}[Data processing inequality]\label{prop:data_processing}
    For random variables $\rX, \rY$ and $\rW$, if $\rW$ is independent of both $\rX$ and $\rY$, and $f$ is a function, then $\|f(\rX,\rW)-f(\rY,\rW)\|_\tvd \le \|\rX-\rY\|_\tvd$.
\end{proposition}

\begin{proposition}[{\cite[Lemma 2.3]{CGS+22-linear-space}}]\label{lem:statistical_test}
	Let $\rX,~\rY,~\rW$ be random variables and let $f$ be a function.
    If there exists $\delta>0$ such that for every fixed $X_0$ in the support of $\rX$, we have
	\[\|f(X_0,\rY)-f(X_0,\rW)\|_\tvd \le \delta\, , \] then the following holds:
	\[\|(\rX,f(\rX,\rY))-(\rX,f(\rX,\rW))\|_{tvd}\le \delta\, .\]
\end{proposition}

\begin{proposition}[{\cite[Lemma B.2]{KK19}}]\label{lem:KKsubstitutionlemma}
	Let $\rX^1,\rX^2$ be random variables taking values on the same sample space,
    let $\rZ^1,\rZ^2$ be random variables taking values on the same sample space, and
    let $f$ be a function.
    If $\rZ^2$ is independent of $\rX^1,\rX^2$, then \[
	\|(\rX^1,f(\rX^1,\rZ^1))-(\rX^2,f(\rX^2,\rZ^2)) \|_{\tvd} \le \|(\rX^1,f(\rX^1,\rZ^1))-(\rX^1,f(\rX^1,\rZ^2)) \|_{\tvd} + \|\rX^1-\rX^2\|_{\tvd} \, .
	\]
\end{proposition}

\subsection{Convolutions}

If $f, g : \Z_q^N \to \C$ are functions, their \emph{convolution} $f * g : \Z_q^N \to \C$ is defined as \[
    (f * g)(z) \coloneqq \Exp_{\rx \in \Z_q^N} [f(\rx) \cdot g(z - \rx)]. \]

We now give two propositions which help interpret convolutions.
The first shows that convolving a function $f$ with a density function can be viewed as ``noising'' $f$;
the second shows that convolving two density functions corresponds to the density function of a sum distribution.

\begin{proposition}\label{prop:convolution with density}
    Let $g : \Z_q^N \to \C$ and $\calD \in \Dist{\Z_q^N}$.
    Then \[
    (g * \Den{\calD})(z) = \Exp_{\rx \sim \calD} [g(z-\rx)]. \]
\end{proposition}
\begin{proof}
    Apply \Cref{prop:prelim:density reweight} with the function $\phi(x) \coloneqq g(z-x)$.
\end{proof}

\begin{proposition}\label{prop:prelim:density of sum}
Let $\cD_1, \cD_2 \in \Dist{\Z_q^N}$.
Let $\cD\in \Delta(\Z_q^N)$ denote the marginal distribution of $\rx + \ry$
when $\rx \sim \cD_1, \ry \sim \cD_2$ independently, so that for $z \in \Z_q^N$, $\calD(z) = \Pr_{\rx\sim \cD_1, \ry\sim \cD_2}[\rx+\ry = z]$.
Then the density function of $\calD$ is the convolution of the density functions of $\calD_1$ and $\calD_2$: For every $z \in \Z_q^N$, \[
\Den{\calD}(z) = (\mu_{\calD_1} * \mu_{\calD_2})(z). \]
\end{proposition}

\begin{proof}
    We have $\calD(z) = \sum_{x,y \in \Z_q^N} \calD_1(x) \cdot \calD_2(y) \cdot \Ind{x + y
    = z} = \sum_{x \in \Z_q^N} \calD_1(z) \cdot \calD_2(z - x)$.
    Correspondingly, \[
    \Den{\calD}(z) = q^N \cdot \sum_{x \in \Z_q^N} \Den{\calD_1}(x) \cdot \Den{\calD_2}(z - x) 
    = \Exp_{\rx \in \Z_q^N} [\Den{\calD_1}(\rx) \cdot \Den{\calD_2}(z - \rx)]
    = (\mu_{\calD_1} * \mu_{\calD_2})(z). \qedhere \] 
\end{proof}

\subsection{Fourier analysis}
In this subsection, we use $N$ for the dimension of a $\Z_q$-vector space (or module, if $q$ is not
a prime power) over which we perform Fourier analysis.
We will use various settings of $N$ within the paper, e.g., $n$, $k$, $nk$, $km$, $nk'$, etc.

\begin{definition}[Fourier characters]\label{def:prelim:fourier character}
For $u, x \in \Z_q^N$, we define the \emph{character function} \[
    \chi_u(x) \coloneqq \omega_q^{\langle u, x \rangle}, \]
where we fix the $q$-th root of unity $\omega_q \coloneqq e^{2\pi i / q}$ and $\langle \cdot, \cdot \rangle$ represents the standard inner product modulo $q$.
We will write $\omega = \omega_q$ when $q$ is clear from context.
\end{definition}

\begin{definition}[Fourier coefficients]\label{def:prelim:fourier coefficient}
    For any function $g : \Z_q^N \to \C$ and $u \in \Z_q^N$, we define the \emph{Fourier coefficient} \[
        \hat{g}(u) \coloneqq \Exp_{\rx \sim \Z_q^N} [g(\rx) \cdot \overline{\chi_u(\rx)}], \]
    where $\overline{\cdot}$ denotes complex conjugation, so that $\overline{\chi_u(\rx)} = \omega^{-\langle u, x \rangle}$.
    We sometimes refer to the vector of $u$ as a \emph{frequency}.
\end{definition}

\begin{proposition}[Fourier coefficients of convolutions, {e.g., \cite[Theorem~8.60]{OD14}}]\label{prop:prelim:convolution}
    For every $f, g : \Z_q^N \to \C$ and $u \in \Z_q^N$, \[
        \hat{f*g}(u) = \hat{f}(u) \cdot \hat{g}(u).
    \]
    Conversely, $\hat{f \cdot g}(u) = q^N (\hat{f} * \hat{g})(u) = \sum_{v \in \Z_q^N} f(v) g(u-v)$.
\end{proposition}

\begin{proposition}[Expansion in the Fourier basis]\label{prop:prelim:fourier expansion}
    For every $f : \Z_q^N \to \C$ and $x \in \Z_q^N$, we have \[
        f(x) = \sum_{u \in \Z_q^N} \hat{f}(u) \cdot \chi_u(x). \]
\end{proposition}

\begin{proposition}[Parseval's identity]\label{prop:prelim:parseval}
    For every $g : \Z_q^N \to \C$, \[
    \|g\|_2^2 = \sum_{U \in \Z_q^N} |\hat{g}(U)|^2. \]
\end{proposition}

\begin{proposition}[Fourier coefficients of (density functions of) distributions]\label{prop:prelim:fourier of density}
    If $\calD \in \Dist{\Z_q^N}$, then \[
    \hat{\Den{\calD}}(u) = \Exp_{\rx \sim \calD} \bracks*{ \overline{\chi_u(\rx)} }. \]
\end{proposition}

\begin{proof}
    Combine \Cref{prop:prelim:density reweight,def:prelim:fourier coefficient}.
\end{proof}

\begin{proposition}\label{prop:prelim:vanish around circle}
    For every $q \ge 2 \in \N$, \[
    \Exp_{\rb \in \Z_q} [\omega_q^\rb] = 0. \]
\end{proposition}

\begin{proof}
    The geometric sum formula gives $\sum_{b=0}^{q-1} \omega^b = \frac{1-\omega^q}{1-\omega} = 0$ since $\omega^q = 1$.
\end{proof}

\begin{proposition}[Fourier coefficients of one-wise uniform distributions]\label{prop:prelim:fourier of one-wise}
    Let $\calY \in \OWU{\Z_q^N}$ be a one-wise uniform distribution.
    If $u \in \Z_q^N$ has $|\supp(u)| = 1$, then \[
    \hat{\Den{\calY}}(u) = 0. \]
\end{proposition}

\begin{proof}
    Suppose $j \in [N]$ is the single nonzero coordinate in $u$.
    Then, \[
        \hat{\Den{\calY}}(u) = \Exp_{\rx \sim \calY} \bracks*{ \overline{\chi_u(\rx) }}
        = \Exp_{\rx \sim \calY} \bracks*{ \omega^{-\langle u, \rx \rangle}}
        = \Exp_{\rx \sim \calY} \bracks*{ \omega^{-u_j \cdot \rx_j}}
        = \Exp_{\rb \in \Z_q} [\omega^{-u_j \cdot \rb}]
        = \Exp_{\rb \in \Z_q} [\omega^{\rb}]
        = 0, \]
    where we used, respectively, \Cref{prop:prelim:density reweight}, the definition of $\chi$, that $u$ is supported only on $j$, one-wise uniformity, the nonzeroness of $u_j$, and \Cref{prop:prelim:vanish around circle}.
\end{proof}

\begin{proposition}[Fourier coefficients of product distributions]\label{prop:prelim:fourier of product}
    Let $\calY \in \Delta(\Z_q^k)$.
    Then for every $U \in \Z_q^{k \times m}$, we have: \[
        \hat{\Den{\calY^{\otimes m}}}(U) = \prod_{j=1}^m \hat{\Den{\calY}}(U_j). \]
\end{proposition}

\begin{proposition}[Fourier coefficients of marginals]\label{prop:prelim:fourier of marginal}
    Let $\tilde{N},N \in \N$, $\Pi : \Z_q^{\tilde{N}} \to \Z_q^N$ be ($\Z_q$-)linear and surjective,
    $\tilde{g} : \Z_q^{\tilde{N}} \to \C$ be any function, and $g : \Z_q^N \to \C$ be defined by $g(x) \coloneqq \Exp_{\tilde{\rx} \in \Z_q^{\tilde{N}}} [g(\tilde{\rx}) \mid \Pi(\tilde{\rx}) = x]$.
    Then for every $u \in \Z_q^N$, \[
        \hat{g}(u) =
        \hat{\tilde{g}}(\Pi^\intercal (u)), \]
    where $\Pi^\intercal$ is the (standard) transpose of $\Pi$ (which is adjoint with respect to the standard inner product).
\end{proposition}

\begin{proof}
    Let $u \in \Z_q^N$.
    Note that for every $\tilde{x} \in \Z_q^{\tilde{N}}$, we have $\langle u, \Pi(\tilde{x}) \rangle = \langle \Pi^\intercal(u), \tilde{x} \rangle$ by adjointness, and hence $\chi_u(\Pi(\tilde{x})) = \chi_{\Pi^\intercal(u)}(\tilde{x})$.
    Hence,
    \begin{multline*}
        \hat{g}(u) = \Exp_{\rx \in \Z_q^N} \bracks*{ g(\rx) \cdot \overline{\chi_u(\rx)} }
        = \Exp_{\rx \in \Z_q^N} \bracks*{ \Exp_{\tilde{\rx} \in \Z_q^{\tilde{N}}} [\tilde{g}(\tilde{\rx}) \mid \Pi(\tilde{\rx}) = \rx] \cdot \overline{\chi_u(\rx)} }
        = \Exp_{\tilde{\rx} \in \Z_q^{\tilde{N}}} \bracks*{ \tilde{g}(\tilde{\rx}) \cdot \overline{\chi_{u}(\Pi(\tilde{\rx}))} } \\
        = \Exp_{\tilde{\rx} \in \Z_q^N} \bracks*{ \tilde{g}(\tilde{\rx}) \cdot \overline{\chi_{\Pi^\intercal(u)}(\tilde{\rx})}}
        = \hat{\tilde{g}}(\Pi^\intercal(u)),
    \end{multline*}
    where the third equality uses that for $\Pi : \Z_q^{\tilde{N}} \to \Z_q^N$ surjective, the marginal distribution of $\Pi(\tilde{\rx})$
    for uniform $\tilde{\rx} \in \Z_q^{\tilde{N}}$ is uniform on $\Z_q^N$.
\end{proof}

\begin{proposition}[Fourier coefficients of coordinate-projected functions]\label{prop:prelim:projection}
    Let $N \in \N$, $\calS \subseteq [N]$ a subset of the coordinates,
    $g : \Z_q^\calS \to \C$ be any function,  and $\Pi_\calS : \Z_q^N \to \Z_q^\calS$ the standard projection map onto the coordinates in $\calS$.
    Then for every $u \in \Z_q^N$, \[
        \hat{g \circ \Pi_\calS}(u) = \begin{cases}
            \hat{g}(\Pi_\calS(u)) & \text{if }\supp(u) \subseteq \calS, \\
            0 & \text{otherwise}.
        \end{cases} \]
\end{proposition}

\begin{proof}
    Let $u \in \Z_q^N$.
    We write e.g. $u_\calS \coloneqq \Pi_\calS(u)$ for short.
    We have \[
        \hat{g \circ \Pi_\calS}(u) = \Exp_{\rx \in \Z_q^N} [g(\rx_\calS) \cdot \omega^{-\langle u, \rx \rangle}]
        = \Exp_{\overline{\rx} \in \Z_q^\calS} \bracks*{ g(\overline{\rx}) \cdot \Exp_{\rx \in \Z_q^N} [ \omega^{-\langle u, \rx \rangle} \mid \rx_\calS = \overline{\rx}] } \]
    where we used the fact that the marginal distribution of $\rx_\calS$ for $\rx \in \Z_q^N$ is uniform on $\Z_q^\calS$.
    Now fix $\overline{x} \in \Z_q^\calS$; we claim that \[
    \Exp_{\rx \in \Z_q^N} [ \omega^{-\langle u, \rx \rangle} \mid \rx_\calS = \overline{x}]
    = \begin{cases} \omega^{-\langle u_\calS, \overline{x} \rangle} & \text{if } \supp(u) \subseteq \calS \\ 0 & \text{otherwise}, \end{cases} \]
    which obviously proves the proposition.
    Indeed, if $\supp(u) \subseteq \calS$, then $\langle u, x \rangle = \langle u_\calS, x_\calS \rangle$
    and hence $\omega^{-\langle u,x \rangle} = \omega^{-\langle u_\calS,x_\calS \rangle}$ (for every $x \in \Z_q^N$).
    On the other hand, if $\supp(u) \not\subseteq \calS$, then picking some $i \in \supp(u) \setminus \calS$,
    we can write $\langle u, x \rangle = u_i \cdot x_i + \langle u_{[N]\setminus i}, x_{[N]\setminus i} \rangle$, and therefore \[
    \Exp_{\rx \in \Z_q^N} [ \omega^{-\langle u, \rx \rangle} \mid \rx_\calS = \overline{x}]
    = \Exp_{\rb \in \Z_q} \bracks*{ \omega^{-u_i \cdot \rb} \Exp_{\rx \in \Z_q^N} [\omega^{-\langle u_{[N]\setminus\{i\}}, \rx_{[N]\setminus\{i\}} \rangle} \mid \rx_\calS = \overline{x}\text{ and }\rx_j = \rb] } \]
    which again vanishes by \Cref{prop:prelim:vanish around circle} (and the fact that $u_i \ne 0$).
\end{proof}

\begin{remark}
    \Cref{prop:prelim:projection} can be easily generalized to the case where $\Pi : \Z_q^N \to \Z_q^{N'}$ is any surjective linear map,
    but we will not need it in this generality, so choose to keep its current form for simplicity in notation.
\end{remark}



\section{The communication problem and streaming reduction}

We now turn to defining the communication problem used in our lower bound.
This is based on the communication problem defined in~\cite{FMW26}, as modified by~\cite{STV25} for the single-pass setting.

\begin{definition}[Partite hypermatchings]\label{def:phm}
For $n,m,k \in \N$, we define the set of $k$-uniform partite hypermatchings with $m$ edges on a fixed $k$-partite vertex-set $[n] \times [k]$ as: \[
    \PHM{m}{k}{n} \coloneqq \{ M \in [n]^{m \times k} : \text{all entries in each column of }M\text{ are distinct} \}. \]
A $k$-partite hypermatching $M \in \PHM{m}{k}{n}$
gives rise to an injective mapping $\iota_M : [m] \times [k] \to [n] \times [k]$ defined via $\iota_M(j,\ell) \coloneqq (M_{j\ell},\ell)$.
(Indeed, the mapping $M \mapsto \iota_M$ yields a one-to-one correspondence between the set of $k$-partite hypermatchings $\PHM{m}{k}{n}$
and the set of injective mappings $\iota : [m] \times [k] \to [n] \times [k]$ with the property that
$\iota([m] \times \{\ell\}) \subseteq \iota([n] \times \{\ell\})$ for every $\ell \in [k]$.)
If $e = (e_1,\ldots,e_k) \in [n]^k$ is a row of such a matrix, we write $\supp(e) \coloneqq \{(e_\ell, \ell) : \ell \in [k]\} \subset [n]\times [k]$
as the set of vertices participating in the edge $e$.
We similarly write $\supp(M) \coloneqq \bigcup_{j=1}^m \supp(M_j)$, where $M_j$ is the $j$-th row of $M$.
\end{definition}

\begin{definition}[Gadget hypergraphs]\label{def:hg}
For $T,k,k' \in \N$, we define the set of $k$-uniform hypergraphs with $T$ edges on a fixed vertex-set $[k']$ as: \[
    \HG{T}{k}{k'} \coloneqq \{ G \in [k']^{T \times k} : \text{all entries in each row of }G\text{ are distinct} \}. \]
A vector $e \in [k']^k$ where every entry is distinct may equivalently be viewed as an injection $\phi : [k] \to [k']$.
By (slight) abuse of notation, we will also defined $\phi(j,\ell)$ as $(j,\phi(\ell))$ for $(j,\ell) \in [m] \times [k]$.
Thus, we can use such an injection $\phi$ to embed $k$-partite matchings into $k'$-partite graphs.
\end{definition}

\begin{definition}\label{def:correspondence}
    For an injection $\phi : [k] \to [k']$ and a $k$-partite hypermatching $M \in \PHM{m}{k}{n}$,
    we define a corresponding projection operator $\Pi^\phi_M : \Z_q^{n \times k'} \to \Z_q^{m \times k}$
    via $(\Pi^\phi_M X)_{j,\ell} \coloneqq X_{\phi(\iota_M(j,\ell))}$.
    For a matrix $X \in \Z_q^{n \times k'}$, viewed as a labeling on a set of vertices $[n] \times [k']$,
    $\Pi^\phi_M X$ is precisely the \emph{induced} labeling on $[m] \times [k]$.
\end{definition}

We now define the following communication problem, which is a single-pass variant of the problem in~\cite{FMW26}.

\begin{definition}[``Distributional implicit hidden partition'' problem]\label{def:DIHP}
Suppose we have a $k$-uniform hypergraph $G \in \HG{T}{k'}{k}$
and a sequence of one-wise uniform distributions $\calY_1,\ldots,\calY_T \in \OWU{\Z_q^k}$.
Let $m \le n \in \N$.
Consider the following communication problem DIHP.
\begin{itemize}
    \item \emph{Communication structure}:
    There are $T$ players named $\Player_t$ for $t \in [T]$.
    This is a one-way communication game; the players speak in order from $1$ to $T$.
    \item \emph{Input type:}
    Each player $\Player_t$ gets a $k$-partite hypermatching $M_t \in \PHM{m}{k}{n}$ and a noisy signal $Z_t \in \Z_q^{m \times k}$.
    (I.e., each player can be viewed as a function $\Player_t : \PHM{m}{k}{n} \times \Z_q^{m \times k} \times \{0,1\}^s \to \{0,1\}^s$.)
    \item \emph{Input distributions:}
    The goal is to distinguish between the $\yes$ and $\no$ input distributions.
    In both cases, the inputs $\rM_t$ are uniformly random and independent.
    The two cases are defined as follows:
    \begin{itemize}
        \item In the $\no$ case, $\rZ_t$ is also picked uniformly and independently for every $t \in [T]$.
        \item In the $\yes$ case, we sample an underlying hidden assignment $\rX^* \in \Z_q^{n \times k'}$ (uniformly at random).
        Then, for every $t \in [T]$,
        we sample $\rY_t \sim \calY_t^{\otimes m}$ and set $\rZ_t \coloneqq \Pi^{\phi_t}_{\rM_t} \rX^* - \rY_t$. \qedhere
    \end{itemize}
\end{itemize}
\end{definition}

We interpret these distributions in the following way:
In the $\yes$ case, $X$ is $\Z_q$-labeling on $[n] \times [k']$.
\`A la \Cref{def:correspondence}, the $k$-partite hypermatching $M_t$ and the injection $\phi_t : [k] \to [k']$ give a correspondence by which $X$ also ``induces'' a $\Z_q$-labeling on $[m] \times [k]$, namely, the labeling $\Pi_{M_t}^{\phi_t} X$.
The players want to determine whether their inputs $Z_t$ are consistent with a global labeling $X$ ($\yes$ case) or are just independent and random ($\no$ case).
However, in the $\yes$ case, the players do not receive the induced labelings $\Pi_{M_t}^{\phi_t} X$ directly;
instead, they are ``masked'' by randomized noise drawn from the corresponding distribution $\calY_t$.

The goal of the players is to determine whether their input is drawn from the $\yes$ or the $\no$ distribution. As in previous works \cite{CGS+22-linear-space,KK19}, we consider the blackboard communication model, where players communicate via a shared blackboard in a fixed sequential order beginning with $\Player_1$. Without loss of generality, we assume that the last player $\Player_t$ outputs a single bit indicating which distribution their input is drawn from.
In the $i$-th round, $\Player_i$ observes the current contents of the blackboard and writes their message $\rS_i$. We further assume that each $\Player_i$ reveals their hypermatching $\rM_i$ alongside their message $\rS_i$, at no cost.\footnote{Note that this assumption only yields a stronger lower bound.} The communication cost of a protocol is then defined as the worst-case total length (in bits) of all \emph{messages} written on the blackboard, i.e., $\sum_{i \in [T]} |\rS_i|$.

We prove the following linear lower bound on the communication complexity of DIHP in \cref{sec:communication_lb}.

\begin{restatable}[Linear lower bound for DIHP]{theorem}{communicationlb}\label{thm:main DIHP}
	For every $q,k\in\mathbb{N}$ and $\delta\in(0,1/2)$, there exists $\alpha_0>0$ such that for every $\alpha\in(0,\alpha_0]$ and $T\in\mathbb{N}$, there exists $n_0\in\N$ and $\tau\in(0,1)$ such that the following holds. When $n\geq n_0$, the communication complexity of any protocol for DIHP that succeeds with advantage $\delta$ is at least $\tau n$.
\end{restatable}
\subsection{The streaming reduction}

We now present a reduction from DIHP to \textsc{Max-CSP} due to~\cite{FMW26}, which (as we will see later in this section) can be implemented in the streaming setting.

\begin{definition}[$\MaxCSP{\calF}$ instances from labeled matchings]\label{def:reduction}
Let $k,k',q_0,q \in \N$, $\phi : [k] \to [k']$ be an injection, and $f : [q_0]^k \to \{0,1\}$ be a predicate.

Let $M \in \PHM{m}{k}{n}$ be a $k$-partite hypermatching and $Z \in \Z_q^{m \times k}$.
We define the instance $\Phi_{\phi,f,M,Z}$ of $\MaxCSP{\calF}$ on the variable-set $[n] \times [k']$,
whose constraints are the following:
For each row $j \in [m]$ such that $Z_j = 0$, create a constraint with predicate $f$ and variable sequence $e_j \coloneqq ((M_{j,1},\phi(1)),\ldots,(M_{j,k},\phi(k))) = (\phi(\iota_M(j,\ell)))_{\ell \in [k]}$.
\end{definition}

Note that only variables in $\phi(\supp(M))$ are used in the instance $\Phi_{\phi,f,M,Z}$
and all constraints in this instance use the same predicate ($f$).
Also, $q_0$, which appears in the domain of $f$, may differ from $q$, which appears in the domain of $Z$.

\begin{proposition}[{\cite[Proof of Lemma 3.1]{FMW26}}]\label{prop:onewise reduction}
    Let $q_0,k \in \N$ and $\calF \subseteq \{0,1\}^{[q_0]^k}$.
    Suppose that $\MaxCSP{\calF}$ admits a $(\gamma,\beta)$-integrality gap instance.
    Then there exists $q, k' \in \N$ such that for every $\epsilon > 0$, there exists $\alpha_0 \in (0,1)$ such that the following holds.
    
    For every $\alpha \in (0,\alpha_0)$,
    there exists $T \in \N$, injections $\phi_1,\ldots,\phi_T : [k] \to [k']$ corresponding to a gadget hypergraph $G \in \HG{T}{k'}{k}$,
    one-wise uniform distributions $\calY_1,\ldots,\calY_T \in \OWU{\Z_q^k}$,
    and predicates $f_1,\ldots,f_T : [q_0]^k \to \{0,1\}$ such that the following holds:
    Suppose we sample inputs $(\rM_1,\rZ_1),\ldots,(\rM_T,\rZ_T)$ from DIHP with parameters $G$, $\calY_1,\ldots,\calY_T$, $n$, and $m \coloneqq \alpha n$.
    Then, defining the random CSP instance
    $\rPhi \coloneqq \bigsqcup_{t=1}^T \Phi_{\phi_t,f_t,\rM_t,\rZ_t}$, we have:
    \begin{itemize}
        \item In the $\yes$ case, $\Pr[\optCSP{\rPhi} \ge \gamma-\epsilon] \ge 1-o_n(1)$.
        Further, if $\gamma=1$, then $\optCSP{\Phi} = 1$ deterministically.
        \item In the $\no$ case, $\Pr[\optCSP{\rPhi} \le \beta+\epsilon] \ge 1-o_n(1)$.
    \end{itemize}
\end{proposition}

\begin{proof}[Proof sketch]
    Since the coefficients of the basic LP are rational (see \Cref{fig:basic-lp}) and its feasible region is nonempty,
    there is a $(\gamma,\beta)$-integrality gap instance $\Phi_0$ for $\MaxCSP{\calF}$ in which all the local distributions $\calY^{0}_C \in \Dist{[q_0]^k}$ are rational.
    
    Let $\calV$ denote the variable set of this instance.
    For each $v \in \calV$, let $\calX_v \in \Dist{[q_0]}$ denote the marginal distribution on the variable $v$ which is consistent with all the local distributions.
    Take $q \in \N$ so that for every $v \in \calV$ and $a \in [q_0]$, $q \cdot \calX_v(a) \in \N$ (all marginal probabilities are multiples of $1/q$).

    For every variable $v \in \calV$, consider an arbitrary partition $\calP_v$ of $[q]$ into sets $S_1^{(v)}, \ldots, S_{q_0}^{(v)}$ with $|S_{a}^{(v)}| = q \cdot \calX_v(a)$ for all $a \in [q_0]$. 
    Let $\kappa_v : [q] \to [q_0]$ be the function mapping $S_a^{(v)}$ to $a$ for all $a \in [q_0]$, so that sampling $\rb \sim [q]$ uniformly,
    the random variable $\kappa_v(\rb)$ has distribution $\calX_v$.
    Observe that if we sample $\ra \sim \calX_v$ and then sample $\rb \sim \kappa_v^{-1}(\ra)$, the marginal distribution of $\rb$ is uniform on $[q]$.

    Let $C = (f,e)$ be a constraint in the basic LP integrality gap instance $\Phi_0$ (with $e = (i_1,\ldots,i_k)$).
    Let $\calY^0_C \in \Dist{[q_0]^k}$ denote the local distribution corresponding to $C$.
    We define a corresponding distribution $\calY_C \in \Dist{[q]^k}$
    by sampling $\ra = (\ra_1,\ldots,\ra_k) \sim \calY^0_C$
    and then outputting a random sample from $\kappa_{i_1}^{-1}(\ra_1) \times \cdots \times \kappa_{i_k}^{-1}(\ra_k)$.
    $\calY_C$ is one-wise uniform by the observation in the preceding paragraph 
    (and using the fact that for each $\ell \in [k]$, $\ra_{\ell}$ is distributed as $\calX_{i_\ell}$).
    Conversely, sampling $\rb = (\rb_1,\ldots,\rb_k) \sim \calY_C$, the random variable $(\kappa_{i_1}(\rb_1),\ldots,\kappa_{i_k}(\rb_k))$ has distribution $\calY^0_C$.

    We will set $k' \coloneqq |\calV|$ and eventually $T \coloneqq T_0 \cdot K$, where $T_0$ is the number of constraints in the original instance $\Phi_0$ and $K$ is a large enough constant (enough to ensure concentration below).
    The gadget hypergraph $G$ on $k'$ vertices is constructed by adding $K$ copies of $e$ for each constraint $C=(f,e)$ in $\Phi_0$, with all copies labeled with the same (one-wise uniform) distribution $\calY_C$ and corresponding to the same predicate $f$.
        
    \paragraph{The $\yes$ case analysis.}
    Given a signal $X \in \Z_q^{n \times k'}$, define $\kappa(X) \in [q_0]^{n \times k'}$ via applying the functions $\kappa_v$ appropriately,
    i.e., $(\kappa(X))_{i,v} \coloneqq \kappa_v(X_{i,v})$, which we view as an assignment to the CSP instance $\Phi$.
    We argue that in the $\yes$ case, even fixing the matchings $M_1,\ldots,M_T$, we have \[
    \frac{\Exp[\text{\# constraints in }\rPhi\text{ satisfied by }\kappa(\rX)]}
    {\Exp[\text{\# constraints in }\rPhi]} = \gamma, \]
    where the above expectation is over the choice of a random $X$.
    (To carry out the full argument of~\cite{FMW26}, we would then use concentration bounds over the choice of $\rY_1,\ldots,\rY_T$ and $\rX$ to get that the actual ratio is roughly $\gamma$ with high probability.)

    Recall that, for $t \in [T]$ and $j \in [m]$, an injection $\phi_t : [k] \to [k']$,
    a matching $M_t \in \PHM{m}{k}{n}$, and labeling $Z_t \in \Z_q^{m \times k}$, we inject the constraint $(f_t, (M_t)_j)$ into $\Phi$ iff $(Z_t)_j = 0$.
    In the $\yes$ case, $\rZ_t = \Pi^{\phi_t}_{M_t} \rX - \rY_t$, where $\rX \sim \Z_q^{n \times k}$ and $\rY_t \sim \calY_t^{\otimes m}$ independently.
    Hence, in particular, $(\rZ_t)_j = \Pi^{\phi_t}_{M_j} \rX + (\rY_t)_j$, where $\rX \sim \Z_q^{n \times k}$ and $(\rY_t)_j \sim \calY_t$ independently.
    Note that $\Pr[(\rZ_t)_j = 0] = q^{-k}$; this only uses the uniformity of $\rX$, and so is true even conditioned on any fixing of $(Y_t)_j$.
    Hence, by linearity, the expected number of constraints in $\rPhi$ is $q^{-k} \cdot mT$.
    On the other hand,
    \begin{align*}
        &\Pr[f_t(\kappa(\rX)_{(M_t)_{j,1},\phi_t(1)},\ldots,\kappa(\rX)_{(M_t)_{j,k},\phi_t(k)}) = 1 \wedge (\rZ_t)_j = 0]\\
        ={}& \Pr[f_t(\kappa_1(\rX_{(M_t)_{j,1},\phi_t(1)}),\ldots,\kappa_k(\rX_{(M_t)_{j,k},\phi_t(k)})) = 1 \mid (\rZ_t)_j = 0] \cdot \Pr[(\rZ_t)_j = 0].
        \end{align*}
    The second factor is again $q^{-k}$.
    For the first factor, we use that $(X_{(M_t)_{j,1},\phi_t(1)},\ldots,X_{(M_t)_{j,k},\phi_t(k)}) = \Pi^{\phi_t}_{(M_t)_j} X$
    and conditioned on $(\rZ_t)_j = 0$, the marginal distribution of $\Pi^{\phi_t}_{(M_t)_j} X$ is $\calY_C$.
    Hence, the marginal distribution of $\kappa_1(\rX_{(M_t)_{j,1},1}),\ldots,\kappa_k(\rX_{(M_t)_{j,k},k})$ is uniform on $\calY_C^0$,
    so that the first factor is just $\Exp_{\ra \sim \calY^0_C}[f_t(\ra)]$.
    Now summing over (all $K$ copies of) all $T_0$ constraints $C$, we recover the expression $q^{-k} \cdot \sum_{C=(f,e) \in \Phi} \Exp_{\ra \sim \calY^0_C} [f(\ra)]$, which is exactly the (unnormalized) LP objective.

    \paragraph{The $\no$ case analysis.}
    For the $\no$ case, we can use that $\rPhi$ is essentially a random blowup of the base instance $\Phi_0$.
    We can therefore show that the expected value of every fixed assignment to $\rPhi$ is at most $\beta$ 
    (since any such assignment projects down to an assignment to $\Phi_0$),
    then use union+concentration bounds (assuming $\alpha_0$ is sufficiently small and $K$ sufficiently large).
\end{proof}

We are now ready to prove our main theorem, as restated below.
\mainthm*
\begin{proof}
  Consider any $(\gamma,\beta)$-LP integrality gap instance of $\MaxCSP{\cF}$. Suppose for contradiction that there exists an $\epsilon_0>0$ and a streaming algorithm $\calA$ for $\GapMaxCSP{\gamma-\epsilon_0}{\beta+\epsilon_0}{\calF}$ that requires only $o(n)$ space. Using \cref{prop:onewise reduction}, we will derive a $o(n)$-bit communication protocol for DIHP, contradicting \cref{thm:main DIHP}. To apply \cref{prop:onewise reduction}, we set $(q_0,k,\epsilon)_{\cref{prop:onewise reduction}}= (q,k,\epsilon_0)$. For these parameters, consider the DIHP instance produced in \cref{prop:onewise reduction} and the corresponding random CSP instance $\rPhi \coloneqq \bigsqcup_{t=1}^T \Phi_{\phi_t,f_t,\rM_t,\rZ_t}$. Since the reduction described in \cref{def:reduction} is local, each player $\Player_t$ can independently sample $\Phi_{\phi_t,f_t,\rM_t,\rZ_t}$. Starting from $\Player_1$, the players can sequentially run $\calA$ on their local instance and pass the memory state to the next player; specifically, starting from the memory state passed by $\Player_{t-1}$, player $\Player_t$ runs $\calA$ on $\Phi_{\phi_t,f_t,\rM_t,\rZ_t}$. Finally, $\Player_T$ outputs $\yes$ if $\calA$ categorizes the instance as having value at least $\gamma+\epsilon_0$, and $\no$ otherwise. Correctness follows immediately from \cref{prop:onewise reduction}. Since $\calA$ uses only $o(n)$ memory, this yields a $o(n)$-bit communication protocol for DIHP. The case $\gamma=1$ is analogous.
\end{proof}

\section{Lower bound on the communication complexity}\label{sec:communication_lb}

\newcommand{\Post}[5]{\mathrm{Post}_{#2,#3}(#1;#4,#5)}
\newcommand{\PostShort}[3]{\mathrm{Post}(#1;#2,#3)}
\newcommand{\Input}[4]{\mathrm{Input}_{#1,#2}(#3,#4)}
\newcommand{\InputShort}[2]{\mathrm{Input}(#1,#2)}

\newcommand{\rcalB}{\boldsymbol{\calB}}
\newcommand{\rcalD}{\boldsymbol{\calD}}

In this section, we prove the main communication lower bound for DIHP, as restated below.

\communicationlb*

By Yao's minimax principle \cite{Yao77}, it suffices to prove such a lower bound against all \emph{deterministic} protocols for DIHP. 
Namely, a protocol for DIHP can be specified by deterministic message functions $\Player_1,\dots,\Player_T$ so that $S_t= \Player_t(M_{1:t},S_{1:t-1},Z_t)$ denotes the message sent by the $t$-th player on input $(M_{1:t},S_{1:t-1},Z_t)$.
Without loss of generality, the final function $\Player_T$ outputs a simple bit ``$\yes$/$\no$'' indicating the output of the protocol.

\subsection{Informal overview of proof}

To prove \Cref{thm:main DIHP}, we consider the following ``hybrid experiment'': Independently sample a common sequence of random $k$-partite hypermatchings $\rM_{1:T}$, a hidden signal $\rX \in \Z_q^{n \times k'}$, and then sample
$\yes$ inputs $\rZ_{1:T}^Y$ (i.e., $\rZ^Y_t = \Pi^{\phi_t}_{\rM_t}\rX^* - \rY_t$ for hidden signal for  $\rY_t \sim \calY_t^{\otimes k}$) and $\no$ inputs $\rZ_{1:T}^N$ (uniformly random).
Note that $(\rM_{1:T},\rZ_{1:T}^Y)$ have the marginal distribution of a $\yes$ input,
and $(\rM_{1:T},\rZ_{1:T}^N)$ have the marginal distribution of a $\no$ input.
Thus, this experiment is a \emph{coupling} between the $\yes$ and $\no$ input distributions.

In the above experiment, let $\rS_{1:T}^Y$ denote the messages sent in the $\yes$ case (so that $\rS^Y_t = r_t(\rM_{1:t},\rS_{1:t-1}^Y,\rZ^Y_t)$) and $\rS_{1:T}^N$ denote the messages in the $\no$ case
(so that $\rS_t^N = r_t(\rM_{1:t},\rS_{1:t-1}^N,\rZ^N_t)$).
To prove \Cref{thm:main DIHP} we need to show that the distributions of the random variables $\rS^Y_T$ and $\rS^N_T$ are close in total variation distance.
For the induction, we prove the much stronger statement that the distributions of $(\rM_{1:T},\rS_{1:T}^Y)$ and $(\rM_{1:T},\rS_{1:T}^N)$ are close in total variation distance, i.e.,
\[
\|(\rM_{1:T},\rS_{1:T}^Y)-(\rM_{1:T},\rS_{1:T}^N)\|_{tvd}\leq\delta \, .
\]

We prove the above bound via a hybrid lemma analogous to those in \cite{CGS+22-linear-space, KK19}. Roughly speaking, this lemma states that if the first $t-1$ players' inputs are drawn from the $\yes$ distribution, then the $t$-th player's output on a $\yes$ input is distributed very similarly to their output on a $\no$ input, even conditioned on all previous hypermatchings and messages.
Formally, the lemma identifies a \emph{filtration} (nested sequence) of events $\cE_1\supset\cE_2\supset\cdots\supset\cE_T$ in the $\yes$ experiment such that: (i) $\cE_t$ enforces a typicality condition on the $k$-partite hypermatchings received by the first $t$ players and the messages of the first $t-1$ players; and (ii) whenever these random variables take typical values, player $\Player_t$ cannot distinguish whether their input is sampled from the $\yes$ distribution or the $\no$ distribution, assuming all previous players' inputs were drawn from the $\yes$ distribution.

\subsection{Formal statement of hybrid lemma}

The probability space underlying the hybrid lemma (\Cref{lem:hybrid}) will be the following distribution.
Let $\Omega^Y \coloneqq \Omega_{k,q,\alpha,n,T,G,\calY_{1:T}}$ denote the distribution over tuples $(\rX^*,\rM_{1:T},\rZ_{1:T})$ where $\rX^*\sim \Unif(\Z_q^n)$, $\rM_t \in \PHM{m}{k}{n}$ is a $k$-partite hypermatching with $m$ edges, and $\rZ_t \in \Z_q^{m\times k}$ is a noisy signal, as defined in the $\yes$ case of \Cref{def:DIHP} for every $t \in [T]$, i.e.,
for each $t \in [T]$, we sample $\rY_t \sim \calY_t^{\otimes m}$ and set $\rZ_t \coloneqq \Pi^{\phi_t}_{\rM_t} \rX^* - \rY_t$.

These variables along with a deterministic protocol given by $\Player_1,\ldots,\Player_T$ specify additional random variables that are determined by $(\rX^*,\rM_{1:T},\rZ_{1:T})$ including $\rcalD_{1:t}$, $\rS^Y_t$.

\begin{restatable}[Hybrid lemma]{lemma}{hybrid}\label{lem:hybrid}
	For every $q,k\in\N$, there exists $\alpha_0>0$ such that for every $T\in\N$, and $\delta\in(0,1)$, there exists $\theta\in(0,1)$ and $n_0 <\infty$  such that  the following holds for every $n \geq n_0$:
	
    Let $\Pi=(\Player_t)_{t \in [T]}$ be a deterministic protocol for DIHP where $m\le \alpha_0 n$ and each message function $r_t$ outputs a message of at most $\theta n$ bits. Let $(\rX^*,\rM_{1:T},\rZ_{1:T}) \sim \Omega$. 
	Then there exists a sequence of events $\{\cE_t\}_{t\in[T]}$ such that:
    \begin{enumerate}
        \item For $t \in [T]$,  $\cE_t$ only depends on $\rM_{1:t}$ and $\rS^Y_{1:t-1}$ (with $\rS^Y_{1:0}$ denoting an empty set of variables).
        \item For every $t \in [T]$,  $\cE_{t} \Rightarrow \cE_{t-1}$ and $\Pr[\overline{ \cE_{t}}\, |\, \cE_{t-1}]\leq (\delta/(2T))$ (where $\calE_0$ is the trivial event occurring with probability $1$).
        \item For every fixed $M_{1:t}$ and $S^Y_{1:t-1}$ satisfying $\cE_t$, one has
	\begin{equation}\label{eq:hybrid}
		\|\rS^Y_t-\Player_t(M_{1:t},S^Y_{1:t-1},\rU)\|_{\tvd}\leq\delta/(2T),
	\end{equation}
	where $\rU\sim\Unif(\Z_q^{m\times k})$.
    \end{enumerate}
\end{restatable}    

Using \Cref{lem:hybrid}, we can prove \Cref{thm:main DIHP} using a simple inductive argument. The proof is analogous to the proof of Lemma 6.3 in \cite{KK19} and Theorem 3.5 in \cite{CGS+22-linear-space}.
So we skip this proof here and include it in \Cref{proof:thm:main DIHP} for completeness sake.
\Cref{lem:hybrid} is proved in \Cref{sec:hybrid} below.

\subsection{Informal proof of \Cref{lem:hybrid}}

We prove \Cref{lem:hybrid} formally in \Cref{sec:hybrid}. Here, we give a brief overview of the proof along with the definitions and lemmas that we will need.
At a high level, the main idea is to show that for every $t\in [T]$, after the first $t$ rounds of communication, sufficient randomness remains in $\rX^*$ so that $\Player_{t+1}$ cannot determine whether their input is drawn from the $\yes$ or the $\no$ distribution.
As in \cite{CGS+22-linear-space,KK19}, we formally capture the randomness in $\rX^*$ through the so-called \emph{boundedness} of its posterior distribution, conditioned on the hypermatchings and the messages of the first $t$ players. The boundedness condition (stated formally in \Cref{def:bounded distribution}) bounds the $\ell_1$-norm of the Fourier coefficients at each Hamming weight. The subsequent \emph{boundedness implies uniformity} lemma (\Cref{lemma:boundedness implies uniformity}, proved in \Cref{sec:boundedness implies uniformity}) then implies that if the posterior distribution of $\rX^*$ after $t$ rounds of communication is bounded, the total variation distance between $\Player_{t+1}$'s inputs under the $\yes$ and $\no$ distributions is very small.

The next step is to inductively show that if the posterior distribution of $\rX^*$ is bounded after the first $t$ rounds of communication (roughly captured by event $\cE_t$), then it remains bounded after the $(t+1)$-th round. This is the step where we depart significantly from prior work \cite{CGS+22-linear-space,KK19}.
In previous works, the players' inputs in the $\yes$ case were noiseless, i.e., they received the induced labelings $\Pi_{M_t}^{\phi_t} X$ directly.%
\footnote{In \cite{CGS+22-linear-space}, the authors specifically consider a ``shifted noise'' distribution that applies a uniform shift to all coordinates; via a recentering argument, they reduce this to the noiseless case.}
Consequently, the posterior distributions in those works were always uniform over a set, and showing that this set is large and well-structured sufficed to establish boundedness of the posterior distribution.
The analysis in our setting is more intricate because the posterior distribution of $\rX^*$ is no longer uniform over a set. While the strong structural guarantees enjoyed by the posterior sets in previous works do not hold here, we show that weaker properties suffice to carry out the induction.
The inductive step lemma is stated formally as \Cref{lem:inductive_step} and proved in \Cref{sec:inductive_step}. Assuming \Cref{lemma:boundedness implies uniformity,lem:inductive_step}, we formally prove \Cref{lem:hybrid} in \Cref{sec:hybrid}.

\subsection{Definitions and statements of key lemmas}

\begin{definition}[(One-step) posterior distribution]\label{def:posterior distribution}
    Let $n,k,k',m \in \N$, $\phi : [k] \to [k']$ be an injection, and $\calY \in \OWU{\Z_q^k}$ a one-wise uniform distribution.
    For a ``prior'' distribution $\calD \in \Dist{\Z_q^{n \times k'}}$, a $k$-partite hypermatching $M \in \PHM{m}{k}{n}$, and a set $\calB \subseteq \Z_q^{m \times k}$,
    we define the \emph{posterior distribution} $\Post{\calD}{\phi}{\calY}{M}{\calB}$
    as the conditional distribution on $\rX \in \Z_q^{n \times k'}$
    when we sample $\rX \sim \calD$ and $\rY \sim \calY^{\otimes m}$ independently
    and condition on the event that $\Pi^\phi_M \rX - \rY \in \calB$.
    The density function of this distribution is
    \[ \Den{\Post{\cD}{\phi}{\calY}{M}{\calB}}(X_0) = q^{n \cdot k'} \cdot \Pr_{\substack{\rX \sim \calD, \\ \rY \sim \calY^{\otimes m}}}
    \bracks*{ \rX = X_0 \mid \Pi^\phi_M \rX - \rY \in \calB }. \]
    We omit $\phi$ and $\calY$ and write $\PostShort{\calD}{M}{\calB}$ when clear from context.
\end{definition}

\begin{definition}[Input distribution]\label{def:input distribution}
    Let $n,k,k',m \in \N$, $\phi : [k] \to [k']$ an injection, and $\calY \in \OWU{\Z_q^k}$ a one-wise uniform distribution.
    For a distribution $\calD \in \Dist{\Z_q^{n \times k'}}$ and a $k$-partite hypermatching $M \in \PHM{m}{k}n$, we define the \emph{input distribution} $\Input{\phi}{\calY}{\calD}{M}$ as the distribution of $\Pi^\phi_M \rX - \rY$ when $\rX \sim \calD$ and $\rY \sim \calY^{\otimes m}$ independently.
    We omit $\phi$ and $\calY$ and write $\InputShort{\calD}{M}$ when clear from context.
\end{definition}

\newcommand{\U}[4]{\mathrm{U}_{#1,#2,#4}(#3)}

\begin{definition}[Fourier bound function]\label{def:bounded distribution}
A distribution $\cD \in \Dist{\Z_q^N}$ is \emph{$(C,s)$-bounded} if, for every $h \in [N]$, \[
\sum_{\substack{u\in\Z_q^N,\\\wt(u)=h}} \abs*{ \hat{\Den{\calD}}(u) } \le \U{C}{s}{h}{N}, \]
where for $C < \infty$ and $s,h,N \in \N$, we define: \[
	\U{C}{s}{h}{N} \coloneqq \begin{cases}
    1, & \text{if } h=0\\
		\parens*{ \frac{C\sqrt{s N}}{h}}^{h/2}, & \text{if } 1\le h\le s \\
		\min\braces*{ \parens*{ \frac{C\sqrt{N}}{\sqrt{h}}}^{h/2}, \parens*{ \frac{eq^2 N}{h}}^{h/2} }, & \text{if } h>s.
	\end{cases} \qedhere \]
\end{definition}
\begin{observation}[Monotonicity of boundedness with respect to $C$]\label{obs:boundedness monotonicity}
    The function $\U{C}{s}{h}{N}$ increases monotonically with respect to the parameter $s$.
    Hence, if a distribution $\cD \in \Dist{\Z_q^N}$ is $(C,s)$-bounded, then for every $C'\ge C$, $\calD$ is also $(C',s)$-bounded.
\end{observation}

\begin{restatable}[Boundedness implies uniformity]{lemma}{boundedness}
\label{lemma:boundedness implies uniformity}
    For every $k,q,k' \ge 2 \in \N$ ($k' \ge k$), there exists $\alpha_0 > 0$ such that for every $\delta \in (0,1/2)$ and $C < \infty$, there exists $\tau > 0$ such that the following holds.

    For every injection $\phi : [k] \to [k']$, 
    one-wise uniform distribution $\cY\in \OWU{\Z_q^k}$,
    $n \in \N$ sufficiently large,
    $\calD\in \Dist{\Z_q^{n \times k'}}$ a $(C,s)$-bounded distribution satisfying $\|\Den{\calD}\|_\infty \le q^b$ and $4\log_q(3/\delta) \le b \le s \le \tau n$,
    $m \le \alpha_0 n$, \[
    \Pr_{\rM \in \PHM{m}{k}{n}} \bracks*{ \|\Den{\InputShort{\calD}{\rM}} - 1\|_\infty \le \delta }~\ge~1-\delta, \]
    where $\rM \in \PHM{m}{k}{n}$ is a randomly sampled matching.
\end{restatable}

Observe that the parameter $b$ in \Cref{lemma:boundedness implies uniformity} can be eliminated, if the inequality is replaced with  $4\log_q \max\{3/\delta,\|\Den{\calD}\|_\infty\} \le s \le \tau n$.

\begin{observation}\label{obs:boundedness implies uniformity}
    For $\calZ \in \Delta(\Z_q^{m \times k})$,
    the condition $\|\Den{\calZ} - 1\|_\infty \le \delta$ is equivalent to that for every $Z_0 \in \Z_q^{m \times k}$, it holds that $1-\delta \le \Den{\calZ}(Z_0) \le 1+\delta$.
    In this case, we also have $\norm{ \calZ - \Unif(\Z_q^{m \times k}) }_{\mathrm{tvd}} \le \delta$,
    and for every function $f : \Z_q^{m \times k} \to \R$,
    $\abs*{ \Exp_{\rZ \sim \calZ}[f(\rZ)] - \Exp_{\rZ\in\Z_q^{m \times k}}[f(\rZ)] } \le \delta \|f\|_1$.
\end{observation}

\begin{restatable}[Inductive step]{lemma}{inductivestep}\label{lem:inductive_step}
    For every $k,q,k' \in \N$ there exist $\alpha_0 > 0$ and $C_0 < \infty$ such that for every $C \ge C_0$ and $\delta \in (0,1/2)$ there exist $\sigma \in (0,1)$ and $C'' > 0$ such that the following holds.

    For every injection $\phi : [k] \to [k']$ and one-wise uniform distribution $\calY \in \OWU{\Z_q^k}$,
    for every $n,b,b',s,s',m \in \N$ satisfying $m \le \alpha_0 n$, $0 < b,b',s
    < \sigma n$ and $b + b' + \log_q(1-1/\delta) \le s$ and every $(C,s)$-bounded distribution $\calD \in \Dist{\Z_q^{n \times k'}}$ satisfying $\|\Den{\calD}\|_\infty \le q^b$, we have that \[
    \Pr_{\rM\in\PHM{m}{k}{n}} \bracks*{
    \begin{aligned}
    &\forall \calB \subset \Z_q^{m \times k} \text{ such that } |\calB| \ge q^{m\cdot k-b'} , \\
    &\hspace{0.5in} \PostShort{\calD}{\rM}{\calB}\text{ is }(C'',s)\text{-bounded and } \|\Den{\PostShort{\calD}{\rM}{\calB}}\|_\infty \le \frac1{1-\delta} \cdot q^{b+b'}
    \end{aligned}}
    \ge 1-5\delta\, \]
    where $\rM \in \PHM{m}{k}{n}$ is a randomly sampled $k$-partite hypermatching.
\end{restatable}

\subsection{Hybrid lemma: Proof of \Cref{lem:hybrid}}\label{sec:hybrid}

In this section, we formally prove \Cref{lem:hybrid}.

\begin{proof}[Proof of \Cref{lem:hybrid}]
Let $\Player_1,\dots,\Player_T$ be the message functions corresponding to some deterministic protocol for DIHP.
Recall that the input of $\Player_t$ (in addition to the prior message) is a $k$-partite hypermatching $\rM_t$ and a noisy signal $\rZ_t$.
(We also have fixed $\phi_t, \cY_t$, which are, respectively the injective map and the one-wise uniform distribution corresponding to $\Player_t$.)

\paragraph{Posterior set and distribution.}

We define some auxiliary random variables determined by the randomness of $(\rM_{1:t},\rS^Y_{1:t})$.
For $t \in [T]$, we define:
\begin{align*}
\rcalB_t &\coloneqq \{Z \in \Z_q^{m\times k} : \rS^Y_t = \Player_t(\rM_{1:t},\rS^Y_{1:t-1},Z)\}
\intertext{and, recursively:}
\rcalD_t &\coloneqq \Post{\rcalD_{t-1}}{\phi_t}{\calY_t}{\rM_t}{\rcalB_t}.
\end{align*}
where $\rcalD_0 = \calD_0 = \Unif\{\Z_q^{n\times k'}\}$.
$\rcalB_t$ denotes the set of all inputs $Z_t$ that would
result in the message $\rS^Y_t$ 
(conditioned on the previous messages $\rS^Y_1,\dots,\rS^Y_{t-1}$ and the hypermatchings $\rM_1,\dots,\rM_t$),
while $\rcalD_t$ denotes the posterior distribution on $\rX^*$ conditioned on $\rM_{1:t}$, $\rS^Y_{1:t}$ in the following formal sense:
\begin{claim}\label{clm:posterior_distribution}
    For every $0\le t\le T$, $X_0 \in \Z_q^{n\times k'}$, and fixed $M_{1:t}$, $S^Y_{1:t}$, we have \[
    \Pr_\Omega[\rX^*=X_0 \mid M_{1:t}, S^Y_{1:t}] = \calD_t(X_0) \]
    (here $M_{1:t}$ and $S^Y_{1:t}$ are both tuples of length $t$,
    so that $M_{1:0}, S^Y_{1:0}$ denote empty tuples).
	In particular, for every fixed $M_{1:t}$ and $S_{1:t-1}^Y$, the marginal distribution of $\rS^Y_t$ is the same as the marginal distribution of  $\Player_t(M_{1:t},S_{1:t-1}^Y,\rZ)$, where $\rZ=\Pi^{\phi_t}_{M_t} \rX - \rY_t$ with $\rX\sim\calD_{t-1}$ and $\rY_t \sim \cY_t^{\otimes m}$.
\end{claim}

\begin{proof}
  We will prove by induction. The base case when $t=0$ is trivially satisfied since $\cD_0$ is the uniform distribution.
  Inductively, we have:
   \begin{align*}
  &\Pr_\Omega [\rX^*=X_0 \mid M_{1:t}, S^Y_{1:t}] \\
  &= \Pr_{\substack{\rX\in\Z_q^{n\times k'}, \\ \forall t'\in[t],~\rY_{t'}\sim \cY_{t'}^{\otimes m}}}[\rX=X_0 \mid  \forall t'\in [t], ~S^Y_{t'} = \Player_{t'}(M_{1:t'},S^Y_{1:t'-1},\Pi^{\phi_{t'}}_{M_{t'}}\rX - \rY_{t'})] \\
  &= \Pr_{\substack{\rX\sim  \calD_{t-1}, \\ \rY_{t}\sim \cY_{t}^{\otimes m}}}[\rX=X_0 \mid S^Y_{t} = \Player_{t}(M_{1:t},S^Y_{1:t-1},\Pi^{\phi_{t}}_{M_{t}}\rX - \rY_{t})] \tag{induction hypothesis} \\
    &= \Pr_{\substack{\rX\sim  \calD_{t-1}, \\ \rY_{t}\sim \cY_{t}^{\otimes m}}}[\rX=X_0 \mid \Pi^{\phi_{t}}_{M_{t}}\rX - \rY_{t}\in \calB_t]
    = \calD_t(X_0)\, . \qedhere
\end{align*}
\end{proof}

\paragraph{Parameters.} The required parameters that we need to specify for \Cref{lem:hybrid} are $\alpha_0, \theta,n_0$, and $\delta'$. In addition to setting these parameters below, we also set the following parameters required in the proof: (i) $\{C_t\}_{0\le t\le T}$ and $s$ to quantify the boundedness of the posterior distributions, (ii) $s^*$ to upper bound the length of each message, and (iii) $b$ to upper bound the $\ell_{\infty}$ norms of the density functions of the posterior distributions, which is required to apply \Cref{lemma:boundedness implies uniformity,lem:inductive_step}.

Given $k,q$, and $k'$, let $\alpha_{0,1}$ be the $\alpha_0(k,q,k')$ parameter from \Cref{lemma:boundedness implies uniformity} and $\alpha_{0,2}$ be the $\alpha_0(k,q,k')$ parameter from \Cref{lem:inductive_step}. We set $\alpha_0$ to be $\min\{\alpha_{0,1},\alpha_{0,2}\}$. Let $C_0$ be the $C_0(k,q,k')$ parameter from \Cref{lem:inductive_step}.
Now, given $T$ and $\delta$, we need to specify $\tau>0$, $n_0<\infty$, and $\delta'$.
Along the way, we also specify the intermediary parameters that we mentioned above. Observe that \Cref{lem:inductive_step} takes as input parameters $k,q,k'\delta$, and $C \geq C_0$, and specifies parameters $C'' = C''(k,q,k',\delta,C)$ and $\sigma=\sigma(k,q,k',\delta,C)$ for which the lemma holds. Observe also that \Cref{lemma:boundedness implies uniformity} takes as input parameters $k,q,k'\delta$, and $C \geq C_0$, and specifies $\tau=\tau(k,q,k',\delta,C)$ for which the lemma holds.
We set $\delta'\coloneqq \delta/(30T)$.
We define $\{C_t\}_{0\le t\le T}$ recursively as $C_t \coloneqq \max\{C_0,C''(k,q,k',\delta',C_{t-1})\}$.
For $t\in [T]$, let $\sigma_t \coloneqq  \sigma(k,q,k',\delta',C_t)$, $\tau_t \coloneqq  \tau(k,q,k',\delta',C_t)$, and $\gamma \coloneqq  \min\{\{\sigma_t\}_{t\in [T]},\{\tau_t\}_{t\in [T]}\}$.
We set $s\coloneqq \gamma n$, $b\coloneqq s/2T$, $b'\coloneqq b$, $\theta\coloneqq \gamma/8T$, and $n_0\coloneqq (1/\gamma) \cdot \max\{4\log_q(3/\delta'),8T\log_q(1/(1-\delta'))\}$. Let $m,n\in \mathbb{N}$ be such that $n\ge n_0, m\le \alpha_0 n$.
We set $s^*\coloneqq \theta n$.

\paragraph{Events.}
\newcommand{\ePost}{\mathrm{post}}
\newcommand{\eIn}{\mathrm{in}}

We define the events $\cE_1,\dots,\cE_T$ (in the probability space $\Omega^Y$) as follows.
For $t\in [T]$, we define $\cE_t^\ePost$ to be the event that the distribution $\rcalD_{t-1}$ is $(C_t,s)$-bounded and $\|\Den{\rcalD_{t-1}}\|_\infty\le q^{2(t-1)b}$;
$\cE_t^\eIn$ to be the event that $\|\Input{\phi_t}{\calY_t}{\rcalD_{t-1}}{\rM_t} - \Unif(\Z_q^{m\times k})\|_\tvd\le \delta'$;
and $\cE_t \coloneqq \cE_{t-1}\cap \cE_t^\ePost\cap \cE_t^\eIn$, where $\calE_0$ is the trivial event occuring with probability $1$.

\paragraph{Condition 1.} It immediately follows from the above definitions that for all $t\in [T]$, $\cE_t$ is fully determined by $\rM_{1:t}$ and $\rS^Y_{1:t-1}$.
(Further, $\calE^\ePost_t$ itself is determined only by $\rM_{1:t-1}$ and $\rS^Y_{1:t-1}$.)

\paragraph{Condition 2.}
It follows also from the definitions that $\cE_{t}\implies \cE_{t-1}$, for all $t \in [T]$.
We upper bound the failure probabilities of the events using induction.
For the base case, we first prove that $\Pr[\cE_1]= 1$. Since $\cD_0$ is the uniform distribution, $\hat{\Den{\cD_0}(u)}=1$ if $u=0^{n\times k'}$, and $\hat{\Den{\cD_0}(u)}=0$ otherwise.
Therefore, $\cD_0$ is trivially $(C_0,s)$-bounded. 
We also have $\|\Den{\cD_0}\|_\infty\le 1$.
For every $M\in  \PHM{m}{k}{n}$ and injection $\phi:[k]\to [k']$, $\Pi^\phi_M \rX$ is uniformly distributed in $\Z_q^{m \times k}$, since $\rX\sim \cD_0$.
Therefore, $\Pi^\phi_M \rX - \rY$ is also uniformly distributed for $\rY\sim \cD$ for any distribution $\cD\in \Delta(\Z_q^{m \times k})$.
Thus,  $\cE_1$ holds with probability $1$. 
 
We now prove that for any $t\in [T-1]$, $\Pr[\overline{\cE_{t+1}}\mid \cE_{t}]\le \delta/(2T)$.
Specifically, we will show 
\begin{align}
&\Pr\bracks*{ \overline{\cE_{t+1}^\eIn}\mid \cE_{t+1}^\ePost \wedge \cE_t }\le \delta' \,, \label{eqn:err1}  \\
&\Pr\bracks*{ \overline{\cE_{t+1}^\ePost}\mid \cE_t^\ePost \wedge \calE_{t-1} } \le 7\delta'\, .\label{eqn:err2}
\end{align}
We proceed modulo these. 
Noting that \[
    \Pr\bracks*{ \calE_{t+1} \mid \calE_t }
    = \Pr\bracks*{ \calE^\ePost_{t+1} \wedge \calE^\eIn_{t+1} \mid \calE_t }
    = \Pr\bracks*{ \calE^\eIn_{t+1} \mid \calE^\ePost_{t+1} \wedge \calE_t } \cdot \Pr\bracks*{ \calE^\ePost_{t+1} \mid \calE_t }, \]
the elementary inequality $1-ab \le (1-a) + (1-b)$ (for $a, b \le 1$) implies
\begin{equation}\label{eq:elem}
    \Pr\bracks*{ \overline{\cE_{t+1}} \mid \cE_t } \le \Pr\bracks*{ \overline{\cE_{t+1}^\eIn}\mid \calE^\ePost_{t+1} \wedge \cE_t } + \Pr\bracks*{ \overline{\cE_{t+1}^\ePost}\mid \cE_t }.
\end{equation}
We can upper-bound the second term on the RHS of \Cref{eq:elem} by (using that $\calE_t \implies \calE_{t-1}$):
\begin{multline*}
    \Pr\bracks*{ \overline{\cE_{t+1}^\ePost}\mid \cE_t \wedge \calE_{t-1} }
    = \frac{\Pr\bracks*{ \overline{\cE_{t+1}^\ePost} \wedge \cE_t \mid \calE_{t-1} }}{\Pr[\calE_t \mid \calE_{t-1}]}
    \le \frac{\Pr\bracks*{ \overline{\cE_{t+1}^\ePost} \wedge \calE_t^\ePost  \mid \calE_{t-1}}}{\Pr[\calE_t \mid \calE_{t-1}]} \\
    = \frac{\Pr\bracks*{ \overline{\cE_{t+1}^\ePost} \mid \calE_t^\ePost \wedge \calE_{t-1} } \cdot \Pr\bracks*{ \calE_t^\ePost \mid \calE_{t-1} }}{\Pr[\calE_t \mid \calE_{t-1}]}
    \le \frac{\Pr\bracks*{ \overline{\cE_{t+1}^\ePost} \mid \calE_t^\ePost \wedge \calE_{t-1} }}{\Pr[\calE_t \mid \calE_{t-1}]}
\end{multline*}
where the first inequality used $\calE_t \implies \calE_t^\ePost$.
Now, we conclude using \Cref{eq:elem,eqn:err1,eqn:err2}:
\[
\Pr\bracks*{ \overline{\cE_{t+1}} \mid \cE_t } \le \delta' + (7\delta')/(1 - \delta)
\le 15\delta' \, \le \delta/(2T), \]
proving Condition 2.

\paragraph{Proof of \Cref{eqn:err1}.}
We prove that for every fixing of $M_{1:t},S^Y_{1:t}$ such that $\cE_{t+1}^\ePost$ and $\cE_t$ hold, the event $\cE_{t+1}^\eIn$ is likely.
The probability is over the choice of $\rM_{t+1}$, the one additional variable needed to determine $\calE_{t+1}^\eIn$;
the conditional distribution of $\rM_{t+1}$ is still uniformly random.

By definition of $\cE_{t+1}^\ePost$, we have that $\calD_t$ is $(C_t,s)$-bounded and $\|\Den{\calD_t}\|_\infty\le q^{2tb}$.
Therefore, applying the boundedness implies uniformity lemma (\Cref{lemma:boundedness implies uniformity}) with parameters $(\delta,C,\tau,\cD,b)_{\rm{\Cref{lemma:boundedness implies uniformity}}}=(\delta',C_t,\gamma,\calD_t,2tb)$ and \Cref{obs:boundedness implies uniformity}, we get \[
\Pr_{\rM_{t+1}} \bracks*{\|\Input{\phi_t}{\calY_t}{\calD_t}{\rM_{t+1}} - \Unif(\Z_q^{m\times k})\|_\tvd\le \delta'}\ge 1 -\delta'. \]

\paragraph{Proof of \Cref{eqn:err2}.}
We prove that for every fixing of $M_{1:t-1},S^Y_{1:t-1}$ such that $\calE^\ePost_t$ and $\calE_{t-1}$ holds,
the event $\calE^\ePost_{t+1}$ is likely.
The probability is over the choice of $\rM_t,\rX^*,\rZ_t$, and the conditional marginal distribution of $\rM_t$ is still uniformly random.

By definition of $\calE^\ePost_t$, $\calD_{t-1}$ is $(C_{t-1},s)$ bounded and $\|\Den{\calD_{t-1}}\|_\infty \le q^{2(t-1)b}$.
Applying the inductive step (\Cref{lem:inductive_step}) with parameters $(C,\delta,\sigma,C'',\cD,b)_{\rm{\Cref{lem:inductive_step}}} = (C_{t-1},\delta',\gamma,C_t,\calD_{t-1},2(t-1)b)$, we have \[
    \Pr_{\rM_t} \bracks*{
    \begin{aligned}
    &\forall \calB \subset \Z_q^{m \times k} \text{ such that } |\calB| \ge q^{m \cdot k-b'} , \\
    &\hspace{0.5in} \PostShort{\calD_{t-1}}{\rM_t}{\calB}\text{ is }(C_t,s)\text{-bounded and } \|\Den{\PostShort{\calD_{t-1}}{\rM_t}{\calB}}\|_\infty \le q^{2tb}
    \end{aligned}} 
    \ge 1-5\delta'\, , \]
where we used the fact for our choice of parameters $b,b'$, and $\delta'$, $q^{2(t-1)b+b'+\log_q(1/(1-\delta'))} \le q^{2tb}$.

Applying the boundedness implies uniformity lemma (\Cref{lemma:boundedness implies uniformity}) with the choice of parameters $(\delta,C,\tau,\cD,b)_{\rm{\Cref{lemma:boundedness implies uniformity}}}=(\delta',C_{t-1},\gamma,\calD_{t-1},2(t-1)b)$,
we have \[
\Pr_{\rM_t} \bracks*{ \|\Den{\InputShort{\calD_{t-1}}{\rM_t}} -1 \|_\infty \le \delta'}~\ge~1-\delta'. \]
Therefore, with probability at least $1-6\delta'$ (over the randomness of $\rM_t$), we have
\begin{equation}\label{eqn:good matching}
    \Pr_{\rM_t} \bracks*{
    \begin{aligned}
    &\forall \calB \subset \Z_q^{m \times k} \text{ such that } |\calB| \ge q^{m \cdot k-b'} , \\
    &\hspace{0.5in} \PostShort{\calD_{t-1}}{\rM_t}{\calB}\text{ is }(C_t,s)\text{-bounded and } \|\Den{\PostShort{\calD_{t-1}}{\rM_t}{\calB}}\|_\infty \le q^{2tb} \\
    &\text{and } \|\Den{\InputShort{\calD_{t-1}}{\rM_t}} -1\|_\infty \le \delta'
    \end{aligned}} 
    \ge 1-6\delta'.
\end{equation}

Now \emph{fix} such an $M_t$.
By \Cref{clm:input fiber is large} below, with probability at least $1-\delta'$, we have $|\rcalB_t| \ge q^{m \cdot k - b'}$, which when applied to \Cref{eqn:good matching} gives that $\rcalD_t$ is $(C_t,s)$-bounded and $\|\Den{\rcalD_t}\|_\infty  \le q^{2tb}$.
Thus, $\calE^\ePost_t$ holds, as desired.

\begin{claim}\label{clm:input fiber is large}
    For any $t\in [T]$, consider any fixed $M_{1:t}$ and $S^Y_{1:t-1}$ such that $\calD_{t-1}$ is $(C_{t-1},s)$-bounded and $\|\Den{\InputShort{\calD_{t-1}}{M_t}} - 1\|_\infty \le \delta'$.
    Then: \[
        \Pr[|\rcalB_t|\ge q^{m \cdot k-b'} \mid M_{1:t}, S^Y_{1:t-1}] \ge 1-\delta'. \]
\end{claim}

\begin{proof}
For fixed $M_{1:t}, S^Y_{1:t-1}$, the message $\rS^Y_t = \Player_t(M_{1:t},S^Y_{1:t-1},\rZ_t)$ partitions the space $\Z_q^{m \times k}$ into at most $2^{s^*}$ sets.
Let us denote this partition by $\calP$.
We now argue that in expectation, $\tfrac{1}{|\rcalB_t|}$ is small and then apply Markov's inequality to get the desired result.
We have
\begin{align*}
  \Exp_{\substack{\rX^*\sim  \calD_{t-1}, \\ \rY_{t}\sim \cY_{t}^{\otimes m}}}\bracks*{\tfrac{1}{|\rcalB_t|}} &= \sum_{P\in \calP} \tfrac{1}{|P|} \cdot \Pr_{\substack{\rX^*\sim  \calD_{t-1}, \\ \rY_{t}\sim \cY_{t}^{\otimes m}}}\bracks*{\Pi^{\phi_t}_{M_t}\rX^* - \rY_t \in P}  \\
  &\le \sum_{P\in \calP} \tfrac{1}{|P|} \cdot |P| ~ (1+\delta') ~ q^{-m\cdot k} \\
  &\le (1+\delta')~ q^{-m \cdot k} ~2^{s^*} \,. 
\end{align*}
Therefore, applying Markov's inequality:
\[ \Pr_{\substack{\rX^*\sim  \calD_{t-1}, \\ \rY_{t}\sim \cY_{t}^{\otimes m}}} [
|\rcalB_t| \ge (1-\delta') q^{m\cdot k} 2^{-s^*} (1+\delta')^{-1}] \ge 1-\delta'\,. \]
Finally, we use that $(1-\delta') q^{m\cdot k} 2^{-s^*} (1+\delta')^{-1} \ge q^{m\cdot k-b'}$ from the choice of $s^*$.
\end{proof}

\paragraph{Condition 3.} For every fixed $M_{1:t}$ and $S^Y_{1:t-1}$ satisfying $\cE_t$, we have that (by definition of $\calE_t^\eIn$ and \Cref{obs:boundedness implies uniformity}): \[
\| \Input{\phi_t}{\calY_t}{\calD_{t-1}}{M_t} - \Unif(\Z_q^{m\times k}) \|_\tvd \le \delta'\, .\] 
Applying the data-processing inequality (\Cref{prop:data_processing}) and \Cref{clm:posterior_distribution}, we conclude that \[
\|\rS^Y_t - \Player_t(M_{1:t},S^Y_{1:t-1},\rU)\|_\tvd \le \delta' \, , \]
as desired.
\end{proof}

\section{Boundedness implies uniformity: Proof of \Cref{lemma:boundedness implies uniformity}}\label{sec:boundedness implies uniformity}

In this section, we prove \Cref{lemma:boundedness implies uniformity}, restated as follows:

\boundedness*

In order to prove the lemma, we begin with some definitions.

\begin{definition}[Singleton-freeness]\label{def:singleton-free}
    We say a set $S \subseteq [m] \times [k]$ is \emph{singleton-free} if for every $j \in [m]$, $|S \cap (\{j\} \times [k])| \ne 1$.
    (If we view $S$ equivalently as a matrix in $\{0,1\}^{m \times k}$,
    then $S$ is singleton-free if no row has weight exactly 1.)
    We let $\SF^{m,k} \coloneqq \{S \subseteq [m] \times [k] : S\text{ singleton-free}\}$;
    when $m$ and $k$ are clear from context, we write $\SF = \SF^{m,k}$.
\end{definition}

\begin{proof}
    Follows immediately from \Cref{prop:prelim:fourier of product,prop:prelim:fourier of one-wise}.
\end{proof}

The following simple proposition then forms the heart of our analysis:

\begin{proposition}\label{prop:BiU:OwU to SF}
    Let $k, m \in \N$, $\calY \in \OWU{\Z_q^k}$ be a one-wise uniform distribution, and $U \in \Z_q^{m \times k}$.
    If $\supp(U) \not\in \SF^{m,k}$ (i.e., there exists $j \in [m]$ with $|\supp(U_j)| = 1$, where $U_j$ is the $j$-th row of $U$), then $\hat{\Den{\calY^{\otimes m}}}(U) = 0$.
\end{proposition}

Recall the definition of the embedding $\iota_M : [m] \times [k] \to [n] \times [k]$ corresponding to a matching $M$ (\Cref{def:phm}).
In particular,
\begin{equation}\label{eq:BiU:Valid}
\iota_M(\SF) = \{U \subseteq [n] \times [k]: U \subseteq \supp(M) \text{ and }\forall j \in [m],\ |U \cap \supp(M_j))| \ne 1\}.
\end{equation}
For an injection $\phi : [k] \to [k']$, we also consider the set $\phi(\iota_M(\SF)) = \{(j,\phi(\ell)) : (j,\ell) \in \calS\}$ (as in \Cref{def:hg}).

\subsection{Bounding $\infty$-distance from uniform via $1$-norm of ``singleton-free'' Fourier mass}

We first prove the following upper-bound on the $\infty$-distance from uniform of $\InputShort{\calD}{M}$'s density function
in terms of the ($1$-norm) Fourier mass of $\Den{\calD}$ on ``singleton-free'' coefficients.
This is the key reason we need to consider singleton-free coefficients in this paper.

\begin{lemma}\label{lemma:BiU:bound from valid}
    For every $q, m, k, n, k' \in \N$, injection $\phi : [k] \to [k']$, 
    $k$-partite hypermatching $M \in \PHM{m}{k}{n}$, and one-wise uniform distribution $\calY \in \OWU{\Z_q^k}$, we have: \[
    \|\Den{\InputShort{\calD}{M}} - 1\|_\infty 
    \le \sum_{\substack{U \ne 0 \in \Z_q^{n \times n'}, \\ \supp(U)\in \phi(\iota_M(\SF))}} \abs*{ \hat{\Den{\calD}}(U) } \]
    where, as in \Cref{lemma:boundedness implies uniformity}, $\InputShort{\calD}{M}$ is the distribution of $\Pi^\phi_M \rX - \rY$ for $\rX \sim \calD$ and $\rY \sim \calY^{\otimes m}$,
    and $\iota_M(\SF)$ is as in \Cref{eq:BiU:Valid}.
\end{lemma}

\begin{proof}
    Let $\calP_M \in \Dist{\Z_q^{m \times k}}$ denote the distribution of the projection $\Pi^\phi_M \rX$ for $\rX \sim \calD$.
    Hence $\InputShort{\calD}{M}$ is the distribution of sum of independent samples from $\calP_M$ and $-\calY^{\otimes m}$.
    Consequently, by \Cref{prop:prelim:density of sum}, we have $\Den{\InputShort{\calD}{M}} = \Den{\calP_M} * \Den{-\calY^{\otimes m}}$
    and therefore by \Cref{prop:prelim:convolution,prop:prelim:fourier expansion}, for every $Z_0 \in \Z_q^{m \times k}$, we have: \[
    \Den{\InputShort{\calD}{M}}(Z_0) = \sum_{U \in \Z_q^{m \times k}} \hat{\Den{\calP_M} * \Den{-\calY^{\otimes m}}}(U) \cdot \chi_U(Z_0)
    = \sum_{U \in \Z_q^{m \times k}} \hat{\Den{\calP_M}}(U) \cdot \hat{-\Den{\calY^{\otimes m}}}(U) \cdot \chi_U(Z_0). \]
    Hence, using the triangle inequality, that the zero Fourier coefficient is $1$ for a density function, and $|\chi_U(Z_0)| = 1$: \[
    |\Den{\InputShort{\calD}{M}}(Z_0) - 1| \le \sum_{U\ne0\in \Z_q^{m \times k}} |\hat{\Den{\calP_M}}(U)| \cdot |\hat{\Den{-\calY^{\otimes m}}}(U)|. \]
    The quantity on the RHS is independent of $Z_0$ and thus gives an upper-bound on $\|\Den{\InputShort{\calD}{M}}-1\|_\infty$.

    By \Cref{prop:prelim:fourier of marginal} and the definition of $\calP_M$, we know that for every $U \in \Z_q^{m \times k}$, \[
    \hat{\Den{\calP_M}}(U) = \hat{\Den{\calD}}((\Pi^\phi_M)^\intercal(U)). \]
    Consequently, using the trivial estimate $|\hat{\Den{-\calY^{\otimes m}}}(U)| \le 1$, we get by \Cref{prop:BiU:OwU to SF}: \[
    \|\Den{\InputShort{\calD}{M}} - 1\|_\infty \le \sum_{\substack{U \ne 0 \in \Z_q^{m \times k}, \\ \supp(U) \in \SF} }
    |\hat{\Den{\calD}}((\Pi^\phi_M)^\intercal(U))|. \]
    Finally, we observe that $\{(\Pi^\phi_M)^\intercal U : U \ne 0 \in \Z_q^{m \times k}\text{ and }\supp(U) \in \SF\} = \{U \ne 0 \in \Z_q^{n \times k'} : \supp(U) \in \phi(\iota_M(\SF))\}$ (just by unpacking definitions).
\end{proof}

\subsection{Combinatorial bound}

We next prove an upper-bound on the probability that a fixed set $U$ is in $\iota_\rM(\SF)$ when $\rM$ is sampled randomly in terms of the cardinality $|U|$ of $U$.

\begin{lemma}\label{lemma:BiU:combinatorial bound}
    For every $k \in \N$, there exists $C < \infty$ and $\alpha_0 > 0$ such that the following holds.
    Let $n \in \N$, $m \le \alpha_0 n \in \N$,
    $U \subset [n] \times [k]$, and $h \coloneqq |U|$.
    Then for every $2 \le h \le km \in \N$, sampling a random $k$-partite hypermatching $\rM \in \PHM{m}{k}{n}$, we have \[
    \Pr[U \in \iota_\rM(\SF)]
    \le \parens*{ \frac{C h}n }^{h/2}. \]
    (The probability vanishes at $h = 1$ and $h > km$.)
\end{lemma}

\begin{proof}
    Let $\calV_\ell \coloneqq [n] \times \{\ell\}$ so that $[n] \times [k] = \bigsqcup_{\ell=1}^k \calV_\ell$.
    Let $h_\ell \coloneqq |U \cap \calV_\ell|$.
    
    \paragraph*{Step 1: Symmetry reduction.}
    Define $\calU \coloneqq \{ U' \subset [n] \times [k] : |U' \cap \calV_\ell| = h_\ell \}$.
    To calculate $\Pr[U \in \iota_\rM(\SF)]$,
    by column-wise permutation symmetry, it is equivalent to instead fix one matching, such as \[
    M \coloneqq \begin{bmatrix} 1 & 1 & \cdots & 1 \\ 2 & 2 & \cdots & 2 \\ \vdots & \vdots & \ddots & \vdots \\ m & m & \cdots & m \end{bmatrix}, \]
    sample $\rU$ uniformly at random from the set $\calU \coloneqq \{U \subset [n] \times [k] : \forall \ell \in [k],\ |U \cap \calV_\ell| = h_\ell\}$,
    and calculate $\Pr[\rU\in \iota_M(\SF)]$.
    That is, defining $\calN \coloneqq \calU \cap \iota_M(\SF)$, we want to calculate $|\calN|/|\calU|$.

    \paragraph*{Step 2: Counting and partitioning.}
    For $U \in \calN$, we define $D(U) \coloneqq \{ j \in [m] : U \cap \supp(M_j) \ne \emptyset \}$.
    Note that if $U \in \calN$ and $j \in D(U)$ then $|U \cap \supp(M_j)| \ge 2$ by definition of $\calN$.
    Hence, for $d(u) \coloneqq |D(U)|$, we have $d(U) \le |U|/2 = h/2$.

    For a fixed size $d$, there are $\binom{m}d$ sets $D\subseteq [m]$ of size $d$,
    and for a fixed such set $D$, there are at most $\binom{dk}{h}$ possible sets $U \in \calN$ with $D(U) = U$
    (since $U$ must be fully supported on the rows in $D$).
    Hence we have: \[
    |\calN| \le \sum_{d=1}^{h/2} \binom{\alpha n}{d} \binom{dk}{h}. \]
    Conversely, we have \[
    |\calU| = \prod_{\ell=1}^k \binom{n}{h_\ell} \ge \binom{n}h. \]

    \paragraph{Step 3: Estimation.}
    Note that for every $h \in [1,km]$ and $d \le h/2$, we have $\binom{m}d \le \parens*{ \frac{4km}h }^{h/2}$, since by cases:
    \begin{align*}
    h \in [1, m] &\implies \binom{m}{d} \le \parens*{ \frac{em}{d}}^{d} \le \parens*{ \frac{em}{h/2}}^{h/2} \le \parens*{ \frac{4k m}{h}}^{h/2}, \\
    h \in [m, km] &\implies \binom{m}{d} \le 2^{m} \le 4^{h/2} \le \parens*{ \frac{4km}{h} }^{h/2}.
    \end{align*}
    For the second inequality on the first line, we used that $(C/x)^x$ is an increasing function of $x$ on the interval $x \in (0,C/e)$.
    We also have, by binomial estimates, $\binom{dk}{h} \le \parens*{ \frac{edk}{h} }^h \le (2k)^h$ (using that $ed \le eh/2 \le 2h$) and $\binom{n}{h} \ge \parens*{ \frac{n}{h} }^h$.
    Using $m \le \alpha_0 n$ and the trivial bound $h/2 \le 2^{h/2}$, we altogether get a bound of \[
    \frac{|\calN|}{|\calU|} \le \frac{h}2 \cdot \frac{ \parens*{\frac{4 k \alpha_0 n}{h}}^{h/2} \cdot (2k)^h}{\parens*{ \frac{n}h }^h }
    = 2^{h/2} \cdot \frac{ \parens*{\frac{4 k \alpha_0 n}{h}}^{h/2} \cdot (4k^2)^{h/2}}{\parens*{ \frac{n^2}{h^2} }^{h/2} } 
    = \parens*{ \frac{32 k^3 \alpha h}{n}}^{h/2}, \]
    as desired (setting $C \coloneqq 32k^3 \alpha_0$).
\end{proof}

\subsection{Finishing the proof}

We now use the following fact from \cite[Lemma~5.18]{CGS+22-linear-space}:

\begin{lemma}\label{lem:bound from CGSVV}
    For every $k,q \ge 2 \in \N$, $C_1,C_2 < \infty$ and $\delta \in (0,1)$,
    there exists $\tau > 0$ such that the following holds.
    
    For all $n,s \in \N$ satisfying $4\log_q(3/\delta) \le s \le \tau n$, we have: \[
    \sum_{h=2}^{km} \U{C_1}{s}{h}{n} \cdot \parens*{ \frac{C_2 h}n }^{h/2} \le \delta^2.
    \]
\end{lemma}

\begin{proof}[Proof of \Cref{lemma:boundedness implies uniformity}]
    Using Markov's inequality, it suffices to prove that
    \begin{equation}\label{eq:bounded implies uniform:expectation}
    \Exp_\rM \bracks*{ \|\Den{\InputShort{\calD}{\rM}}  - 1\|_\infty  } \le \delta^2.
    \end{equation}
    
    By \Cref{lemma:BiU:bound from valid} and linearity of expectation, we have: 
\begin{align*}
    \Exp_\rM \bracks*{ \|\Den{\InputShort{\calD}{\rM}} -1\|_\infty  }
    &~\le~ \Exp_\rM \bracks*{ \sum_{\substack{U \ne 0 \in \Z_q^{n \times k'}, \atop \supp(U) \in \phi(\iota_\rM(\SF))}} \abs*{ \hat{\Den{\calD}}(U) } } \\
    &~=~ \sum_{U \ne 0 \in \Z_q^{n \times k'}} \abs*{ \hat{\Den{\calD}}(U) } \cdot \Pr_\rM [\supp(U) \in \phi(\iota_\rM(\SF))].
\end{align*}
    
    Now, consider a fixed $U \ne 0 \in \Z_q^{n \times k'}$ with size $h \coloneqq |U|$.
    We know that $\Pr_\rM [\supp(U) \in \phi(\iota_\rM(\SF))] = 0$ if $h < 2$ or $h > km$.
    On the other hand, 
    if $2 \le h \le km$,
    then we claim that $\Pr_\rM[\supp(U) \in \phi(\iota_\rM(\SF))] \le \parens*{ \frac{C_2 h}n }^{h/2}$
    (where $C_2$ is the constant coming from \Cref{lemma:BiU:combinatorial bound}).
    Indeed, we observe that if $\supp(U) \not\subseteq \phi([n] \times [k])$ then if $\Pr_\rM[\supp(U) \in \phi(\iota_\rM(\SF))] = 0$,
    while if $\supp(U) \subseteq \phi([n] \times [k])$  then $\Pr_\rM[\supp(U) \in \phi(\iota_\rM(\SF))] = \Pr_\rM[\phi^{-1}(\supp(U)) \in \iota_\rM(\SF)] \le \parens*{ \frac{C_2 h}n }^{h/2}$ by \Cref{lemma:BiU:combinatorial bound}.
    Altogether, these calculations give
    \[
    \Exp_\rM[\|\Den{\InputShort{\calD}{\rM}}-1\|_\infty] \le \sum_{h=2}^{k m} \parens*{ \sum_{\substack{U \ne 0 \in \Z_q^{n \times k'}, \\
    \wt(U) = h}  } |\hat{\Den{\calD}}(U)| } \cdot \parens*{ \frac{C_2 h}n }^{h/2}
    \le \sum_{h=2}^{k m} \U{C_1}{s}{h}{k'n} \cdot \parens*{ \frac{(C_2 k') h}{k'n} }^{h/2}, \]
    which is then bounded by the foregoing lemma (\cref{lem:bound from CGSVV}).
\end{proof}

\section{A ``singleton-free'' level inequality}\label{sec:singleton-free}
In this section, we derive upper bounds on the $\ell_1$ norm of a certain subset of Fourier coefficients over Hamming balls
--- centered at the origin at first, and then alter at arbitrary centers in \cref{sec:singleton-free arbitrary} below.
These bounds are later used in the proof of \cref{lem:inductive_step} in \cref{sec:inductive_step}. We begin with some definitions.

The \emph{weight} of a frequency matrix $U \in \Z_q^{m \times k}$ is $\wt(U) \coloneqq |\supp(U)|$.
As a convenient abbreviation, we define the notation below for sets of ``singleton-free'' frequency matrices (where the set $\SF^{m,k}$ is as in \Cref{def:singleton-free}.)
\begin{align*}
\calU_q^{m,k} &~\coloneqq~ \{U \in \Z_q^{m,k} : \supp(U) \in \SF^{m,k}\} \\
\calU_q^{m,k}(h) &~\coloneqq~ \{U \in \calU_q^{m,k} : \wt(U) = h\}
\end{align*}

In particular, $\calU_q^{m,k} = \{U \in \Z_q^{m \times k} : \forall j \in [m],\ |\supp(U_j)| \ne 1\}$, where $U_j$ is the $j$-th row of $U$.
The goal of this section is to prove the following ``refined'' level inequality bounding the $\ell_1$-norm of singleton-free Fourier coefficients:

\begin{theorem}\label{thm:mass on UH}
For every $q, k \in \N$ and $\theta < \infty$, there exists $\zeta < \infty$ and $\epsilon_0 > 0$ such that the following holds.
For every $m, b \in \N$, $g : \Z_q^{m \times k} \to \C$ with $\|g\|_1 = 1$, and $h \in \N$ with $h < \theta b$ and $\log_q(\|g\|_\infty) \le b \le \epsilon_0 km$, we have: \[ 
\sum_{U \in \calU_q^{m,k}(h)} |\hat{g}(U)| \le \parens*{ \frac{\zeta \sqrt{b m}}{h} }^{h/2}. \]
\end{theorem}
Writing $U_j \in \Z_q^k$ for row~$j$ of a matrix $U \in \Z_q^{m \times k}$, we 
define the \emph{row weight} of $U \in \Z_q^{m \times k}$ as the number of non-zero rows 
$\rwt(U) \coloneqq |\{j : U_j \ne 0\}|$ (i.e., the number of nonzero rows).
We also denote by $\calU_q^{m,k}(h,\ell)$, the set of set of singleton-free frequencies with weight $h$ and row-weight $\ell$
\[
\calU_q^{m,k}(h,\ell) ~\coloneqq~ \{U \in \calU_q^{m,k}(h) : \rwt(U) = \ell\}.
\]
We will prove the above theorem by bounding the $\ell_2$ Fourier mass for the sets $\calU_q^{m,k}(h,\ell)$, together with a counting argument and Cauchy-Schwarz to obtain $\ell_1$ bounds.

\subsection{Bounding the squared Fourier mass on $\calU_q^{m,k}(h,\ell)$}
We now prove the following bound on the $\ell_2$ Fourier mass of the set $\calU_q^{m,k}(h,\ell)$.
\begin{lemma}\label{thm:squared mass on UHL}
For every $q, k \in \N$ and $\theta > 0$, there exists $\zeta > 0$ such that the following holds.
For every $m,b \in \N$, $g : \Z_q^{m \times k} \to \C$ with $\|g\|_1 = 1$, and $h > \ell \in \N$ with $h - \ell < \theta b$ and $\log_q(\|g\|_\infty) \le b \le \epsilon_0 km$, we have: \[ 
\sum_{U \in \calU_q^{m,k}(h,\ell)} |\hat{g}(U)|^2 \le \parens*{ \frac{\zeta b}{h - \ell}}^{h-\ell}. \]

\end{lemma}
We will prove the lemma using hypercontractivity for a noisy version of $g$. We consider two different noise operators below.

\subsubsection{The row noise operator and row hypercontractivity}

\begin{definition}
    A \emph{multiplier operator} is an operator on $L^2(\Z_q^{m \times k})$ (the vector space of functions $\Z_q^{m \times k} \to \C$)
    which is diagonalizable in the Fourier basis.
    That is, letting $\Lambda : \Z_q^{m \times k} \to \C$ be a \emph{multiplier function},
    the corresponding \emph{multiplier operator}, denoted $T_\Lambda : L^2(\Z_q^{m \times k}) \to L^2(\Z_q^{m \times k})$
    is defined on the Fourier basis via $T_\Lambda \chi_U \coloneqq \Lambda(U) \chi_U$ for every frequency $U \in \Z_q^{m \times k}$.
    Hence, for arbitrary $g : \Z_q^{m \times k} \to \C$, we have $\hat{T_\Lambda g}(U) = \Lambda(U) \cdot \hat{g}(U)$.
    Equivalently, $T_\Lambda g = g * \Lambda'$, where $*$ denotes convolution and $\Lambda'$ is the inverse Fourier transform of $\Lambda$.
\end{definition}
 
\begin{definition}[Row noise]
For $\rho \in [0,1]$, we define:
\begin{itemize}
    \item $\calD^*_\rho \in \Dist{\Z_q^k}$, is the distribution  which outputs $0$ with probability $\rho$ and a uniformly random element of $\Z_q^k$ with probability $1-\rho$.
    \item The Fourier transform $\lambda^*_\rho \coloneqq \hat{\Den{\calD^*_\rho}}$, which we view as a multiplier function.
    We can calculate $\lambda^*_\rho(u) = \begin{cases} 1 & \text{if }u = 0, \\ \rho & \text{if }u \ne 0 \end{cases}$.
    \item For $U \in \Z_q^{m \times k}$, $\Lambda^* \coloneqq \hat{\Den{(\calD_\rho^*)^{\otimes m}}}$.
    (Equivalently, $\Lambda^*(U) = \prod_{j=1}^m \lambda^*_\rho(U_j) = \rho^{\rwt(U)}$ 
    where $\rwt(U) = |\{i : U_i \neq 0\}|$ is the row weight.)
    \item The corresponding \emph{row noise operator} $T^*_\rho \coloneqq T_{\lambda^*_\rho}$, which acts by convolution with $\Den{\calD^*_\rho}$.
    That is, \[
    (T^*_\rho g)(X) = \Exp_{\rY \sim (\calD_\rho^*)^{\otimes m}}[g(X - \rY)]. \qedhere \]
\end{itemize}
\end{definition}

Viewing each row as a single coordinate supported on $q^k$ different values, we get the following hypercontractive inequality:

\begin{lemma}[Row hypercontractive inequality, e.g.,~{\cite{OD14}}]\label{lem:row_hc}
Let $Q \coloneqq q^k$.
For any $1 < p < 2$ and $g : \Z_q^{m \times k} \to \C$, it holds that
\begin{align*}
\|T^*_\rho g\|_2 \le \|g\|_p
&&\text{where}&&
\rho \coloneqq \sqrt{p-1} \cdot Q^{1/2-1/p}.
\end{align*}
\end{lemma}
 
\subsubsection{A modified noise operator}
 
\begin{definition}
For $\rho \in [0,1]$, we define:
\begin{itemize}
\item The \emph{modified noise multiplier} (for a single row) \[ \lambda_\rho(u) \coloneqq \begin{cases}
1 & \text{if } u = 0, \\
0 & \text{if } |\supp(u)| = 1, \\
\rho^{|\supp(u)|-1} & \text{otherwise}.
\end{cases}
\]
\item The corresponding modified noise multiplier for a matrix:
$\Lambda_\rho(U) \coloneqq \prod_{j=1}^m \lambda_\rho(U_j)$.
\item The corresponding
\emph{modified noise operator} $T_\rho \coloneqq T_{\Lambda_\rho}$. \qedhere
\end{itemize}
\end{definition}

\begin{claim}\label{claim:extraction}
For every $U \in \calU_q^{m,k}(h,\ell)$ and $\rho \in [0,1]$,
$\Lambda_\rho(U) = \rho^{h-\ell}.$
\end{claim}
\begin{proof}
Taking the logarithm base $\rho$ of both sides, on the LHS, rows with support $0$ contribute $0$,
and a row $j$ with support at least $2$ contributes $|\supp(U_j)|-1$.
We then calculate $\sum_{j : |\supp(U_j)| \ge 2} (|\supp(U_j)|-1) = h-\ell$.
\end{proof}

\begin{claim}
For every $u \in \Z_q^k$ and $\rho \in [0,1]$, it holds that $\lambda_\rho(u) \le \lambda^*_\rho(u)$.
\end{claim}
\begin{proof}
    
If $u = 0$, we have $\lambda_\rho(u) = \lambda^*_\rho(u) = 1$.
Otherwise, $\lambda_\rho(u) \le \rho$ by inspection, whereas $\lambda_\rho(u) = \rho$.
\end{proof}

\begin{claim}\label{claim:domination}
For every $g : \Z_q^{m \times k} \to \C$ and $\rho \in [0,1]$, it holds that 
$\|T_\rho g\|_2 \le \|T^*_\rho g\|_2$.
\end{claim}
\begin{proof}
We calculate \[
\|T_\rho g\|_2^2 = \sum_{U \in \Z_q^{m \times k}} (\Lambda_\rho(U))^2|\hat{g}(U)|^2 
\le \sum_{U \in \Z_q^{m \times k}} (\Lambda^*_\rho(U))^2|\hat{g}(U)|^2 = \|T^*_\rho g\|_2^2. \qedhere \]
\end{proof}

\begin{claim}\label{claim:chain}
    For every $g : \Z_q^{m \times k} \to \C$ and $r \in [0,1]$,
    it holds that $\sum_{U\in\calU_q^{m,k}(h,\ell)} |\hat{g}(U)|^2 \le \rho^{-2(h-\ell)} \|T^*_\rho g\|_2^2$.
\end{claim}
\begin{proof}
    By \Cref{claim:extraction,claim:domination}, we have:
\begin{multline*}
\rho^{2(h-\ell)}\sum_{U \in \calU_q^{m,k}(h,\ell)} |\hat{g}(U)|^2
= \sum_{U \in \calU_q^{m,k}(h,\ell)} \rho^{2(h-\ell)}|\hat{g}(U)|^2
= \sum_{U \in \calU_q^{m,k}(h,\ell)} (\Lambda_\rho(U))^2|\hat{g}(U)|^2 \\
\le \sum_{U \in \Z_q^{m \times k}} (\Lambda_\rho(U))^2|\hat{g}(U)|^2
= \|T_\rho  g\|_2^2
\le \|T^*_\rho g\|_2^2.\qedhere
\end{multline*}
\end{proof}

\subsubsection{Proof of \Cref{thm:squared mass on UHL}}

We are now equipped to prove \Cref{thm:squared mass on UHL}.
 
\begin{proof}[Proof of \Cref{thm:squared mass on UHL}]
We proceed in two steps.
 
\medskip
\noindent\textbf{Step 1: Row hypercontractivity.}
Let $p \in (1,2)$.
By \Cref{lem:row_hc}, \[
\|T^*_\rho g\|_2^2 \le \|g\|_p^2. \]
Since $\|g\|_1 = 1$ and $\|g\|_\infty \le q^b$: \[
\|g\|_p^p = \Exp[|g|^p] = \Exp[|g| \cdot |g|^{p-1}] \le \|g\|_\infty^{p-1} \Exp[|g|] = q^{b(p-1)}. \]
Hence $\|g\|_p^2 \le q^{2b(1-1/p)}$.
 
\medskip
\noindent\textbf{Step 2: Setting parameters.}
Set $Q \coloneqq q^k$ and $p \coloneqq 1 + \frac{h-\ell}{\theta b}$.
This is less than $2$ since $h - \ell < \theta b$.
Combining \Cref{claim:chain} and Step~1:
\begin{equation*}\label{eq:master}
\sum_{U \in \calU_q^{m,k}(h,\ell)} |\hat{g}(U)|^2 \le \underbrace{\rho^{-2(h-\ell)}}_{\eqqcolon W_1} \cdot \underbrace{q^{2b(1-1/p)}}_{\eqqcolon W_2}.
\end{equation*}
We bound the terms separately as: \[
W_1 = \parens*{ \sqrt{p-1} \cdot Q^{1/2-1/p} }^{-2(h-\ell)}
= \parens*{ \frac{\theta b}{h-\ell} \cdot Q^{2/p-1} }^{h-\ell} \le \parens*{ \frac{(\theta q^k)b}{h-\ell} }^{h-\ell}, \]
where we used $p \ge 1 \implies Q^{2/p-1} \le Q$ and $Q = q^k$,
and \[
W_2 = q^{2b(1-1/p)} \le q^{2b(p-1)} = (q^{2/\theta})^{h-\ell}, \]
where the inequality used that $1-1/p \le p-1$ for $p > 0$.
This gives the desideratum with $\zeta \coloneqq \theta q^{k+2/\theta}$.
\end{proof}

\subsection{Bounding the Fourier mass on $\calU_q^{m,k}(h)$}

We now return to proving \Cref{thm:mass on UH}.
This first requires the following lemma:

\begin{lemma}\label{lemma:size of valid}
    For every $q, k \in \N$, there exists $\zeta > 0$ such that the following holds.
    For every $m \in \N$ and $h > \ell \in \N$, we have: \[ 
    |\calU_q^{m,k}(h,\ell)| 
    \le \parens*{ \frac{\zeta m}{\ell}}^{\ell}. \]
\end{lemma}
 
\begin{proof}
We severely overcount and give an upper bound on the set $\{U \in \Z_q^{m \times k} : \rwt(U) = \ell\}$.
Each $U \in \calU_q^{m,k}(h,\ell)$ is uniquely determined by specifying, for every nonzero row $j$ (i.e., $U_j \ne 0$), the row vector $U_j$.
It therefore suffices to count the number of ways to choose a subset of nonzero rows and then assign these rows to (nonzero) row vectors.
There are $\binom{m}{\ell}$ ways to choose a set of $\ell$ rows to be nonzero.
Each row then can be assigned to at most $q^k-1$ vectors.
Hence, altogether we get a bound of \[
\binom{m}{\ell} \cdot (q^k-1)^\ell \le \parens*{ \frac{em}{\ell} }^\ell \cdot (q^k-1)^\ell = \parens*{ \frac{\zeta m}{\ell} }^\ell \]
for $\zeta \coloneqq e(q^k-1)$, as desired.
\end{proof}

We can now prove \Cref{thm:mass on UH}.

\begin{proof}[Proof of \Cref{thm:mass on UH}]
    Let $\zeta_1,\zeta_2 > 0$ be the two constants from \Cref{thm:squared mass on UHL,lemma:size of valid}, respectively.
    By Cauchy-Schwarz, and since $h/k \le \rwt(U) \le h/2$ for every $U \in \calU_q^{m,k}(h)$, we have: \[
    \sum_{U \in \calU_q^{m,k}(h)} |\hat{g}(U)|
    = \sum_{\ell=h/k}^{h/2} \sum_{U \in \calU_q^{m,k}(h,\ell)} |\hat{g}(U)| \\
    \le \sum_{\ell=h/k}^{h/2} \parens*{\underbrace{|\calU_q^{m,k}(h,\ell)| \cdot \sum_{U \in \calU_q^{m,k}(h,\ell)} |\hat{g}(U)|^2}_{\eqqcolon S_\ell} }^{1/2}. \]
    We can bound each individual $S_\ell$ via \Cref{thm:squared mass on UHL,lemma:size of valid}:
    \begin{align*}
    S_\ell ~=~ |\calU_q^{m,k}(h,\ell)| \cdot \sum_{U \in \calU_q^{m,k}(h,\ell)} |\hat{g}(U)|^2
    &~\le~  \parens*{ \frac{\zeta_1 b}{h-\ell} }^{h - \ell} \cdot \parens*{ \frac{\zeta_2 m}\ell }^\ell \\
    &~\le~ \parens*{ \frac{2\zeta_1 b}{h} }^{h - \ell} \cdot \parens*{ \frac{\zeta_2 m}\ell }^\ell 
    ~=~ \parens*{ \frac{2\zeta_1 b}{h} }^h \cdot \parens*{ \frac{\zeta_2 h m}{2\zeta_1 b \ell} }^\ell \\
    &~\le~ \parens*{ \frac{2\zeta_1 b}{h} }^h \cdot \parens*{ \frac{\zeta_2 k m}{2\zeta_1 b} }^\ell,
    \end{align*}
    where the last two inequalities use $\ell \le h/2$ and $\ell \ge h/k$, respectively.
    Hence \begin{align*}
    \sum_{U \in \calU_q^{m,k}(h)} |\hat{g}(U)|
    ~\le~ \sum_{\ell=h/k}^{h/2} \sqrt{S_\ell}
    &~\le~ \parens*{ \frac{2 \zeta_1 b }h }^{h/2} \cdot \sum_{\ell=h/k}^{h/2} \parens*{ \sqrt{\frac{\zeta_2 k m}{2\zeta_1 b}} }^\ell \\
    &~\le~ \parens*{ \frac{2 \zeta_1 b }h }^{h/2} \cdot 2 \parens*{ \sqrt{\frac{\zeta_2 k m}{2\zeta_1 b}} }^{h/2}
    ~=~ 2 \parens*{ \frac{\sqrt{2 \zeta_1\zeta_2 k b m}}h }^{h/2},
    \end{align*}
    where the final inequality sets $\epsilon_0 \coloneqq \frac{\zeta_2}{8\zeta_1}$, forcing $\sqrt{\frac{\zeta_2 k m}{2\zeta_1 b}} \ge 2$ (using that $b \le \epsilon_0 km$).\footnote{
        Note that if $X \ge 2$ then $\sum_{j=0}^H X^j = \frac{X^{j+1} - 1}{X - 1} \le \frac{X^{j+1}}{X/2} = 2X$, where we used $X-1 \ge X/2$.}
\end{proof}

\subsection{``Singleton-free'' level inequality around arbitrary centers}\label{sec:singleton-free arbitrary}

In this section, we ``boost'' the theorem which we stated at the beginning of this section (\Cref{thm:mass on UH})
to give bounds for singleton-free matrices in balls around arbitrary points.

We begin with the following simple corollary of \Cref{thm:mass on UH}:

\begin{corollary}\label{cor2:mass on UH}
For every $q, k \in \N$ and $\theta < \infty$, there exists $\zeta < \infty$ and $\epsilon_0 > 0$ such that the following holds.
For every $m, b \in \N$, $g : \Z_q^{m \times k} \to \C$ with $\|g\|_1 = 1$, $V \in \Z_q^{m \times k}$, $h \in \N$ with $h < \theta b$ and $\log_q(\|g\|_\infty) \le b \le \epsilon_0 km$, we have: \[ 
\sum_{U-V \in \calU_q^{m,k}(h)} |\hat{g}(U)| \le \U{C}{s}{h}{km}. \]
\end{corollary}

\begin{proof}
    Define $g' : \Z_q^{m \times k} \to \C$ via $g'(U) \coloneqq g(U) \cdot \overline{\chi_V(U)}$,
    so that $\hat{g'}(U) = \Exp_{\rX \in \Z_q^{m \times k}} [g'(\rX) \cdot \overline{\chi_V(\rX)}]
    = \Exp_{\rX \in \Z_q^{m \times k}} [g(\rX) \cdot \overline{\chi_{U+V}(\rX)}]
    = \hat{g}(U+V)$.
    Hence, \[
    \sum_{U-V \in \calU_q^{m,k}(h)} |\hat{g}(U)|
    = \sum_{U \in \calU_q^{m,k}(h)} |\hat{g}(U+V)|
    = \sum_{U \in \calU_q^{m,k}(h)} |\hat{g'}(U)|, \]
    at which point we can apply \Cref{thm:mass on UH} (since $g'$ is just a phase shift of $g$, we have $\|g\|_1 = \|g'\|_1$ and $\|g\|_\infty = \|g'\|_\infty$).
\end{proof}

\begin{lemma}\label{cor:mass on UH}
For every $q, k \in \N$ and $\theta < \infty$, there exists $\zeta < \infty$ and $\epsilon_0 > 0$ such that the following holds.
For every $m \in \N$, $g : \Z_q^{m \times k} \to \C$ with $\|g\|_1 = 1$, $V \in \Z_q^{m \times k}$, and $h \in \N$ with $h < \theta b$, we have: \[ 
\sum_{\substack{U \in \SF, \\ \wt(U-V)=h}} |\hat{g}(U)| \le q^{k \cdot \rwt(V)} \cdot \parens*{ \frac{\zeta \sqrt{b m}}{h-\kappa} }^{(h-\kappa)/2}, \]
where $\log_q(\|g\|_\infty) \le b \le \epsilon_0 km$ and $\kappa \coloneqq |\{j \in [m] : |\supp(V_j)| = 1\}|$.
\end{lemma}

\begin{proof}
Fix $V \in \Z_q^{m \times k}$.
For $U \in \Z_q^{m \times k}$, define a matrix $Z(U) \in Z_q^{m \times k}$ row-wise,
via \[
(Z(U))_j \coloneqq -(U+V)_j \cdot \Ind{|\supp(V_j)| = 1 \vee |\supp((U+V)_j)|=1}. \]
We have a few basic facts:
\begin{itemize}
\item For every $U \in \Z_q^{m \times k}$, $U + V + Z(U) \in \SF$ by definition.
\item If $U \in \SF$, then for every $j \in [m]$, $(Z(U))_j \ne 0 \implies V_j \ne 0$.
Indeed, by contrapositive, if $V_j = 0$, then $(U+V)_j = U_j$, and hence $|\supp((U+V)_j)| = |\supp(U_j)| \ne 1$ since $U \in \SF$.
\item $\wt(Z(U)) \ge \kappa$.
\item If $\wt(U+V) = h$, then $\wt(U+V+Z(U)) = h-\wt(Z(U))$.
\end{itemize}

Let $\calZ \coloneqq \{Z \in \Z_q^{m \times k} : \forall j \in [m], Z_j \ne 0 \implies V_j \ne 0\}$
denote the set of frequencies supported only on rows in $V$'s support.
By the second item above, we have $Z(U) \in \calZ$ for every $U \in \SF$.
Also, note that $|\calZ| \le q^{k \cdot \rwt(V)}$.

We can now bound, using the preceding \Cref{cor2:mass on UH}:
\begin{align*}
\sum_{\substack{U \in \SF, \\ \wt(U+V) = h}} |\hat{g}(U)|
&~=~ \sum_{Z\in\calZ} \sum_{\substack{U \in \SF, 
Z(U) = Z, \atop
\wt(U+V+Z) = h-\wt(Z)}} |\hat{g}(U)| \\
&~\le~ \sum_{Z\in\calZ} \sum_{\substack{U+V+Z\in\SF,
Z(U) = Z, \atop
\wt(U+V+Z) = h-\wt(Z)}} |\hat{g}(U)| \\
&~\le~ \sum_{Z\in\calZ} \sum_{\substack{U+V+Z \in \SF, \atop \wt(U+V+Z) = h-\wt(Z)}} |\hat{g}(U)| \\
&~\le~ \sum_{Z \in \calZ} \parens*{ \frac{\zeta \sqrt{b m}}{h-\wt(Z)} }^{(h-\wt(Z))/2} \\
&~\le~ \sum_{Z \in \calZ} \parens*{ \frac{\zeta \sqrt{b m}}{h-\kappa} }^{(h-\kappa)/2}
~\le~ q^{k\cdot\rwt(V)} \cdot \parens*{ \frac{\zeta \sqrt{b m}}{h-\kappa} }^{(h-\kappa)/2}.
\end{align*}
\end{proof}

\section{The inductive argument: Proof of \Cref{lem:inductive_step}}\label{sec:inductive_step}

Finally, in this section, we prove \Cref{lem:inductive_step}, restated as follows:

\inductivestep*

\subsection{An expression for the density function}

We first give an analytic expression for the density function of the distribution $\PostShort{\calD}{M}{\calB}$ in terms of the density functions of $\calD$, $\calY^{\otimes m}$, and $\calB$ (note that $\Pr_{\rZ \sim \Z_q^{m \times k}}[\rZ \in \calB] = |\calB|/q^{m \cdot k}$ below).

\begin{lemma}\label{lemma:inductive step:posterior form}
    Let $\calD \in \Dist{\Z_q^{n \times k'}}$ be a distribution,
    $\phi : [k] \to [k']$ be an injection,
    $\calY \in \Dist{\Z_q^k}$ a distribution,
    $M \in \PHM{m}{k}{n}$ a $k$-partite hypermatching,
    and $\calB \subseteq \Z_q^{m \times k}$ a set.
    Then for every $X_0 \in \Z_q^{n \times k'}$: \[
    \Den{\PostShort{\calD}{M}{\calB}}(X_0) =  
    \Den{\calD}(X_0)
    \cdot (\Den{\calY^{\otimes m}} * \Den{\calB})(\Pi^\phi_M (X_0))
    \cdot \nu_{M,\calB}, \]
    where \[
    \nu_{M,\calB} \coloneqq \frac{ \Pr_{\rZ \sim \Z_q^{m \times k}}[\rZ \in \calB ]}
    {\Pr_{\rX \sim \calD, \rY \sim \calY^{\otimes m}}[\Pi^\phi_M \rX-\rY \in \calB ]} \]
    is a normalization factor which does not depend on $X_0$.
\end{lemma}
\begin{remark}
    Though we will not need this fact, it is easy to check that $(\Den{\calY^{\otimes m}} * \Den{\calB})(\Pi^\phi_M (X_0))
    = \Den{\PostShort{\calU}{M}{\calB}}(X_0)$,
    where $\calU \coloneqq \Unif{\Z_q^{n \times k'}}$ is the \emph{uniform} distribution on $\Z_q^{n \times k'}$.
    That is, $\PostShort{\calU}{M}{\calB}$ is the posterior distribution of $\rX$ when sampling $\rX$ from the \emph{uniform} prior ($\calU$) conditioned on $\Pi^\phi_M \rX - \rY \in \calB$ for $\rY \sim \calY^{\otimes m}$.)
    Thus, \Cref{lemma:inductive step:posterior form} states that (up to normalization) the posterior on $\rX$ with respect to the prior $\calD$
    is the product of the prior $\calD$ and the posterior with respect to the uniform prior.
\end{remark}
\begin{proof}
    Let $X_0 \in \Z_q^{n \times k'}$.
    By Bayes' rule,
    \begin{multline*}
        \Pr_{\rX \sim \calD, \rY \sim \calY^{\otimes m}}[\rX = X_0 \mid \Pi^\phi_M \rX - \rY \in \calB] \\
        = \Pr_{\rY \sim \calY^{\otimes m}}[\Pi^\phi_M X_0 - \rY \in \calB] \cdot \Pr_{\rX \sim \calD}[\rX = X_0] \cdot \parens*{ \Pr_{\rX \sim \calD, \rY \sim \calY^{\otimes m}}[\Pi^\phi_M \rX - \rY \in \calB] }^{-1}.
    \end{multline*}
    Consequently, \[
        \Den{\PostShort{\calD}{M}{\calB}}(X_0) = \Pr_{\rY \sim \calY^{\otimes m}}[\Pi^\phi_M X_0 - \rY \in \calB] \cdot \Den{\calD}(X_0)  \cdot \nu_{M,\calB} \cdot \frac{q^{m \cdot k}}{|\calB|}.
    \]
    Finally, \[ \frac{\Pr_{\rY \sim \calY^{\otimes m}}[\Pi^\phi_M X_0 - \rY \in \calB]}{|\calB|}
    = \Pr_{\rB \sim \calB, \rY \sim \calY^{\otimes m}}[\rB + \rY = \Pi^\phi_M X_0] \] 
    and therefore
    \begin{align*}
        \Den{\PostShort{\calD}{M}{\calB}}(X_0)
        &= \Pr_{\rB \sim \calB, \rY \sim \calY^{\otimes m}}[\rB + \rY = \Pi^\phi_M X_0] \cdot \Den{\calD}(X_0)  \cdot \nu_{M,\calB} \cdot q^{m \cdot k} \\
        &= (\Den{\calB} * \Den{\calY^{\otimes m}}) (\Pi^\phi_M X_0) \cdot \Den{\calD}(X_0) \cdot \nu_{M,\calB}. \qedhere
    \end{align*}
\end{proof}

\Cref{lemma:inductive step:posterior form} lets us derive a useful upper bound on the $1$-norm of Fourier coefficients in a ball around the origin in a posterior distribution:

\begin{corollary}\label{cor:inductive step:upper bound}
    In the setup of \Cref{lemma:inductive step:posterior form}, we have \[
    \sum_{\substack{W \in \Z_q^{n \times k'}, \\ \wt(W) = h}} \abs{ \hat{\Den{\PostShort{\calD}{M}{\calB}}}(W)}
    \le 
    \nu_{M,\calB} \cdot \sum_{V \in \Z_q^{n \times k'}} \abs{\hat{\Den{\calD}}(U)}
    \sum_{\substack{U \in \SF, \\ \wt(U+\Pi^\phi_M U) = h - (\wt(V) - \wt(\Pi^\phi_M V))}} \abs{\hat{\Den{\calB}}(W)}. \]
\end{corollary}
Note that $\wt(V)-\wt(\Pi^\phi_M V)$ equals the size of the support of $V$ \emph{outside} of $\phi(\supp(M))$.

\begin{proof}
    We apply \Cref{lemma:inductive step:posterior form,prop:prelim:convolution} to get:
\begin{align*}
\hat{\Den{\PostShort{\calD}{M}{\calB}}}(W)        &=  
    \nu_{M,\calB} \cdot \sum_{V \in \Z_q^{n \times k'}} \hat{\Den{\calD}}(V)
    \cdot \hat{((\Den{\calY^{\otimes m}} * \Den{\calB}) \circ \Pi^\phi_M)} (W-V) \\
\intertext{and then \Cref{prop:prelim:convolution,prop:prelim:projection} to get:}
    &=  
    \nu_{M,\calB} \cdot \sum_{\substack{V \in \Z_q^{n \times k'}, \\ \supp(W-V) \subseteq \phi(\supp(M))}} \hat{\Den{\calD}}(V)
    \cdot \hat{\Den{\calY^{\otimes m}}}(\Pi^\phi_M(W-V))  \cdot \hat{\Den{\calB}}(\Pi^\phi_M(W-V)).
\end{align*}
Hence by \Cref{prop:prelim:fourier of one-wise}, the triangle inequality, and upper bounding $|\hat{\Den{\calY^{\otimes m}}}(\Pi^\phi_M(W-V))|$ by $1$, \[
|\hat{\Den{\PostShort{\calD}{M}{\calB}}}(W) | 
\le \nu_{M,\calB} \cdot \sum_{\substack{V \in \Z_q^{n \times k'}, \\ \supp(W-V) \subseteq \phi(\supp(M)), \\ \Pi^\phi_M (W-V) \in \SF}}
\abs{\hat{\Den{\calD}}(V)} \cdot \abs{\hat{\Den{\calB}}(\Pi^\phi_M(W-V))}.
\]
Hence 
\begin{align*}
    \sum_{\substack{W \in \Z_q^{n \times k'}, \\ \wt(W) = h}} \abs{ \hat{\Den{\PostShort{\calD}{M}{\calB}}}(W)}
    &\le 
    \nu_{M,\calB} \cdot \sum_{V \in \Z_q^{n \times k'}} \abs{\hat{\Den{\calD}}(V)}
    \sum_{\substack{W \in \Z_q^{n \times k'}, \\ \wt(W) = h, \\ \supp(W-V) \subseteq \phi(\supp(M)), \\ \Pi^\phi_M(W-V) \in \SF}} \abs{\hat{\Den{\calB}}(\Pi^\phi_M(W-V))} \\
    \intertext{and substituting $U \coloneqq \Pi^\phi_M(W-V)$,}
    &\le 
    \nu_{M,\calB} \cdot \sum_{V \in \Z_q^{n \times k'}} \abs{\hat{\Den{\calD}}(V)}
    \sum_{\substack{U \in \SF, \\ \wt(U+\Pi^\phi_M V) = h - (\wt(V) - \wt(\Pi^\phi_M V))}} \abs{\hat{\Den{\calB}}(U)},
\end{align*}
where we used that for $W \in \Z_q^{n \times k'}$ with $\supp(W-V) \subseteq \phi(\supp(M))$,
we have $\wt(W) = h \iff \wt(\Pi^\phi_M W) = h - (\wt(V) - \wt(\Pi^\phi_M V))$
(because $\supp(W-V) \subseteq \phi(\supp(M))$, 
implies $W$ and $V$ agree on all coordinates outside $\phi(\supp(M))$,
and thus $\wt(W) - \wt(\Pi^\phi_M W) = \wt(V) - \wt(\Pi^\phi_M V)$.)
\end{proof}

\subsection{The inductive step ``in expectation''}

\Cref{cor:inductive step:upper bound} gives a natural approach to proving \Cref{lem:inductive_step}:
We formulate and prove an ``in-expectation'' version of the lemma, which then will imply the lemma by Markov's inequality.

\begin{lemma}[Inductive step in expectation]\label{lemma:inductive step:in expectation}
    For every $q,k,k' \in \N$, there exist $\alpha_0 > 0$ and $C_0 < \infty$ such that for every $C > C_0$,
    there exist $\sigma \in (0,1)$ and $C' < \infty$ such that the following holds.
    
    For every injection $\phi : [k] \to [k']$, one-wise uniform distribution $\calY \in \OWU{k}$,
    and every $n,m,s,h \in \N$ satisfying $m \le \alpha_0 n$ and $1 \le h \le s \le \sigma n$, and every $\calD \in \Dist{\Z_q^{n \times k'}}$ which is $(C,s)$-bounded: \[
    \sum_{V \in \Z_q^{n \times k'}} |\hat{\Den{\calD}}(V)| \cdot \Exp_{\rM\in\PHM{m}{k}{n}} \bracks*{ \max_{\substack{g : \Z_q^{m \times k} \to \C, \\
    \|g\|_1 = 1, \\
    \|g\|_\infty \le q^s}} \sum_{\substack{U \in \SF, \\ \wt(U+\Pi^\phi_\rM V) = h - (\wt(V) - \wt(\Pi^\phi_\rM V))}} |\hat{g}(U)| }
    \le \parens*{ \frac{C' \sqrt{s n k'}}{h} }^{h/2}. \]
\end{lemma}

Before proving this lemma, we show how it will imply \Cref{lem:inductive_step}.
To do so, we will require the following simple analogue of \cite[Lemma 6.3]{CGS+22-linear-space},
which is a trivial upper-bound on the $1$-norm mass in high-Hamming weight balls:

\begin{lemma}\label{lemma:trivial bound}
    For every $q,N \in \N$ and $g : \Z_q^N \to \C$ with $\|g\|_1 = 1$ and $\log_q(\|g\|_\infty) \le h$, we have \[
    \sum_{\substack{U \in \Z_q^N, \\ \wt(U) = h}} |\hat{g}(U)| \le \parens*{ \frac{q^2e N}{h} }^{h/2}.
    \]
\end{lemma}

\begin{proof}
    By Parseval's identity (\Cref{prop:prelim:parseval}) and the assumptions, we have:
    \begin{equation}\label{eq:trivial 2 bound}
    \sum_{\substack{U \in \Z_q^N, \\ \wt(U) = h}} |\hat{g}(U)|^2
    \le \sum_{U \in \Z_q^N} |\hat{g}(U)|^2 = \|g\|_2^2 
    \le \|g\|_1 \cdot \|g\|_\infty \le q^h.
    \end{equation}
    Also, there are (exactly) $(q-1)^h \cdot \binom{N}h$ frequencies $U \in \Z_q^N$ with $\wt(U) = h$.
    By Cauchy-Schwarz, \[
    \sum_{\substack{U \in \Z_q^N, \\ \wt(U) = h}} |\hat{g}(U)|
    \le \sqrt{(q-1)^h \binom{N}h \sum_{\substack{U \in \Z_q^N, \\ \wt(U) = h}} |\hat{g}(U)|^2}
    \le \sqrt{(q-1)^h \parens*{ \frac{eN}h }^h \cdot q^h}
    \le \parens*{ \frac{eq^2 N}h }^{h/2},
    \]
    where the second inequality uses the standard binomial estimate and \Cref{eq:trivial 2 bound}.
\end{proof}

We can now give a proof of \Cref{lem:inductive_step}.

\begin{proof}[Proof of \Cref{lem:inductive_step} (modulo \Cref{lemma:inductive step:in expectation})]
    Let $\alpha_0$ and $C_0$ be as defined in \Cref{lemma:inductive step:in expectation}.
    Given parameters $C$ and $\delta$, we then invoke \Cref{lemma:inductive step:in expectation} and let $\sigma$ and $C'$ denote the resulting values.

    Let $\calE(M)$ denote the event that the matching $M \in \PHM{m}{k}{n}$ is such that $\|\mu_{\InputShort{\calD}{M}} -1 \|_\infty > \delta$.
    For $h \in [\sigma n]$, let $\calF_h(M)$ denote the event that \[
    \frac1{1-\delta} \sum_{V \in \Z_q^{n \times k'}} |\hat{\Den{\calD}}(V)| \cdot \parens*{ \max_{\substack{g : \Z_q^{m \times k} \to \C, \\
    \|g\|_1 = 1, \\
    \|g\|_\infty \le q^s}} \sum_{\substack{U \in \SF, \\ \wt(U+\Pi^\phi_M V) = h - (\wt(V) - \wt(\Pi^\phi_M V))}} |\hat{g}(U)| }
    > \frac1{\delta^h}  \parens*{ \frac{C' \sqrt{s n k'}}{\max\{h,s\}} }^{h/2}. \]
    \paragraph*{Sufficiency of the events.}
    We first claim that whenever $\overline{\calE(M) \cup \bigcup_{h=1}^{\sigma n} \calF_h(M)}$ holds,
    then the desideratum holds, i.e., for every $\calB \subseteq \Z_q^{m \times k}$ such that $|\calB| \ge q^{n-b'}$,
    we have that $\PostShort{\calD}{M}{\calB}$ is $(C',s)$-bounded and $\|\Den{\PostShort{\calD}{M}{\calB}} - 1\|_\infty < q^{s}$.

    We first observe that 
    $\nu_{M,\calB}^{-1} = \frac{\Pr_{\rZ \sim \InputShort{\calD}{M}}[\rZ \in \calB]}{\Pr_{\rZ \in \Z_q^{m \times k}}[\rZ \in \calB]}$,
    which is at least $1-\delta$ by \Cref{obs:boundedness implies uniformity} and the definition of $\overline{\calE(M)}$.
    Hence, $\nu_{M,\calB} \le \frac1{1-\delta}$.
    
    Now, \Cref{lemma:inductive step:posterior form} and \cref{prop:convolution with density} give that 
    $\|\Den{\PostShort{\calD}{M}{\calB}}\|_\infty \le \|\Den{\calD}\|_\infty \cdot \|\Den{\calB}\|_\infty \cdot \nu_{M,\calB}$.
    We have $\|\Den{\calD}\|_\infty \le q^b$ and $\|\Den{\calB}\|_\infty \le q^{b'}$ (by assumption),
    and $\nu_{M,\calB} \le \frac1{1-\delta}$,
    and therefore the product is at most $q^{b+b'+\log_q(1/1-\delta')}\le q^s$, by our assumption.

    We now choose a constant $C'' < \infty$ such that for every $h \in [nk']$, we can guarantee the inequality $\sum_{W \in \Z_q^{n \times k'}, \wt(W) = h} |\hat{\Den{\PostShort{\calD}{M}{\calB}}}(W)| \le \U{C''}{s}{h}{nk'}$.
    We do so using a few cases on $h$:
    \begin{itemize}
        \item If $h \le s$, by \Cref{cor:inductive step:upper bound} and the inequality on $\nu_{M,B}$, since $\overline{\calF_h(M)}$ holds,
        it is the case that \[
        \sum_{W \in \Z_q^{n\times k'}, \wt(W) = h} |\hat{\Den{\PostShort{\calD}{M}{\calB}}}(W)| \le \frac1{\delta^h}\parens*{ \frac{C'\sqrt{snk'}}{h} }^{h/2}
        = \parens*{ \frac{(C'/\sqrt{\delta})\sqrt{snk'}}{h} }^{h/2}. \]
        Since $h \le s$, the RHS is at most $\U{C''}{s}{h}{nk'}$ as long as $C'' \ge C'/\sqrt{\delta}$.
        \item If $s \le h \le \sigma n$, the same argument gives $\sum_{W \in \Z_q^{n\times k'}, \wt(W) = h} |\hat{\Den{\PostShort{\calD}{M}{\calB}}}(W)| \le \parens*{ \frac{(C'/\sqrt{\delta})\sqrt{nk'}}{\sqrt{h}} }^{h/2}$.
        By \Cref{lemma:trivial bound} , we also have that $\sum_{W \in \Z_q^{n\times k'}, \wt(W) = h}  |\hat{\Den{\PostShort{\calD}{M}{\calB}}}(W)| \le (eq^2nk'/h)^{h/2}$.
        (This uses $\|\Den{\PostShort{\calD}{M}{\calB}}\|_\infty \le q^s$ as already proved above.)
        Hence, for $C'' \ge C'/\sqrt{\delta}$ we have, 
\[\sum_{W \in \Z_q^{n\times k'}, \wt(W) = h} |\hat{\Den{\PostShort{\calD}{M}{\calB}}}(W)| \le \U{C''}{s}{h}{nk'}. \]

        \item Finally, if $\sigma n \le h \le nk'$, we still have $\sum_{W \in \Z_q^{n\times k'}, \wt(W) = h} |\hat{\Den{\PostShort{\calD}{M}{\calB}}}(W)| \le (eq^2nk'/h)^{h/2}$.
        To guarantee that $\sum_{W \in \Z_q^{n\times k'}, \wt(W) = h} |\hat{\Den{\PostShort{\calD}{M}{\calB}}}(W)| \le \parens*{ \frac{C''\sqrt{nk'}}{\sqrt{h}} }^{h/2}$,
        it therefore suffices to show that $\frac{C''\sqrt{nk'}}{\sqrt{h}}^{h/2} \ge (eq^2nk'/h)^{h/2}$.
        This is achieved by setting $C'' \ge q^2e/\sqrt{ \sigma}$,
        so that we have $h \ge \sigma nk' \implies C''\sqrt{h} \ge q^2e\sqrt{nk'} \implies C''\sqrt{nk'}/\sqrt{h} \ge q^2enk'/h$.
    \end{itemize}
    Hence, if we set $C'' \coloneqq \max\{C'/\sqrt{\delta},q^2e/\sqrt{\sigma}\}$, we get that in all cases we have
\[    
\sum_{W \in \Z_q^{n\times k'}, \wt(W) = h} |\hat{\Den{\PostShort{\calD}{M}{\calB}}}(W)| \le \U{C''}{s}{h}{nk'}. 
\]

    \paragraph*{Bounding the probability over events.}
    Since $\calD$ is $(C,s)$-bounded, the boundedness implies uniformity lemma (\Cref{lemma:boundedness implies uniformity})
    posits directly that $\Pr[\calE(\rM)] \le \delta$.
    Further, applying Markov's inequality to \Cref{lemma:inductive step:in expectation}, we conclude that $\Pr[\calF_h(\rM)] \le \frac{\delta^h}{1-\delta}$.
    Thus a union bound gives $\Pr[\calF_h(\rM)] \le \frac1{1-\delta} \sum_{h=1}^\infty \delta^h \le 4\delta$.
\end{proof}

\subsection{Combinatorial bound}

Consider the set $\calV \coloneqq [n] \times [k]$ as a set of $k\cdot n$ vertices partitioned into $k$ parts $\calV_\ell \coloneqq [n] \times \{j\}$, each of size $n$.
If $U \subseteq \calV$ is a set of vertices, and $M \in \PHM{m}{k}{n}$, we define the following:
\begin{itemize}
	\item Sets of edges: $K(M,U) \coloneqq \{j \in [m] : |e_j \cap U| = 1\}$ and $D \coloneqq \{j \in [m] : |e_j \cap U| \geq 2\}$, where $e_j \subseteq \calV$ is the $j$-th edge of $M$.
	\item Size parameters: $\kappa(M,U) \coloneqq |K|$, $d(M,U) \coloneqq |D(M,U)|$, and $\eta(M,U) \coloneqq |U \cap \bigcup_{j \in D} e_j|$.
\end{itemize}

We suppress the dependence on $M$ and $U$ when clear from context.
Note that $\eta/k \le d \le \eta/2$ and $d \le m = \alpha n \le n/k$, so $\eta \le kd \le n$.
In particular $d \le n/k < n$.

\begin{proposition}\label{prop:estimate:(x/y)^y}
    For every $x, y > 0$, it holds that $(x/y)^y \le e^x$.
\end{proposition}
\begin{proposition}\label{prop:estimate:linear to exp}
    For every $x \in \R$, it holds that $x \le (e^{1/e})^x$.
\end{proposition}

We now define the following function:
\begin{equation}
    p_{C,\alpha}(n,u,\kappa,\eta) 
    \coloneqq \alpha^\kappa \cdot C^u \cdot \parens*{ \frac{u}n }^{\eta/2}.
\end{equation}
We note that this bound looks drastically simpler than the comparable bound in \cite[Lemma 6.24]{CGS+22-linear-space}.

\begin{lemma}\label{lemma:inductive step:combinatorial bound}
For every $k \in \N$ there exists a constant $C < \infty$ such that the following holds.
Let $U \subseteq \calV = [n] \times [k]$ be a fixed set of $|U|=u$ marked vertices.
For a random matching $\rM \in \PHM{m}{k}{n}$ with $m = \alpha n$ hyperedges, we have: \[
\Pr[\kappa(U,\rM) = \kappa,\ \eta(U,\rM) = \eta] \le p_{C,\alpha}(n,u,\kappa,\eta) \]
assuming $k+\eta \le u$, and $0$ otherwise.
\end{lemma}

\begin{proof}
We follow and adapt (simplify) the proof of \cite[Lemma~6.24]{CGS+22-linear-space}.

\paragraph*{Step 1: Symmetry reduction.}
For $j \in [k]$, let $u_\ell \coloneqq |U \cap \calV_\ell|$ denote the support of $U$ on the $j$-th part.
To calculate $\Pr[\kappa(U,\rM) = \kappa,\ \eta(U,\rM) = \eta]$,
by column-wise permutation symmetry, it is equivalent to instead fix one matching, such as \[
M \coloneqq \begin{bmatrix} 1 & 1 & \cdots & 1 \\ 2 & 2 & \cdots & 2 \\ \vdots & \vdots & \ddots & \vdots \\ m & m & \cdots & m \end{bmatrix}, \]
sample $\rU$ uniformly at random from the set $\calU \coloneqq \{U \subseteq \calV : \forall j \in [k],\ |U \cap \calV_j| = u_\ell\}$,
and calculate $\Pr[\kappa(\rU,M) = \kappa,\ \eta(\rU,M) = \eta]$.

\paragraph*{Step 2: Partitioning by $d$.}
For fixed $d \in [\eta/k,\eta/2]$, let $\calN(d) \coloneqq \{U \in \calU : \kappa(U,M) = \kappa,\ d(U,M) = d,\ \eta(U,M) = \eta\}$.
Then for $\rU \sim \calU$, \[
\Pr[\kappa(\rU,M) = \kappa,\ \eta(\rU,M) = \eta] \le \sum_{d=\eta/k}^{\eta/2} \frac{|\calN(d)|}{\prod_{j=1}^k \binom{n}{u_\ell}}. \]

\paragraph*{Step 3: Counting $|\calN(d)|$.}
We now upper-bound $|\calN(d)|$.

To specify an element in $U \in \calN(d)$, it suffices to specify the edge-sets $K=K(U,M) \subseteq [m]$ and $D=D(U,M) \subseteq [m]$,
and then the vertex-sets $U \cap \bigcup_{j \in K} e_j$, $U \cap \bigcup_{j \in D} e_j$, and $U \cap \bigcup_{j \in [m]\setminus(K \cup D)} e_j$.

Pessimistically, there are $\binom{m}{\kappa}$ choices for $K$ and $\binom{m}{d}$ choices for $D$.
Once these have been specified, there are now $k^\kappa$ choices for $U \cap \bigcup_{j \in K} e_j$
(once $K$ has been fixed, $U \cap \bigcup_{j \in K} e_j$ contains exactly one vertex per edge in $K$ and $|K|=\kappa$)
and (pessimistically) $\binom{kd}{\eta}$ choices for $U \cap \bigcup_{j \in D} e_j$ (since $|U \cap \bigcup_{j \in D} e_j| = \eta$ and $\abs{\bigcup_{j \in D} e_j} = kd$).
Once these have all been fixed, it remains to pick $U \cap \bigcup_{j \in [m]\setminus(K \cup D)} e_j$, and we do so by parts;
i.e., we specify the intersections $U \cap \calV_\ell \cap \bigcup_{j \in [m]\setminus(K \cup D)} e_j$ with each part.
Since we must have $|U \cap V_\ell| = u_\ell$, letting $\mu_\ell \coloneqq |U \cap \calV_\ell \cap \bigcup_{j \in (K \cup D)} e_j|$ denote the number of
already-selected vertices in column $\ell$, we conclude that $|U \cap \calV_\ell \cap \bigcup_{j \in [m]\setminus(K \cup D)} e_j| = u_\ell - \mu_\ell$.
Since $U \cap \calV_\ell \cap \bigcup_{j \in [m]\setminus(K \cup D)} e_j \subseteq \calV_\ell \cap \bigcup_{j \in [m]\setminus(K \cup D)} e_j$ and
$|\calV_\ell \cap \bigcup_{j \in [m]\setminus(K \cup D)} e_j| = n-(\kappa+d)$, we deduce: \[
|\calN(d)| \le \binom{m}{\kappa} \cdot \binom{m}{d} \cdot k^\kappa \cdot \binom{kd}{\eta} \cdot \prod_{\ell=1}^k \binom{n-(\kappa+d)}{u_\ell-\mu_\ell}. \]

\paragraph*{Step 4: Applying estimates and collecting like terms.}
Note that for every $\ell$, we have $\mu_\ell \le \kappa+d$, hence $\binom{n-(\kappa+d)}{u_\ell-\mu_\ell} \le \binom{n-\mu_\ell}{u_\ell-\mu_\ell}$.
We can therefore apply $\frac{\binom{a-c}{b-c}}{\binom{a}{b}} \le \parens*{ \frac{b}{a} }^c$, as well as $\binom{a}{b} \le (ea/b)^b$ to the remaining binomials, to get: \[
\frac{|\calN(d)|}{\prod_{\ell=1}^k \binom{n}{u_\ell}} \le \parens*{ \frac{e\alpha n}{\kappa} }^\kappa \cdot k^\kappa \cdot \parens*{ \frac{e\alpha n}{d} }^d \cdot \parens*{ \frac{ekd}{\eta} }^\eta \cdot \parens*{ \frac{u}{n} }^{\kappa+\eta}. \]
Regrouping terms and using that $\alpha^d \le 1$ since $\alpha \le 1$ and $d \ge 0$ gives: \[
\frac{|\calN(d)|}{\prod_{\ell=1}^k \binom{n}{u_\ell}} \le (ek)^{\kappa+\eta} \cdot \alpha^\kappa \cdot \parens*{ \frac{u}{\kappa} }^\kappa
    \cdot \parens*{ \frac{u}{n \eta} }^\eta \cdot \parens*{ \frac{e n}{d} }^d \cdot d^\eta. \]
Writing $(u/(n\eta))^\eta = (u/(\eta/2))^{\eta/2} \cdot (u/(2\eta))^{\eta/2} \cdot (1/n)^\eta$,
we can now apply \Cref{prop:estimate:(x/y)^y} and use $\kappa+\eta \le u$, giving: \[
\frac{|\calN(d)|}{\prod_{\ell=1}^k \binom{n}{u_\ell}} \le (e^{3u} \cdot k^u) \cdot \alpha^\kappa \cdot \parens*{ \frac{u}{2\eta} }^{\eta/2} \cdot \parens*{ \frac{1}{n} }^\eta \cdot \parens*{ \frac{e n}{d} }^d \cdot d^\eta. \]

\paragraph*{Step 5: Summing over $d$.}
Consider the function $g(t) \coloneqq \ln((e n)^t \cdot t^{\eta-t}) = t \ln(e n) + (\eta - t) \ln t$.
Analytically, we have $\frac{d}{dt} g(t) = \ln(e n) + \eta/t - (\ln t + 1) = \ln(n) - \ln(t) + \eta/t \ge \ln(n) - \ln(t)$.
In particular, $g$ is increasing on the interval $(0,n]$.
This gives, using that $\eta \le u$: \[
	\sum_{d=\eta/k}^{\eta/2} \parens*{ \frac{e n}{d} }^d \cdot d^\eta \le \eta/2 \cdot \parens*{ \frac{e n}{\eta/2} }^{\eta/2} \cdot (\eta/2)^\eta
		\le (u \cdot e^{\eta/2} \cdot 2^{-\eta/2}) \cdot (n \cdot \eta)^{\eta/2}. \]
Combining and regrouping gives: \[
\sum_{d=\eta/k}^{\eta/2} \frac{|\calN(d)|}{\prod_{\ell=1}^k \binom{n}{u_\ell}} \le (e^{3u} \cdot k^u \cdot u \cdot e^{\eta/2} \cdot 2^{-\eta/2}) \cdot \alpha^{\kappa} 
\cdot \parens*{ \frac{u}n }^{\eta/2}. \]
We can now finish by setting $C \coloneqq e^{3+1/e+1/2} \cdot k/\sqrt{2}$, where we used that $u \le (e^{1/e})^u$ and $\eta \le u$.
\end{proof}

\subsection{The inductive step in expectation: Proving \Cref{lemma:inductive step:in expectation}}

We finally turn to proving \Cref{lemma:inductive step:in expectation}.
To do so, we require the following analytic lemma, which follows from the proof of~\cite[Lemma 6.24]{CGS+22-linear-space};
we give a dramatically simplified proof in \Cref{sec:casework} below.

\begin{restatable}{lemma}{casework}
\label{lemma:inductive step:casework}
    For every $q,k \in \N$, there exists $\alpha_0 \in (0,1/k)$ such that the following holds.
    For every $C_1,C_2,C_3,C_4 < \infty$, there exists $\epsilon_0 > 0$ and $C_5 < \infty$ such that
    for every $s,n,h,u,\kappa,\eta \in \N$ with $\kappa \le h \le s \le \epsilon_0 n$, and $\kappa+\eta \le u \le \min\{n,h+\eta\}$,
    we have \[
    \U{C_1}{s}{u}{n}  \cdot C_2^{\kappa+\eta} \cdot p_{C_3,\alpha_0}(n,u,\kappa,\eta)\cdot 
    \parens*{ \frac{C_4 \sqrt{sn}}{h+\eta-u} }^{(h-\eta+u)/2}
    \le \parens*{ \frac{C_5 \sqrt{sn}}h }^{h/2}. \]
\end{restatable}

Finally, we give:

\begin{proof}[Proof of \Cref{lemma:inductive step:in expectation}]
    By \Cref{cor:mass on UH} with $\theta \coloneqq 1.1$ (so that $h < \theta s$), there exists $C_4 < \infty$ such that
    for every injection $\phi : [k] \to [k']$, fixed $M \in \PHM{m}{k}{n}$, and $V \in \Z_q^{n \times k'}$ (equivalently, $\Pi^\phi_M V\in \Z_q^{m \times k}$), and $g : \Z_q^{m \times k}\to \C$ satisfying $\|g\|_1 = 1$ and $\|g\|_\infty \le q^s$: 
    \begin{align*}
        &\sum_{\substack{U \in \SF, \\ \wt(U+\Pi^\phi_M V) = h - (\wt(V) - \wt(\Pi^\phi_M V))}} |\hat{g}(U)| \\&\le q^{k \cdot \wt(V)} \cdot \parens*{ \frac{C_4 \sqrt{s m}}{h-(\wt(V) - \wt(\Pi^\phi_M V))-\kappa(M,\Pi^\phi_M V))} }^{(h - (\wt(V) - \wt(\Pi^\phi_M V))-\kappa(M,\Pi^\phi_M V))/2}\\
  & = q^{k \cdot \wt(V)} \cdot \parens*{ \frac{C_4 \sqrt{s m}}{h-(\wt(V) - \eta(M,\Pi^\phi_M V))} }^{(h - (\wt(V) - \eta(M,\Pi^\phi_M V)))/2} \, ,
    \end{align*}
    
    \[
    \]
    where we used the pessimistic upper bounds  $\rwt(\Pi^\phi_M V) \le \wt(V)$ and the equality 
    $\wt(\Pi^\phi_M V) = \kappa(M,\Pi^\phi_M V) + \eta(M,\Pi^\phi_M V)$.
    We now take a maximum over $g$, sample $\rM$ randomly, and sum over $V$.
    This gives:
    \begin{align*}
    &\sum_{V \in \Z_q^{n \times k'}} |\hat{\Den{\calD}}(V)| \cdot \Exp_\rM \bracks*{ \max_{\substack{g : \Z_q^{m \times k} \to \C, \\ \|g\|_1 = 1, \\ \|g\|_\infty \le q^s}} \sum_{\substack{U \in \SF, \\ \wt(U+\Pi^\phi_M V) = h - (\wt(V) - \wt(\Pi^\phi_M V))}} |\hat{g}(U)| } \\
    \le{}& \sum_{V \in \Z_q^{n \times k'}} |\hat{\Den{\calD}}(V)|
    \cdot q^{k \cdot \wt(V)} \cdot \Exp_\rM \bracks*{ \parens*{ \frac{C_4 \sqrt{s m}}{h-(\wt(V) - \eta(\rM,\Pi^\phi_\rM V))} }^{(h - (\wt(V) - \eta(M,\Pi^\phi_\rM V)))/2} } \\
\intertext{Partitioning the sum over $V$ by $\wt(V)$, and the expectation over $\rM$ by $\kappa(M,\Pi^\phi_M V)$ and $\eta(M,\Pi^\phi_M V)$:}
    \le{}& \sum_{u=0}^{nk'} \sum_{\substack{V \in \Z_q^{n \times k'}, \\ \wt(V) = u}} |\hat{\Den{\calD}}(V)|
    \cdot q^{ku} \cdot
    \sum_{\substack{\kappa,\eta \in \N, \\ \kappa+\eta \le u \le h+\eta}}
    \parens*{ \frac{C_4 \sqrt{s m}}{h-(u - \eta)} }^{(h - (u - \eta))/2}
    \cdot \Pr_\rM [\kappa(\rM,\Pi^\phi_\rM V)=\kappa,\ \eta(\rM,\Pi^\phi_\rM V)=\eta] \\
\intertext{By the $(C_1,s)$-boundedness of $\calD$ and \Cref{lemma:inductive step:combinatorial bound} (letting $C_3$ denote the constant from the latter), and using $m\le nk$:}
    \le{}& \sum_{u=0}^{nk'} 
    \sum_{\substack{\kappa,\eta \in \N, \\ \kappa+\eta \le u \le h+\eta}} \U{C_1}{s}{u}{n} 
    \cdot q^{ku} \cdot \parens*{ \frac{C_4\sqrt{k} \sqrt{s n}}{h-(u - \eta)} }^{(h - (u - \eta))/2}
    \cdot p_{C_3,\alpha}(n,u,\kappa,\eta).
\intertext{Defining $C_2 \coloneqq 2q^k$ and applying the foregoing \Cref{lemma:inductive step:casework}:}
    \le{}& \sum_{u=0}^{nk'} 
    \sum_{\substack{\kappa,\eta \in \N, \\ \kappa+\eta \le u \le h+\eta}} 2^{-u} \parens*{ \frac{C_5 \sqrt{sn}}{h} }^{h/2}
    \le \sum_{u=0}^{nk'} u^2 \cdot 2^{-u} \parens*{ \frac{C_5 \sqrt{sn}}{h} }^{h/2}
    \le \parens*{ \frac{169 C_5 \sqrt{sn}}{h} }^{h/2}
    \end{align*}
    where for the final inequality we used that $\sum_{u=0}^{nk'} (u+1)^2 \cdot 2^{-u} \le \sum_{u=0}^\infty (u+1)^2 \cdot 2^{-u} \le 1 + \int_0^\infty (u+1)^2 \cdot 2^{-u} \;du < 13$ and $h \ge 1$.
\end{proof}

\subsection*{AI Disclosure}
We used Claude to assist with the proofs in \Cref{sec:singleton-free}, though all content in the paper is written by the authors. The authors verified the correctness and originality of all content including references.

\printbibliography

\appendix
\section{Proof of \cref{thm:main DIHP}}\label{proof:thm:main DIHP}
 
 We inductively show that for every $t\in [T]$,

    \[
	\|(\rM_{1:t},\rS^Y_{1:t})-(\rM_{1:t},\rS^N_{1:t})\|_\tvd\leq (t/T) \delta\,  \tag{Induction hypothesis}
	\] 
    where $\cE_0$ is the trivial event that is always true.

	First, we prove the base case $t=1$. Recalling that $\rS_0^Y=\rS_0^N=\bot$, we have
	\begin{align*}
		\|(\rM_1,\rS_1^Y)-(\rM_1,\rS_1^N)\|_{\tvd,\cE_1}&=\|(\rM_1,\rS_1^Y)-(\rM_1,\Player_1(\rM_1,\rS_0^N,\rU))\|_{\tvd,\cE_1}\\
		&=\|(\rM_1,\rS_1^Y)-(\rM_1,\Player_1(\rM_1,\rS_0^Y,\rU))\|_{\tvd,\cE_1} \, ,
	\end{align*}
    where $\rU\sim \Unif(\Z_q^{\alpha n\times k})$ and $\|\cdot\|_{\tvd,\cE}$ denotes the total variation distance, conditioned on event $\cE$.
	Observe that for every fixed $M_1$ and $S^Y_0$ satisfying $\cE_1$, we have
    $\|\rS_1^Y-\Player_1(M_1,S_0^Y,\rU)\|_\tvd\le \delta/(2T)$. It follows from data-processing inequality (\cref{prop:data_processing}) that 
    \[\|(\rM_1,\rS_1^Y)-(\rM_1,\Player_1(\rM_1,\rS_0^Y,\rU))\|_{\tvd,\cE_1} \le \delta/2T .\]
    Therefore,
    \begin{align*}
       \|(\rM_1,\rS_1^Y)-(\rM_1,\Player_1(\rM_1,\rS_0^Y,\rU))\|_\tvd &\le \|(\rM_1,\rS_1^Y)-(\rM_1,\Player_1(\rM_1,\rS_0^Y,\rU))\|_{\tvd,\cE_1} + \Pr[\overline{\cE_1}]\\
       & \le \delta/(2T) + \Pr[\overline{\cE_1}] \le \delta/T \, , 
    \end{align*}
    which completes the base case.

    For every $t=2,\dots,T$, we have 
	\begin{multline*}
		\| (\rM_{1:t},\rS^Y_{1:t})-(\rM_{1:t},\rS^N_{1:t})\|_\tvd \\
		= \|(\rM_{1:t},\rS^Y_{1:t-1},\Player_t(\rM_{1:t},\rS^Y_{t-1},\rZ_t))-(\rM_{1:t},\rS^N_{1:t-1},\Player_t(\rM_{1:t},\rS^N_{t-1},\rU))\|_\tvd\, .
	\end{multline*}
	Let us define $\rQ_{t-1}^Y = (\rM_{1:t-1},\rS^Y_{1:t-1})$  and $\rQ_{t-1}^N = (\rM_{1:t-1},\rS^N_{1:t-1})$. 
	Then, we can rewrite the above expression for total variation distance in terms of the new notation as follows:
	\begin{multline}\label{eqn:tvd_substitution}
		\|(\rM_{1:t},\rS^Y_{1:t-1},\Player_t(\rM_{1:t},\rS^Y_{t-1},\rZ_t))-(\rM_{1:t},S^N_{1:t-1},\Player_t(\rM_{1:t},\rS^N_{t-1},\rU))\|_\tvd \\
		= \|(\rQ_{t-1}^Y,\rM_t,\Player_t(\rQ_{t-1}^Y,\rM_t,\rZ_t))-(\rQ_{t-1}^N,\rM_t,\Player_t(\rQ_{t-1}^N,\rM_t,\rU))\|_\tvd \, .
	\end{multline}
	We now apply \Cref{lem:KKsubstitutionlemma} to \Cref{eqn:tvd_substitution}.
    Applying this lemma with $\rX^1 = \rQ_{t-1}^Y$, $\rX^2 = \rQ_{t-1}^N$, $\rZ^1=(\rM_t,\rZ_t)$, $\rZ^2=(\rM_t,\rU)$, and $f$ as the function that maps the tuple $(X,(B,C))$ to $(B,\Player_t(X,B,C))$, we get
	\begin{align}\label{eqn:conditionalTVD}
		&\|(\rQ_{t-1}^Y,\rM_t,\Player_t(\rQ_{t-1}^Y,\rM_t,\rZ_t))-(\rQ_{t-1}^N,\rM_t,\Player_t(\rQ_{t-1}^N,\rM_t,\rU))\|_\tvd \nonumber \\
		\le\ &\|\rQ_{t-1}^Y- \rQ_{t-1}^N\|_\tvd + \| (\rQ_{t-1}^Y,\rM_t,\Player_t(\rQ_{t-1}^Y,\rM_t,\rZ_t))-(\rQ_{t-1}^Y,\rM_t,\Player_t(\rQ_{t-1}^Y,\rM_t,\rU))\|_\tvd \, .
	\end{align}

	Now, by applying the induction hypothesis, we have that 
    \begin{equation}\label{eqn:induction}
    \|\rQ_{t-1}^Y- \rQ_{t-1}^N\|_\tvd \le (t-1)\delta/T.
	\end{equation}
	Next, we bound the second term on the right hand side of \eqref{eqn:conditionalTVD}, i.e.,
	\[\| (\rQ_{t-1}^Y,\rM_t,\Player_t(\rQ_{t-1}^Y,\rM_t,\rZ_t))-(\rQ_{t-1}^Y,\rM_t,\Player_t(\rQ_{t-1}^Y,\rM_t,\rU))\|_\tvd,\] 
	by applying condition (iii) from \cref{lem:hybrid}. According to this condition, for every \emph{fixed} $(M_{1:t})$ and $S^Y_{1:t-1}$ satisfying $\cE_t$, we have
    \begin{equation*}
		\|\Player_t(M_{1:t},S^Y_{1:t-1},\rZ_t)-\Player_t(M_{1:t},S^Y_{1:t-1},\rU)\|_\tvd\leq \delta/(2T)
	\end{equation*}
	where $U\sim\Unif(\Z_q^{(k-1)\alpha n})$. By \cref{lem:statistical_test}, it follows that 
    \begin{equation}\label{eqn:tvd}\| (\rQ_{t-1}^Y,\rM_t,\Player_t(\rQ_{t-1}^Y,\rM_t,\rZ_t))-(\rQ_{t-1}^Y,\rM_t,\Player_t(\rQ_{t-1}^Y,\rM_t,\rU))\|_{\tvd,\cE_t} \le \delta/(2T)\, .\end{equation}
	Combining \cref{eqn:tvd_substitution,eqn:conditionalTVD,eqn:induction,eqn:tvd}, we have  \[
	\|(\rM_{1:t},\rS^Y_{1:t})-(\rM_{1:t},\rS^N_{1:t})\|_\tvd\le
    (t-1)\delta/T + \delta/(2T) + \Pr[\overline{\calE_t} \mid \calE_{t-1}] \le t\delta/T,
	\]
	which completes the induction.
    
    Thus,
    \[ \|(\rM_{1:T},\rS^Y_{1:T})-(\rM_{1:T},\rS^N_{1:T})\|_\tvd\le \delta\, .\]
    This implies that $\Pi$ cannot have advantage more than $\delta$, 
    which contradicts the assumptions of the theorem statement. Therefore, we conclude that any protocol for DIHP with advantage $\delta$ requires $\tau n$ bits of communication, as desired.

\section{Proof of \Cref{lemma:inductive step:casework}}\label{sec:casework}

\newcommand{\CLHS}{C_{\mathrm{LHS}}}
\newcommand{\CRHS}{C_{\mathrm{RHS}}}
\newcommand{\CW}{C_{\mathrm{W}}}

In this appendix, we prove \Cref{lemma:inductive step:casework}, restated as follows:

\casework*

We recall the definitions:
\begin{align}
    \U{C_1}{s}{u}{n} &= \begin{cases}
        \parens*{ \frac{C_1 \sqrt{sn}}{u}}^{u/2}  & \text{if }0 \le u \le s, \\
        \parens*{ \min\braces*{
            \frac{C_1\sqrt{n}}{\sqrt{u}},
            \frac{eq^2 n}{u}
        }}^{u/2} &\text{if }u>s,
    \end{cases}, \label{eq:casework:U} \\
    \U{C_2}{s}{h}{n} &=
 \parens*{ \frac{C_2 \sqrt{sn}}{h}}^{h/2}, \label{eq:casework:W} \\
    p_{C,\alpha}(n,u,\kappa,\eta) 
    &= C^u \cdot \alpha^{\kappa} \cdot \parens*{ \frac{u}n }^{\eta/2}, \label{eq:casework:p}
\end{align}
where for \Cref{eq:casework:W} we used the assumption that $h \le s$.

Our first step is showing that the following lemma implies \Cref{lemma:inductive step:casework}:

\begin{lemma}\label{lemma:induction step:bound reduced}
    There exists $\alpha_0 \in (0,1/k)$ such that
    for every $C_1,\CLHS < \infty$,
    there exist $\epsilon_0 > 0$ and $\CRHS < \infty$ such that the following holds.
    
    For every $\alpha \in (0,\alpha_0)$,
    and $s,n,u,h,\kappa,\eta \in \N$ with $h \le s \le \epsilon_0 n$ and
    $\kappa+\eta \le u \le \min\{n,h+\eta\}$, \[
        \U{C_1}{s}{u}{n} \cdot \CLHS^u \cdot \alpha^{\eta} \cdot \frac{  u^{\eta/2} \cdot h^{u/2-\eta/2} }{n^{u/4+\eta} \cdot s^{u/4-\eta} } \le \CRHS^h, \]
    where $\U{C_1}{s}{u}{n}$ is defined as in \Cref{eq:casework:p} above.
\end{lemma}

\begin{proof}[Proof of \cref{lemma:inductive step:casework} using \Cref{lemma:induction step:bound reduced}]
    Let $h' \coloneqq h + \eta - u$.
    We invoke \Cref{lemma:induction step:bound reduced} with $\CLHS \coloneqq C_2C_3$.

    To prove \Cref{lemma:inductive step:casework}, it suffices to pick $C_5 < \infty$ such that: \[
        \U{C_1}{s}{u}{n} \cdot C_2^u\cdot
        \parens*{ \frac{C_3^u \cdot \alpha^{\kappa} \cdot u^{\eta/2}}{n^{\eta/2}} } \cdot \parens*{ \frac{C_4^{h'/2} \cdot s^{h'/4} \cdot n^{h'/4} }{(h')^{h'/2}} }
        \le \parens*{ \frac{C_5^{h/2} \cdot s^{h/4} \cdot n^{h/4} }{h^{h/2}} }. \]
    Simplifying and moving factors between sides gives the equivalent desideratum: \[
    \U{C_1}{s}{u}n \cdot \CLHS^u \cdot \alpha^{\kappa} \cdot \frac{ u^{\eta/2}  }{n^{\eta/2+(h-h')/4} \cdot s^{(h-h')/4} } \cdot \frac{h^{h/2}}{(h')^{h'/2}}
    \le \frac{C_5^{h/2}}{C_4^{h'/2}}. \] 
    
    Now, we apply several simple inequalities.
    We observe $h^{h/2}/(h')^{h'/2} = (h/(h'/2))^{h'/2} \cdot 2^{-h'/2} \cdot h^{(h-h')/2} \le (e/\sqrt{2})^{h'} \cdot h^{(h-h')/2}$ by \Cref{prop:estimate:(x/y)^y}.
    This, together with the fact that $h' \le h$ and (WLOG) $C_4 > 1$, means that the following inequality implies our desideratum: \[
    \U{C_1}{s}{u}n \cdot \CLHS^u \cdot \alpha^{\kappa} \cdot \frac{ u^{\eta/2}  }{n^{\eta/2+(h-h')/4} \cdot s^{(h-h')/4} } \cdot h^{(h-h')/2}
    \le \frac{C_5^{h/2}}{C_4^{h/2}} \cdot (2/\sqrt{e})^{h/2}. \]
    Substituting in $h-h' = u-\eta$, this is exactly the hypothesis of \Cref{lemma:induction step:bound reduced}. 
    Letting $\CRHS$ denote the resulting constant from \Cref{lemma:induction step:bound reduced},
    we can pick $C_5 \coloneqq \CRHS^2 \cdot C_4 \cdot 2/\sqrt{e}$.
\end{proof}

\begin{proof}[Proof of \Cref{lemma:induction step:bound reduced}]
Let $\epsilon \coloneqq s/n$ (so that $s = \epsilon n$).

We split into five cases based on $u$:
(1a) $1 \le u \le h$,
(1b) $h < u \le s$,
(2a) $s < u \le 16s$,
(2b) $16s < u \le \sqrt{\epsilon} n$,
(3) $\sqrt{\epsilon} n < u \le n$.
Let \[
Q \coloneqq \U{C_1}{s}{u}{n} \cdot \CLHS^u \cdot \alpha^{\eta/2} \cdot \frac{  u^{\eta/2} \cdot h^{u/2-\eta/2} }{n^{u/4+\eta/4} \cdot s^{u/4-\eta/4} } \]
denote the quantity we want to bound.
We determine $\CRHS$ later.

\paragraph{Case 1: $1 \le u \le s$.}
In this regime, \[
\U{C_1}{s}{n}u = C_1^{u/2} \cdot \frac{s^{u/4} \cdot n^{u/4}}{u^{u/2}}. \]
Using the trivial bound $\alpha \le 1$, we get:
\begin{equation}\label{eq:case 1}
    Q \le C_1^{u/2} \cdot \frac{s^{u/4} \cdot n^{u/4}}{u^{u/2}} \cdot \CLHS^u \cdot \frac{u^{\eta/2} \cdot h^{u/2-\eta/2} }{n^{u/4+\eta/4} \cdot s^{u/4-\eta/4} }
    = \underbrace{(C_1 \CLHS)^u \cdot \frac{h^{u/2-\eta/2}}{u^{u/2-\eta/2}} \cdot \frac{s^{\eta/4}}{n^{\eta/4}}}_{\eqqcolon Q_1}.
\end{equation}
We continue bounding $Q_1$ in subcases.

\paragraph{Subcase 1a: $1 \le u \le h$.}
We have: \[
Q_1 \le (C_1\CLHS)^u \cdot \frac{h^{u/2-\eta/2}}{u^{u/2-\eta/2}} \le (C_1\CLHS)^h \cdot \frac{h^u}{u^u} \le (eC_1\CLHS)^h. \]
where the first equality uses the trivial bound $s \le n$, the second uses the assumption $h \ge u$ (and $\eta/2 \ge 0$),
and the third is \Cref{prop:estimate:(x/y)^y}.
This suffices as long as $\CRHS \ge eC_1\CLHS$.

\paragraph{Case 1b: $h < u \le s$.} We now have:
\begin{multline*}
    Q_1 = (C_1\CLHS)^u \cdot \frac{h^{u/2-\eta/2}}{u^{u/2-\eta/2}} \cdot \epsilon^{\eta/4}
    \le (C_1\CLHS)^u \cdot \frac{h^{u/2-\eta/2}}{u^{u/2-\eta/2}} \cdot \epsilon_0^{\eta/4}
    \le (C_1\CLHS)^u \cdot \frac{h^{u/2-\eta/2}}{u^{u/2-\eta/2}} \cdot \epsilon_0^{(u-h)/4} \\
    = (C_1\CLHS \epsilon_0^{1/4})^u \cdot (1/\epsilon_0^{1/4})^h \cdot \frac{h^{u/2-\eta/2}}{u^{u/2-\eta/2}}
    \le (C_1\CLHS \epsilon_0^{1/4})^u \cdot (1/\epsilon_0^{1/4})^h,
\end{multline*}
using, respectively, the definition of $s$,
the assumed bound on $\epsilon_0$,
the assumption that $u-h \le 2\eta/2$,
rearranging,
and finally, that $h \le u$ and $u \ge 2\eta/2$.
This suffices as long as $\epsilon_0 \le 1/(C_1\CLHS)^4$
and $\CRHS \ge 1/\epsilon_0^{1/4}$.
(We note that the inequality on $Q_1$ we showed here did not need $u \le s$, only that $h < u$;
we will use this fact in a subsequent subcase.)

\paragraph{Case 2: $s < u \le \sqrt{\epsilon} n$.}
In this regime, we use: \[
    \U{C_1}{s}{n}u \le \frac{C_1^{u/2} \cdot u^{u/4} \cdot n^{u/4}}{u^{u/2}}
    = \frac{C_1^{u/2} \cdot n^{u/4}}{u^{u/4}} \]
This gives \[
    Q \le C_1^{u/2} \cdot \frac{n^{u/4}}{u^{u/4}} \cdot \CLHS^u \cdot \frac{ u^{\eta/2} \cdot h^{u/2-\eta/2} }{n^{u/4+\eta/4} \cdot s^{u/4-\eta/4} }
    = \underbrace{(\sqrt{C_1} \CLHS)^u \cdot \frac{ u^{u/4} \cdot s^{\eta/4}}{n^{\eta/4} \cdot s^{u/4}} \cdot \frac{h^{u/2-\eta/2}}{u^{u/2-\eta/2}}}_{\eqqcolon Q_2}, \]
and we proceed to bound $Q_2$.
Note that we artificially split $\eta/2-u/4 = u/4 - (u/2 - \eta/2)$
(this will be useful for subsequent calculations).

\paragraph{Case 2a: $s < u \le 16s$.}
We have \[
    Q_2
    \le (\sqrt{C_1} \CLHS)^u \cdot \frac{ (16s)^{u/4} \cdot s^{\eta/4}}{n^{\eta/4} \cdot s^{u/4}} \cdot \frac{h^{u/2-\eta/2}}{u^{u/2-\eta/2}} \\
    = (2\sqrt{C_1} \CLHS)^u \cdot \frac{ s^{\eta/4}}{n^{\eta/4} } \cdot \frac{h^{u/2-\eta/2}}{u^{u/2-\eta/2}}
    = 2^u Q_1, \]
where the equality uses the assumption $u \le 16s$.
In the final RHS, $Q_1$ is the quantity considered in case (1b),
wherein we showed --- without assuming $u < s$ --- that $Q_1 \le (C_1\CLHS \epsilon_0^{1/4})^u \cdot (1/\epsilon_0^{1/4})^h$.
Hence $2^u Q_1 \le (2C_1\CLHS \epsilon_0^{1/4})^u \cdot (1/\epsilon_0^{1/4})^h$,
which is sufficiently bounded if $\epsilon_0 \le 1/(2C_1 \CLHS)^4$.

\paragraph{Case 2b: $16s \le u < \sqrt{\epsilon} n$.}
In this case, \[
    Q_2 \le (\sqrt{C_1} \CLHS)^u \cdot \frac{ (\sqrt{\epsilon} n)^{u/4} \cdot (\epsilon n)^{\eta/4}}{n^{\eta/4} \cdot (\epsilon n)^{u/4}} \cdot \frac{h^{u/2-\eta/2}}{u^{u/2-\eta/2}}
    \le (\sqrt{C_1} \CLHS)^u \cdot \epsilon^{\eta/4-u/8}. \]
Here, the first inequality uses the assumption on $u$ (and rewriting the definition of $s$),
and the second uses that $h \le u$ (and $u \ge 2\eta/2$, and consolidates powers of $\epsilon$).
Now, $\eta/4 - u/8 \ge (u-h)/4 - u/8 = u/8 - h/4 \ge 2s - h/4 \ge 2h - h/4 \ge 0$,
and therefore $\epsilon^{\eta/4-u/8} \le \epsilon^{u/8-h/4} \le \epsilon_0^{u/8-h/4}$, and so \[
Q_2 \le (\sqrt{C_1} \CLHS \epsilon_0^{1/8})^u \cdot (1/\epsilon_0^{1/4})^{h}. \]
Thus, we get the required bound as long as $\epsilon_0 \le 1/(\sqrt{C_1} \CLHS)^8$
and $\CRHS \ge 1/\epsilon_0^{1/4}$.

\paragraph{Case 3: $\sqrt{\epsilon} n \le u \le n$.}
Here, we use the final bound \[
\U{C_1}{s}{n}u \le \parens*{ \frac{8e n}{u} }^{u/2} = \CW^{u/2} \cdot \frac{n^{u/2}}{u^{u/2}}. \]
Thus, we have: \[
    Q \le \CW^{u/2} \cdot \frac{n^{u/2}}{u^{u/2}} \cdot \CLHS^u \cdot \alpha^{b+\eta/2} \cdot \frac{u^{\eta/2} \cdot h^{u/2-\eta/2} }{n^{u/4+\eta/4} \cdot s^{u/4-\eta/4} }
    \le (\sqrt{\CW} \CLHS)^u \cdot \alpha_0^{(u-h)/2} \cdot \underbrace{\frac{n^{u/4-\eta/4}}{ s^{u/4-\eta/4} } \cdot \frac{h^{u/2-\eta/2} }{u^{u/2-\eta/2} }}_{\eqqcolon R_3}, \]
where the equality uses $\alpha \le \alpha_0 \le 1$ and $b \ge 0$, $\eta/2 \ge (u-h)/2 \ge 0$.
We will verify that $R_3 \le 1$.
Assuming this, we get \[
Q \le (\sqrt{\CW} \CLHS)^u \cdot \alpha_0^{(u-h)/2} = (\sqrt{\CW} \CLHS \alpha_0^{1/2})^u \cdot (1/\alpha_0^{1/2})^h, \]
which is sufficiently small as long as $\alpha_0 \le 1/(\sqrt{\CW} \CLHS)^2$
and $\CRHS \ge 1/\alpha_0^{1/2}$.

Finally, we bound \[
    R_3 = \frac{n^{u/4-\eta/4}}{ (\epsilon n)^{u/4-\eta/4} } \cdot \frac{h^{u/2-\eta/2} }{u^{u/2-\eta/2} }
    = \parens*{ \frac{h}{\sqrt{\epsilon} u} }^{u/2-\eta/2}, \]
and we use $h \le s = \epsilon n$ and $u \ge\sqrt{\epsilon} n$ to get $(h/(\sqrt{\epsilon} u)) \le 1$.
\end{proof}

\end{document}